\newcommand{\mytitle}{Smooth Routing in Decaying Trees}

\def\thxPACSs{Funded by Deutsche Forschungsgemeinschaft (DFG, German
	Research Foundation), project PACS (FL~1247/1-1, 522475669).}
\def\HUaffil{Humboldt-Universität zu Berlin,
	Department of Computer Science, Algorithm Engineering Group, Germany}
\def\TUaffil{Technische Universität Berlin, Algorithmics and Computational Complexity, Germany}

\documentclass[pagebackref]{article}
\usepackage[
	margin=3cm,
	a4paper,
]{geometry}

\usepackage{pdflscape}

\usepackage{authblk}

\title{\Large\bf \mytitle}
\author[2]{Till Fluschnik\footnote{\thxPACSs}}
\author[1]{Amela Pucic}
\author[1]{Malte Renken}

\affil[1]{\TUaffil}
\affil[2]{\HUaffil}

\date{}
\usepackage{amsthm}
\usepackage{thmtools}

\usepackage{url}
\usepackage[x11names, table]{xcolor}
\usepackage[utf8]{inputenc}

\usepackage{graphicx}
\usepackage{amsmath}
\usepackage{booktabs}
\usepackage{tabularray}
\UseTblrLibrary{booktabs}

\usepackage[T1]{fontenc}
\usepackage{amssymb}
\usepackage[outline]{contour}
\usepackage{enumerate}
\usepackage{dsfont}
\usepackage{tabularx}
\usepackage[ruled,vlined,linesnumbered]{algorithm2e}

\usepackage[square,numbers]{natbib}

\usepackage[linkcolor=red!50!black,citecolor=green!40!black,colorlinks]{hyperref}
\usepackage[capitalize,nameinlink]{cleveref}

\usepackage{doi}
\usepackage{paralist}
\usepackage{tikz}
\usetikzlibrary{calc,positioning}
\usetikzlibrary{arrows.meta,positioning}
\usepackage{pgfplots}
\pgfplotsset{compat=1.5}
\pgfplotsset{major grid style={very thin,gray!20!white}} %
\usepackage{pgfplotstable}
\usepackage{mathtools}

\usepackage{multirow}

\usepackage{pifont}%

\usepackage{etoolbox}

\usepackage{xargs}
\newcommandx{\set}[2][1=1]{\ensuremath{\{#1,\ldots,#2\}}}
\newcommandx{\tlog}[3][1=,3=]{\log_{#1}^{#3}(#2)}

\newcommandx{\mydefenv}[5][2=A,4=A]{%
  \ifstrequal{#2}{A}{\newtheorem{#1}{#3}}{}
  \ifstrequal{#4}{A}{\crefname{#1}{#3}{#3s}}{\crefname{#1}{#3}{#4}}%
  \Crefname{#1}{#5{.}}{#5s{.}}
}

\newcommandx{\mydefenvn}[5][2=A,4=A]{%
  \ifstrequal{#2}{A}{\newtheorem*{#1}{#3}}{}
  \ifstrequal{#4}{A}{\crefname{#1}{#3}{#3s}}{\crefname{#1}{#3}{#4}}%
  \Crefname{#1}{#5{.}}{#5s{.}}
}

\mydefenv{theorem}[A]{Theorem}{Thm}
\mydefenv{lemma}[A]{Lemma}{Lem}
\mydefenv{proposition}[A]{Proposition}{Prop}
\mydefenv{conjecture}[A]{Conjecture}{Conj}
\mydefenv{corollary}[A]{Corollary}[Corollaries]{Cor}
\mydefenv{observation}[A]{Observation}{Obs}
\mydefenv{fact}{Fact}{Fact}
\theoremstyle{definition}
\mydefenvn{problem}{Problem}{Prob}
\mydefenv{definition}[B]{Definition}{Def}
\mydefenv{rrule}{Data Reduction Rule}{DRR}
\mydefenv{xrule}{Rule}{Rule}
\mydefenv{myalgo}{Algorithm}{Algo}
\theoremstyle{remark}
\mydefenv{example}[B]{Example}{Ex}
\mydefenv{remark}[A]{Remark}{Rem}
\mydefenv{idea}{Idea}{Idea}

\declaretheorem[style=definition,name=Construction,qed=$\diamond$]{construction}
\declaretheorem[style=definition,name=Construction,sibling=construction]{construction*}
\crefname{construction}{Construction}{Constructions}
\crefname{construction*}{Construction}{Constructions}

\newcommandx{\decprob}[6][3=Input,5=Question]{
\begin{problem}[{#1}]\label{prob:#2}
  \textbf{Given} #4,
  the \textbf{question} is whether #6
  \end{problem}
}

\newcommandx{\decprobY}[6][3=Input,5=Question]{\begin{samepage}
  \begingroup
\begin{problem}\label{prob:#2}%
  {{\textsc{#1}}}
  \nopagebreak[4]\end{problem}\nopagebreak[4]\vspace{-0.5em}
  \par\noindent\hangindent=\parindent\textbf{#3}:  #4\nopagebreak[4]
  \par\noindent\hangindent=\parindent\textbf{#5}:  #6
  \par\medskip
  \endgroup
  \end{samepage}
}

\newcommandx{\decprobX}[5][2=Input,4=Question]{%
  \begingroup
  \par\medskip
  \noindent \colorbox{gray!17!white}{\textsc{#1}\index{problem!#1}}\nopagebreak[4]
  \par\noindent\hangindent=\parindent\textbf{#2}:  #3\nopagebreak[4]
  \par\noindent\hangindent=\parindent\textbf{#4}:  #5
  \par  \medskip
  \endgroup
}

\newcommand{\calF}{\mathcal{F}}
\newcommand{\calG}{\mathcal{G}}

\newcommand{\calU}{\mathcal{U}}

\DeclareMathOperator{\dist}{dist}

\newcommand{\yes}{\textnormal{\texttt{yes}}}
\newcommand{\no}{\textnormal{\texttt{no}}}
\newcommand{\RD}{$(\Rightarrow)\quad$}
\newcommand{\LD}{$(\Leftarrow)\quad$}

\newcommand{\N}{\mathbb{N}}
\newcommand{\Nzero}{\mathbb{N}_0}

\newcommand{\prob}[1]{\textnormal{\textsc{#1}}}

\newcommand{\misuigTsc}{\prob{Multicolored Independent Set on Unit Interval Graphs}}
\newcommand{\misuigAcr}{\prob{MIS\nobreakdash-UIG}}

\newcommand{\cocl}[1]{\ensuremath{\operatorname{#1}}}

\newcommand{\NP}{\cocl{NP}}

\newcommandx{\tref}[2][1=]{{\scriptsize{(\Cref{#2}#1)}}}

\newcommand{\true}{\ensuremath{\top}}
\newcommand{\false}{\ensuremath{\bot}}

\newcommand{\Pset}{\ensuremath{\mathcal{P}}}

\newcommand{\dl}{d}

\newcommand{\truthass}{\ensuremath{\beta}}

\newcommand{\ora}[1]{\overrightarrow{#1}}
\newcommand{\ola}[1]{\overleftarrow{#1}}

\DeclareMathOperator*{\bigland}{\bigwedge}

\newcommand{\ceq}{\ensuremath{\coloneqq}}
\newcommand{\cif}{\text{if }}

\newcommand{\etal}{et~al.}%

\definecolor{lilla}{HTML}{750787}

\newcommand{\thecolor}{black}%

\usepackage{siunitx}
\newcommand{\vkmh}[1]{\SI{#1}{\kilo\meter\per\hour}}

\usepackage{makecell}

\usepackage{etoolbox}
\newcommand{\ExternalLink}{%
\tikz[x=1.2ex, y=1.2ex, baseline=-0.05ex]{%
    \begin{scope}[x=1ex, y=1ex]
        \clip (-0.1,-0.1) --++ (-0, 1.2) --++ (0.6, 0) --++ (0, -0.6) --++ (0.6, 0) --++ (0, -1);
        \path[draw, line width = 0.5, rounded corners=0.5] (0,0) rectangle (1,1);
    \end{scope}
    \path[draw, line width = 0.5] (0.5, 0.5) -- (1, 1);
    \path[draw, line width = 0.5] (0.6, 1) -- (1, 1) -- (1, 0.6);
    }
}

\newcommand{\DG}{\ensuremath{\calG}}
\newcommand{\exogenous}{exogenous}
\newcommand{\simevacTsc}{\prob{Smooth Routing in Decaying Graphs}}
\newcommand{\simevacAcr}{\prob{SRDG}}

\DeclareMathOperator{\trt}{\theta}

\newcommand{\TER}{\pi}

\newcommand{\Rel}{X}

\newcommand{\partref}[2]{\cref{#1}\eqref{#1:#2}}
\newcommand{\src}{s}%
\newcommand{\snk}{t}%
\newcommand{\trtd}[2]{\ensuremath{[#1\vert#2]}}

\newcommand{\bigM}{\texttt{bigM}}
\newcommand{\indc}{\texttt{indic}}
\newcommand{\paths}{\texttt{paths}}
\newcommand{\stars}{\texttt{stars}}
\newcommand{\osm}{\texttt{osm}}

\newcommand{\ddlb}{d_{\rm lb}}
\newcommand{\frc}[1]{{#1}^*}

\newcommand{\tikzpreamble}{%
  \def\teps{0.25}
  \tikzstyle{xnode}=[circle,fill,scale=0.5,draw]
  \tikzstyle{xcnode}=[circle,fill=blue!10,inner sep=1pt,draw,font=\scriptsize]
  \tikzstyle{xedge}=[thick,-,color=\thecolor]
  \tikzstyle{xxedge}=[ultra thick,-,color=\thecolor]
  \tikzstyle{xedgedot}=[thick,-,dotted,color=white]
  \tikzstyle{xhabA}=[-,opacity=0.125, line width=9pt, line cap=round,color=magenta]
  \tikzstyle{xhabB}=[-,opacity=0.125, line width=9pt, line cap=round,color=green]
  \tikzstyle{xhabC}=[-,opacity=0.125, line width=9pt, line cap=round,color=cyan]
  \tikzstyle{xhabD}=[-,opacity=0.125, line width=9pt, line cap=round,color=blue]
  \tikzstyle{xhabE}=[-,opacity=0.125, line width=9pt, line cap=round,color=yellow]

	\tikzstyle{xemark}=[midway,inner sep=1pt,fill=white,opacity=0.8,font=\scriptsize,sloped]
	\tikzstyle{xemarkz}=[midway,inner sep=1pt,fill=white,font=\scriptsize,sloped]
	\tikzstyle{xemarkb}=[midway,below,inner sep=1pt,fill=white,opacity=0.4,font=\scriptsize,sloped]
	\tikzstyle{xemarka}=[midway,above,inner sep=1pt,fill=white,opacity=0.4,font=\scriptsize,sloped]

  \tikzstyle{xpath}=[->,>=latex,line width=0.35em,opacity=0.25]
  \tikzstyle{xtypeA}=[color=green!66!black]
  \tikzstyle{xtypeB}=[color=blue]
  \tikzstyle{xtypeC}=[color=orange!80!black]
  \tikzstyle{xpathA}=[xpath,xtypeA]
  \tikzstyle{xpathB}=[xpath,xtypeB]
  \tikzstyle{xpathC}=[xpath,xtypeC]
  \tikzstyle{xpathD}=[xpath,color=magenta]
  \tikzstyle{xpathE}=[xpath,color=red]
  \tikzstyle{xpathF}=[xpath,color=brown]
}

\newcommandx{\tikzES}[2][1={-,thick}]{%
	\foreach \x/\y in {#2}{\draw[#1] (\x) to (\y);}
}
\newcommandx{\tikzESd}[2][1=\scriptsize]{%
	\foreach \x/\y/\z in {#2}{\draw[-,thick] (\x) to node[midway,inner sep=1pt,fill=white,sloped,font=#1]{\z}(\y);}
}

\newcommand{\myabstract}{%
Motivated by evacuation scenarios arising in extreme events such as flooding or forest fires,
we study the problem of smoothly scheduling a set of paths in
graphs
where connections become impassable at some point in time.
A schedule is smooth if no two paths meet on an edge
and the number of paths simultaneously located at a vertex
does not exceed its given capacity.
We study the computational complexity of the problem when the underlying graph is a tree,
in particular a star or a path.
We prove that already in these settings,
the problem is NP-hard even with further restrictions on the capacities or on
the time when all connections ceased.
We provide an integer linear program (ILP) to compute the latest possible time to evacuate.
Using the ILP and its relaxation,
we solve sets of artificial
(where each underlying graph forms either a path or star)
and semi-artificial instances
(where the graphs are obtained from German cities along rivers),
study the runtimes,
and compare the results of the ILP with those of its relaxation.
}

\begin{document}

\maketitle

\begin{abstract}
	\myabstract{}
\end{abstract}

\section{Introduction}

In evacuation scenarios
facing extreme events like flooding or forest fires,
the underlying routing network may suffer under ceasing connections.
Evacuations plans may include \emph{fixed} evacuation routes,
i.e.,
each route is fully described as the sequence of all its vertices.
However,
since the way how extreme events impact the network
can differ,
a fixed evacuation schedule,
i.e.,
an assignment of departure times for each evacuation route on each of its connections,
could fail.
Given a set of \emph{fixed} evacuation routes in a network with ceasing connections and vertex capacities,
we wonder what is the computational complexity of finding a \emph{smooth} feasible evacuation schedule.
Herein,
we call an evacuation schedule smooth when,
on the one hand,
at each vertex the number of routes that simultaneously gather is at most the capacity of that vertex,
and,
on the other hand,
every two routes start to traverse a connection
at different times
and
traverse no two-way connection in opposite directions such that they meet at some time on it.
We are particularly interested in the special case
when the underlying network
forms a
tree.
Herein,
as subclasses of trees,
we focus on
underlying networks
that form
paths (like a main street)
or stars (like larger crossings or simplified roundabouts).

\newcommand{\DGtuple}{(V,E,A,c,\trt,\dl,\tau)}
\paragraph*{Our model and central decision problem.}
Our model uses \emph{decaying} graphs.
A decaying graph~$\DG=\DGtuple$ consist of
a static (mixed) graph~$G=(V,E,A)$
with edge set~$E\subseteq \binom{V}{2}$ and arc set~$A\subseteq V\times V$,
vertex capacities~$c\colon V\to\N$,
and lifetime~$\tau\in\N$,
and each connection
(i.e., edge or arc)
is equipped with a traversal time~$\trt\colon E\cup A\to\set[0]{\tau-1}$
and a deadline~$\dl\colon E\cup A\to \set{\tau} \eqqcolon T$.
Undirected edges model parts
which can be traversed in both directions but not simultaneously
(e.g., narrow streets).
For every static (mixed) graph~$G=(V,E,A)$
we assume that every two adjacent vertices are connected either by an edge,
an arc,
or two antiparallel arcs.

A path~$P$ can be represented as a sequence~$P=(v_1,v_2,\dots,v_{k+1})$ of at least two vertices all of which are mutually distinct or as a sequence~$P=(e_1,\dots,e_k)$ of connections,
where $e_i=\{v_i,v_{i+1}\}$
or $e_i=(v_i,v_{i+1})$.
Let~$V(P)$ denote the set of vertices of~$P$ and~$\Rel(P)=\{(v_i,v_{i+1})\mid i\in\set{k}\}$ be the set of all pairs of consecutive vertices in~$P$.
Let $\src(P)=v_1$ denote the source
and~$\snk(P)=v_{k+1}$ the sink of~$P$.
For a vertex~$v\in V$ and set~$\Pset$ of paths over~$V$,
let~$\Pset(v)$ denote the set of all paths in~$\Pset$ containing~$v$.

Altogether,
an instance~$I=(\DG,\Pset)$ of our problem consists of
a decaying graph~$\DG=\DGtuple$
and
a set~$\Pset=\{P_1,\dots,P_p\}$ of paths in~$G=(V,E,A)$
(see \Cref{fig:example}(left) for an example).
\begin{figure*}[t!]
 \centering
 \begin{tikzpicture}
  \def\xr{2.4}
  \def\yr{0.275}
  \tikzpreamble{}
  \node (v1) at (0,0)[xcnode,label=-90:{$v_1$}]{$1$};
  \node (v2) at (1*\xr,0)[xcnode,label=-90:{$v_2$}]{$2$};
  \node (v3) at (2*\xr,0)[xcnode,label=-90:{$v_3$}]{$3$};
  \node (v4) at (3*\xr,0)[xcnode,label=-90:{$v_4$}]{$2$};
  \node (v5) at (4*\xr,0)[xcnode,label=-90:{$v_5$}]{$2$};
  \node (v6) at (5*\xr,0)[xcnode,label=-90:{$v_6$}]{$1$};

  \foreach \x/\y/\out/\in/\trt in {
			v1/v2/0/180/{$\trtd{1}{9}$},
			v2/v3/30/150/{$\trtd{2}{9}$},
			v3/v2/-150/-30/{$\trtd{1}{8}$},
			v4/v5/30/150/{$\trtd{1}{11}$},
			v5/v4/-150/-30/{$\trtd{3}{6}$},
			v5/v6/0/180/{$\trtd{1}{7}$}
	}{
			\draw[->, >=latex, thick] (\x) to[out=\out,in=\in] node[xemarkz]{\trt} (\y);
	}
	\draw[-, thick] (v3) to node[xemarkz]{$\trtd{2}{14}$} (v4);

	\foreach \x in {1,...,10}{\node (P\x) at (0,1.0*\yr+\x*\yr)[]{};}

	\foreach \no/\x/\y/\col in {
		1/v1/v5/blue!80!black,
		2/v1/v4/orange,
		3/v3/v6/green!50!black,
		4/v5/v2/red!80!black,
		5/v1/v5/lilla,
		6/v2/v6/brown,
		7/v5/v3/magenta,
		8/v3/v2/gray!80!black,
		9/v4/v6/yellow!80!green!80!black,
		10/v1/v4/cyan%
  }{
		\draw[xpath,color=\col] (\x|-P\no) to node[xemark]{$P_{\no}$}(\y|-P\no);
  }
 \end{tikzpicture}
 \hfill
 \includegraphics[width=0.95\textwidth]{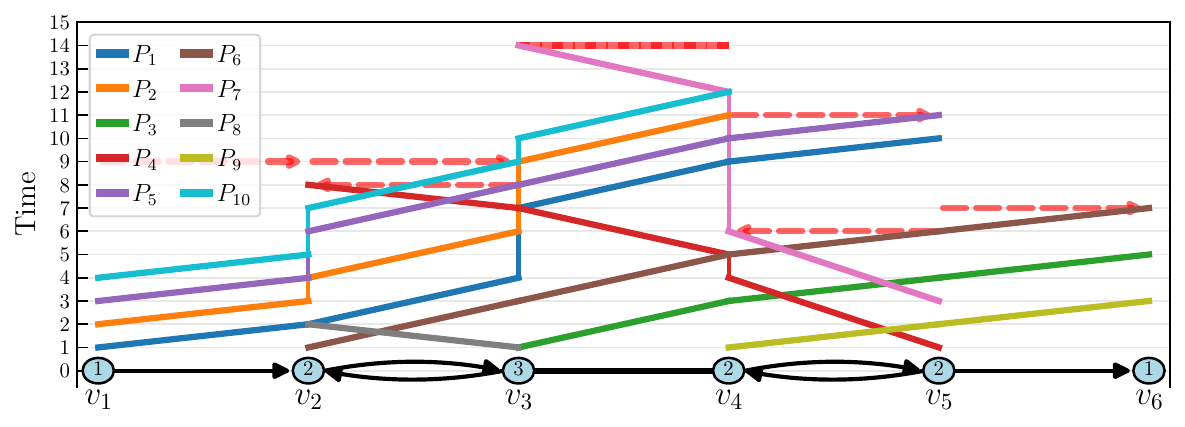}
 \caption{Example instance of~\simevacAcr{} on a decaying path with 6 vertices (top)
 and a solution witnessing feasibility (bottom).
 The capacity of each vertex~$v_i$ is encircled.
 (Top)
 On each connection~$e$ we indicate its traversal time and deadline as~$\trtd{\trt(e)}{d(e)}$.
 Each of the paths~$P_1,\dots,P_{10}$ is described by the vertically aligned source and sink and by an arc for its direction.
 (Bottom)
 For each path its timely location on any vertex or connections is drawn.
 E.g.,
 the path~$P_1$ (blue) starts at time step 1 from its source~$v_1$,
 arrives at and departs from~$v_2$ at time step 2,
 and so on.
 Red dashed lines/arrows indicate the deadline on the respective connection.
 }
 \label{fig:example}
\end{figure*}
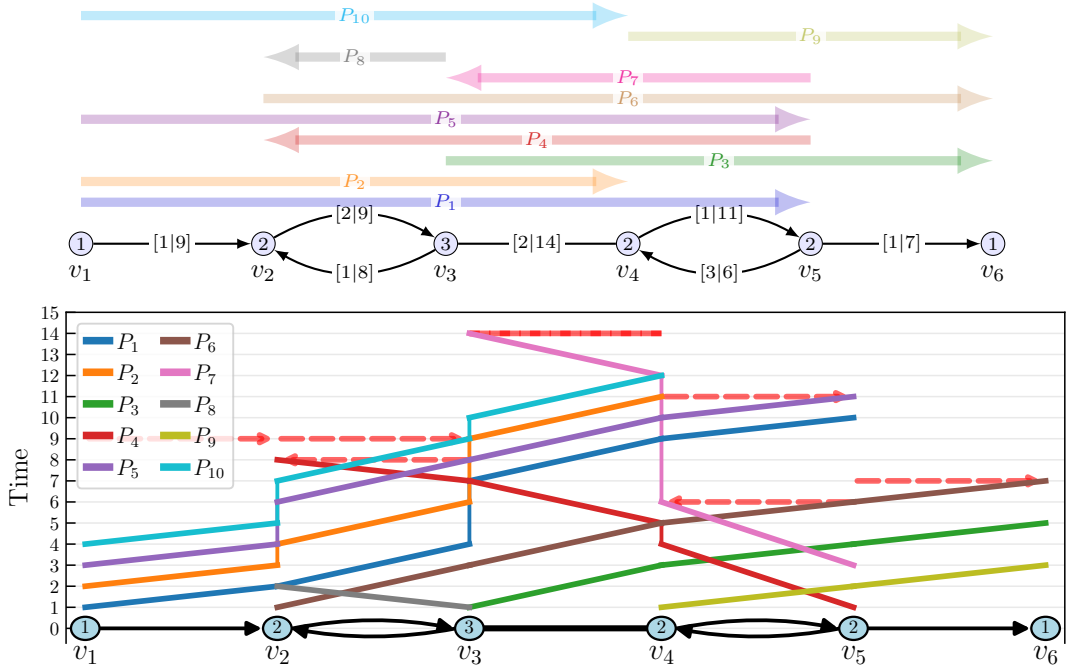

A \emph{temporalization}~$\TER$ of an instance~$I$
is an assignment $\TER(P,e_i)\in T$ for each path $P=(e_1,\dots,e_k)\in \Pset$ such that
$\TER(P,e_i)+\trt(e_i)\leq \TER(P,e_{i+1})$
and~$\TER(P,e_i)+\trt(e_i)\leq \dl(e_i)$
(we refer to this as \emph{adequacy});
$\TER(P,e_i)$ corresponds to the departure time of~$P$ on edge~$e_i$.\footnote{Such temporalizations can be equivalently defined as assignments of departure times for paths from their vertices; We will make use of this equivalence throughout.}
Given a temporalization,
two paths~$P,P'$ are \emph{temporally edge-disjoint}
when the following hold:
(i) if $(v,w)\in \Rel(P)\cap \Rel(P')$,
then $\TER(P,\{v,w\}) \neq \TER(P',\{v,w\})$
(resp., $\TER(P,(v,w)) \neq \TER(P',(v,w))$),
and
(ii) if $e=\{v,w\}\in E$, $(v,w)\in \Rel(P)$, and $(w,v)\in \Rel(P')$,
then $|\TER(P,e) - \TER(P',e)|\geq \max\{1,\trt(e)\}$
(in the special case of zero traversal time,
it forbids departing on that edge simultaneously).
A temporalization is \emph{temporally edge-disjoint} if every two paths~$P,P'\in\Pset$ are temporally edge-disjoint.
A path~$P=(v_1,\dots,v_{k+1})\equiv (e_1,\dots,e_k)$ is (located) at an inner vertex~$v_i$, $1<i\leq k$,
at the time steps $\TER(P,e_{i-1})+\trt(e_{i-1}),\dots,\TER(P,e_i)$.
As special cases,
a path is (located) only at time $\TER(P,e_1)$ at~$\src(P)=v_1$ and only at time~$\TER(P,e_k) + \trt(e_k)$ at~$\snk(P)=v_{k+1}$.
A temporalization \emph{respects} the (vertex) capacities
if for every vertex~$v\in V$ and every time step~$t\in T$,
at most~$c(v)$ many paths are located on~$v$ at~$t$.
A decaying graph is \emph{uncapacitated} when~$c(v)\geq |\Pset(v)|$ for all~$v\in V$,
i.e.,
when all paths containing~$v$ can be located at $v$ at the same time,
and \emph{capacitated} otherwise.
Altogether,
a temporalization is~\emph{valid}
if it is temporally edge-disjoint and respects the vertex capacities.
Finally,
we seek to solve the following problem
(see \Cref{fig:example}(right) for a witnessing solution):

\decprob{\simevacTsc{} (\simevacAcr{})}{dwgbp}
{an instance~$I=(\DG, \Pset)$ with decaying graph~$\DG=\DGtuple$
and
set~$\Pset=\{P_1,\dots,P_p\}$ of paths in~$G=(V,E,A)$}
{there is a valid temporalization~$\TER$ of~$I$.
}

\paragraph*{Our contributions.}

We contribute both on the theoretical and experimental side.
For all settings of the deadlines and capacities being constant or general,
we prove whether \simevacAcr{} on capacitated decaying paths and stars is \NP-hard or polynomial-time solvable
(see~\Cref{tab:results} for an overview).
\begin{table}
 \centering
 \caption{Overview on our complexity-theoretic results.
 Each hardness result still holds when the respective decaying graph is \exogenous.
	}
 \label{tab:results}
	\begin{tblr}{
		colspec = {r l l l},
	}
			\toprule
		& deadlines & \SetCell[c=2]{c} capacitated \\\cmidrule{3-4}
		&  & constant capacities & scaling capacities \\\midrule\midrule
		\SetCell[r=2]{r,bg=gray!10} path
			& constant & \SetCell[c=2]{c,bg=green!10} P~\tref{thm:path:dp} \\
			& scaling  & \SetCell[c=2]{c,bg=red!5} \NP-hard~\tref{thm:nphard:paths:capone}
			\\
		\SetCell[r=2]{r,bg=gray!20} star
			& constant & \SetCell{bg=green!10} P~\tref{obs:nphard:stars:occ}
									& \SetCell{bg=red!10} \NP-hard~\tref{thm:nphard:stars:cap} \\
			& scaling  & \SetCell{bg=red!10} \NP-hard~\tref{thm:nphard:stars:cap:one}
									& \SetCell{bg=red!10} \NP-hard
			\\
		\SetCell[r=2]{r,bg=gray!10} tree
			& constant & \SetCell[r=2,c=2]{c,bg=red!15} \NP-hard~\tref{thm:nphard:trees} \\
			& scaling  &
			\\\bottomrule
	\end{tblr}
\end{table}
Our main result
is that \simevacAcr{}
is \NP-hard on uncapacitated decaying trees with constant deadlines.
We provide an integer linear program (ILP)
to compute the minimum addition to every deadline
to achieve a feasible instance,
i.e.,
the latest possible time to evacuate.
With experiments\footnote{Code and data is available at \url{https://github.com/buhtig-tf/Smooth-Routing-in-Decaying-Trees}.}
on artificial
(decaying stars and paths)
and semi-artificial instances
(street networks from German cities along rivers
for which we simulated flooding),
we prove the practicality of the ILP and its relaxation.

\paragraph*{Related work.}
Issues caused by human behavior through evacuation like congestion~\cite{Litman06} are problematic.
A common framework is
evacuation via origin-destination pairs~\cite{SchichlS15,WuSZ16,HigashikawaKTT24,ShahabiW14},
suitable for centralized planning
e.g.\ when evacuating inhabitants with no access to private automobiles.
Computing and testing evacuation paths
are mostly done for simultaneously initiated evacuation
via,
e.g.,
time-expanded networks and flow models~\cite{dhamala2015survey,MishraMP18,PyakurelNDD19,AkridaCGKS19,YuL12}.

Closest to us is the work
by Klobas et al.~\cite{KlobasMMNZ23}
and Kunz et al.~\cite{KunzMZ23},
studying the problem
\prob{Temporally Disjoint Paths (TDP)}.
Their model yet
differs in several aspects:
They study temporal graphs where all connections are edges that
can also reappear,
they have no vertex capacities and travel times
(translated into our model,
all vertex capacities are one
and all travel times are zero),
and their paths are given only by origin-destination pairs.
Due to the latter,
\prob{TDP} is
\NP-hard already if~$\tau=1$~\cite{KlobasMMNZ23};
\simevacAcr{} is (trivially)
polynomial-time solvable in this case.
Note that
on decaying trees origin-destination pairs fully describe a path.
Kunz et al.~\cite{KunzMZ23} proved \prob{TDP} to be
hard for different graph restrictions and parameters,
e.g.,
NP-hard
(and W[1]-hard w.r.t.\ the number of vertices)
on temporal stars
(without constantly-bounded lifetime).

Waiting-time restrictions for temporal paths are studied~\cite{CasteigtsHMZ21},
even with no waiting time for a set of paths all with the same origin-destination pair~\cite{FLUSCHNIK201969}.
Also studied is the effect of ceasing connections on shortest temporal paths~\cite{BoeckmannTW23}.
Finally,
orthogonal to our concept of simultaneous evacuation,
there is work related to successive evacuation in temporal graphs~\cite{ChimaniT23,FluschnikNSZ23}.

Multi-agent pathfinding~\cite{AtzmonSFWBZ20,SternSFK0WLA0KB19,GaoLLYLW24} coordinates multiple agents sharing limited network resources while e.g.\ avoiding safety-critical conflicts such as collisions~\cite{AlmagorKL24,PertzovskiySZF25} or overloading~\cite{ShahabiW14}.
It appears in dynamic settings~\cite{MuranoPR15}, where edge availability or travel times vary, and formulated on graphs~\cite{YuL13,YuL16}, which encode feasible routes.

\section{Preliminaries}
\label{sec:prelim}

We denote by~$\N$ and~$\Nzero$ the natural numbers excluding and including zero,
respectively.
For a graph~$G=(V,E,A)$,
we denote by~$\calU(G)$ the graph obtained from replacing each arc with an edge
and then deleting duplicated edges.
A \emph{decaying tree}~$\DG=\DGtuple$ is a decaying graph
with $\calU(G)$
being a \emph{tree}
(a connected, cycle-free graph),
where~$G=(V,E,A)$.
If~$\calU(G)$ is a path (star),
then we call~$\DG$ a decaying path (star).
A \emph{star} with~$n$ vertices is a graph with center vertex~$v^*$ and vertices~$v_1,\dots,v_{n-1}$,
each of which is adjacent only with~$v^*$.
In a decaying tree with vertices~$s,z$,
we also write~$(s,z)$-path for the unique path with source~$s$ and sink~$z$.
A decaying tree is \emph{\exogenous},
when for every~$t\in\set{\tau}$,
$\calU(G_t)$ is a tree plus isolated vertices,
where $G_t=(V,E_t,A_t)$
such that~$e\in E_t$ (resp.\ $a \in A_t$) if and only if~$e\in E$ and~$\dl(e)\geq t$ (resp.\ $a\in A$ and $\dl(a)\geq t$).

\section{Computational Complexity}
\label{sec:hardness}

We distinguish between
decaying paths (\Cref{sec:nphard:paths}),
decaying stars (\Cref{sec:nphard:stars}),
and decaying trees (\Cref{sec:nphard:trees}).

\subsection{On capacitated paths}
\label{sec:nphard:paths}

In \Cref{sec:nphard:paths:capone} we show that \simevacAcr{} on decaying paths
is \NP-hard even for unit capacities.
In \Cref{sec:nphard:paths:constdl},
we show that
if the lifetime is constant,
then \simevacAcr{} is polynomial-time solvable.

\subsubsection{On paths with capacity one}
\label{sec:nphard:paths:capone}

\begin{theorem}%
 \label{thm:nphard:paths:capone}
 \simevacAcr{} is \NP-hard even on \exogenous{} decaying paths
 where every vertex has capacity one.
\end{theorem}

Our reduction builds on the reduction from~Klobas~\etal~\cite[Theorem~6]{KlobasMMNZ23}.
On a high level,
one can represent colored intervals
as paths going from ``left to right'' on the decaying path
such that they all cross a common vertex~$v^*$
and can only start and finish at time steps according to the respective interval's ``left and right'' endpoints.
Then,
with unit capacities,
there is a valid temporalization
(that in particular respects the capacity of~1 of~$v^*$)
if and only if there is a set of pairwise disjoint intervals of each color,
where the latter corresponds to an \NP-hard task.

Our reduction builds on the reduction from~Klobas~\etal~\cite[Theorem~6]{KlobasMMNZ23}.
Consequently,
we give a polynomial-time many-one reduction from the \NP-hard~\cite{BevernMNW15} \misuigTsc{} (\misuigAcr) problem,
where,
given a unit interval graph~$G=(V,E)$
with partition~$V=V_1\uplus\dots\uplus V_k$ corresponding to~$k$ colors
where $V_i$ is an independent set for every~$i\in\set{k}$,
the question is whether there is an independent set~$X$ in~$G$ with~$|X\cap V_i|=1$ for every~$i\in\set{k}$.
Herein, a graph~$G=(V,E)$ is a unit interval graph if and only if it can be represented as follows:
each vertex~$v\in V$ is assigned an interval~$S_v\ceq [a_v,b_v]$ of unit length
such that~$\{v,w\}\in E$ if and only if~$S_v\cap S_w\neq \emptyset$.
We assume that every two endpoints of the intervals have a difference of at least two.

\begin{construction}
 \label{const:nphard:paths:capone}
 The construction follows Klobas~\etal~\cite[Theorem~6]{KlobasMMNZ23}
 and thus we adapt their notations.
 We first describe their basic construction
 (see \cref{fig:Klobas} for an illustration of the temporal graph from~\cite[Theorem~6]{KlobasMMNZ23}).
	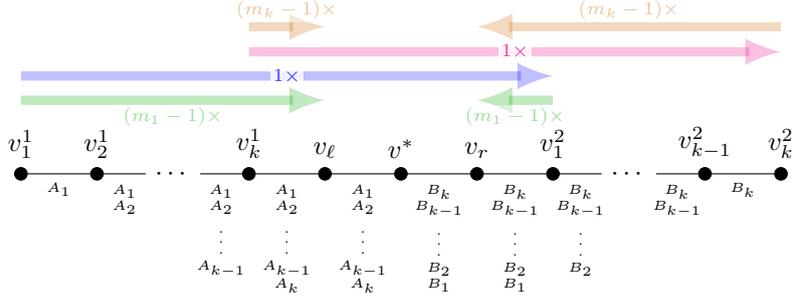
\begin{figure}[t]
		\centering
		\begin{tikzpicture}
			\tikzpreamble{}
			\def\xr{1}
			\def\yr{0.975}

			\node (v11) at (-5*\xr,0)[xnode,label=90:{$v^1_1$}]{};
			\node (v12) at (-4*\xr,0)[xnode,label=90:{$v^1_2$}]{};
			\node (v1j) at (-3*\xr,0)[]{$\cdots$};
			\node (v1k) at (-2*\xr,0)[xnode,label=90:{$v^1_k$}]{};
			\node (vl) at (-1*\xr,0)[xnode,label=90:{$v_\ell$}]{};
			\node (vstar) at (0,0)[xnode,label=90:{$v^*$}]{};
			\node (vr) at (1*\xr,0)[xnode,label=90:{$v_r$}]{};
			\node (v21) at (2*\xr,0)[xnode,label=90:{$v^2_1$}]{};
			\node (v2j) at (3*\xr,0)[]{$\cdots$};
			\node (v2k1) at (4*\xr,0)[xnode,label=90:{$v^2_{k-1}$}]{};
			\node (v2k) at (5*\xr,0)[xnode,label=90:{$v^2_k$}]{};

			\foreach \x/\y/\z in {
			  v11/v12/{$A_1$},
			  v12/v1j/{$\begin{matrix} A_1 \\ A_2 \end{matrix}$},
				v1j/v1k/{$\begin{matrix} A_1 \\ A_2 \\ \vdots \\ A_{k-1}\end{matrix}$},
				v1k/vl/{$\begin{matrix} A_1 \\ A_2 \\ \vdots \\ A_{k-1} \\ A_k\end{matrix}$},
				vl/vstar/{$\begin{matrix} A_1 \\ A_2 \\ \vdots \\ A_{k-1} \\ A_k\end{matrix}$},
				vstar/vr/{$\begin{matrix} B_k \\ B_{k-1} \\ \vdots  \\ B_2 \\ B_1\end{matrix}$},
				vr/v21/{$\begin{matrix} B_k \\ B_{k-1} \\ \vdots  \\ B_2 \\ B_1\end{matrix}$},
				v21/v2j/{$\begin{matrix} B_k \\ B_{k-1} \\ \vdots  \\ B_2\end{matrix}$},
				v2j/v2k1/{$\begin{matrix} B_k \\ B_{k-1}\end{matrix}$},
				v2k1/v2k/{$B_k$}
			}{
				\draw[-] (\x) to node[midway, below]{\tiny \z}(\y);
			}

			\draw[xpathA] ($(v11)+(0,1.0*\yr)$) to node[xemarkb]{$(m_1-1)\times$}($(vl)+(0,1.0*\yr)$);
			\draw[xpathA] ($(v21)+(0,1.0*\yr)$) to node[xemarkb]{$(m_1-1)\times$}($(vr)+(0,1.0*\yr)$);
			\draw[xpathB] ($(v11)+(0,1.33*\yr)$) to node[xemark]{$1\times$}($(v21)+(0,1.33*\yr)$);
			\draw[xpathC] ($(v1k)+(0,2*\yr)$) to node[xemarka]{$(m_k-1)\times$}($(vl)+(0,2*\yr)$);
			\draw[xpathC] ($(v2k)+(0,2*\yr)$) to node[xemarka]{$(m_k-1)\times$}($(vr)+(0,2*\yr)$);
			\draw[xpathD] ($(v1k)+(0,1.66*\yr)$) to node[xemark]{$1\times$}($(v2k)+(0,1.66*\yr)$);
		\end{tikzpicture}
		\caption{Illustration of the temporal graph constructed by Klobas~\etal~\cite[Theorem~6]{KlobasMMNZ23}, where the time label set of each edge
		(i.e., when it is present)
		is the union of the sets drawn below of it.
		Formally,
		they have no vertex capacities which translates to capacity one in our model.
		For the two color classes 1 and~$k$,
		the corresponding paths are depicted in green/blue and orange/magenta,
		respectively.}
		\label{fig:Klobas}
	\end{figure}
 To this end,
 we construct an instance~$I'$ of \simevacAcr{}
 from a given instance~$I=(G=(V=V_1\uplus\dots\uplus V_k,E))$ of \misuigAcr{} as follows.
 Initially,
 we have the vertex set $V'=\{v_\ell,v^*,v_r\}\cup\{v^j_i\mid j\in\{1,2\}, i\in\set{k}\}$
 that forms a path with edge sets $\{\{v^j_i,v^j_{i+1}\}\mid j\in\{1,2\}, i\in\set{k-1}\}$,
 $\{\{v^1_k,v_\ell\},\{v_\ell,v^*\}\}$,
 and $\{\{v^*,v_r\},\{v_r,v^2_1\}\}$,
 where each edge has traversal time zero
 and deadline~$\tau$,
 where $\tau$ is the largest endpoint of the intervals.
 For each color~$i$,
 we associate a left-label set~$L_i = \bigcup_{j=1}^i A_j$ where~$A_j\ceq \bigcup_{v\in V_j} \{a_v\}$ with edge~$\{v^1_i,v^1_{i+1}\}$
 (where we identify~$v^1_{k+1}$ with~$v_\ell$).
 Similarly,
 for each color~$i\in\set{k}$,
 we associate a right-label set~$R_i = \bigcup_{j=i}^k B_j$ where~$B_j\ceq \bigcup_{v\in V_j} \{b_v\}$
 with edge~$\{v^2_{i-1},v^2_{i}\}$
 (where we identify~$v^2_{0}$ with~$v_r$).
 The edge~$\{v_\ell,v^*\}$ and $\{v^*,v_r\}$
 is associated with the label sets~$L_k$ and~$R_1$,
 respectively.
 We make use of the fact that left of~$v^*$,
 the time set of labels on edge~$\{v^1_i,v^1_{i+1}\}$
 is a strict subset of the time labels of edge~$\{v^1_j,v^1_{j+1}\}$
 for every~$j>i$,
 and right of~$v^*$,
 the time set of labels on edge~$\{v^2_i,v^2_{i+1}\}$
 is a strict subset of the time labels of edge~$\{v^1_j,v^1_{j+1}\}$
 for every~$j<i$.
 Next,
 we explain the paths they constructed,
 to which refer to as color-paths.
 For each color~$i$,
 let~$m_i=|V_i|$,
 i.e.,
 the number of intervals of color~$i$.
 Add~$m_i-1$ many~$(v_i^1,v_\ell)$-paths,
 $m_i-1$ many~$(v_i^2,v_r)$-paths,
 and one~$(v_i^1,v_i^2)$-path
 (referred to as the validation color-path).
 We also refer to a color-path starting left---except for the validation paths---and
 right of~$v^*$
 as left and right color-path, respectively.
 This finishes the central construction based on Klobas~\etal~\cite[Theorem~6]{KlobasMMNZ23}.

 We next proceed with our modifications.
 We describe first our path extension,
 i.e.,
 how we extend the path constructed so far at both ends,
 and then introduce what we call blocker paths.
 See~\cref{fig:nphard:paths:capone} for illustration.
 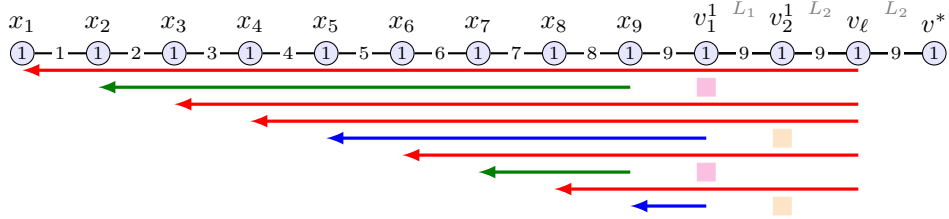
\begin{figure}[t]
		\centering
		\begin{tikzpicture}
			\tikzpreamble{}
			\def\xr{0.5}
			\def\yr{0.225}
			\node (A1) at (0*\xr,0)[xcnode,label=90:{$x_1$}]{$1$};
			\node (A2) at (2*\xr,0)[xcnode,label=90:{$x_2$}]{$1$};
			\node (A3) at (4*\xr,0)[xcnode,label=90:{$x_3$}]{$1$};
			\node (A4) at (6*\xr,0)[xcnode,label=90:{$x_4$}]{$1$};
			\node (A5) at (8*\xr,0)[xcnode,label=90:{$x_5$}]{$1$};
			\node (A6) at (10*\xr,0)[xcnode,label=90:{$x_6$}]{$1$};
			\node (A7) at (12*\xr,0)[xcnode,label=90:{$x_7$}]{$1$};
			\node (A8) at (14*\xr,0)[xcnode,label=90:{$x_8$}]{$1$};
			\node (A9) at (16*\xr,0)[xcnode,label=90:{$x_9$}]{$1$};
			\node (B1) at (18*\xr,0)[xcnode,label=90:{$v^1_1$}]{$1$};
			\node (B2) at (20*\xr,0)[xcnode,label=90:{$v^1_2$}]{$1$};
			\node (Vl) at (22*\xr,0)[xcnode,label=90:{$v_\ell$}]{$1$};
			\node (Vstar) at (24*\xr,0)[xcnode,label=90:{$v^*$}]{$1$};

			\foreach \x in {1,...,9}{
				\node (Y\x) at (A\x|-{(0,-\x*\yr)})[draw=none]{};
			}
			\foreach \x/\y/\z in {B1/B2/{$L_1$},%
					B2/Vl/{$L_2$},%
					Vl/Vstar/{$L_2$}%
					}{
				\draw[-,thick] (\x) -- node[above,yshift=1em,color=gray] {\scriptsize\z} (\y);
			}
			\foreach \x/\y/\z in {A1/A2/$1$,A2/A3/$2$,A3/A4/$3$,A4/A5/$4$,
					A5/A6/$5$,A6/A7/$6$,A7/A8/$7$,A8/A9/$8$,A9/B1/$9$,B1/B2/$9$,B2/Vl/$9$,Vl/Vstar/$9$
					}{
				\draw[-,thick] (\x) -- node[midway,fill=white,inner sep=1pt] {\scriptsize\z} (\y);
			}

			\def\blckA{red}
			\def\blckO{blue}
			\def\blckN{green!50!black}
			\foreach \x/\y/\z/\clr in {Vl/A1/Y1/\blckA,
				A9/A2/Y2/\blckN,
				Vl/A3/Y3/\blckA,
				Vl/A4/Y4/\blckA,
				B1/A5/Y5/\blckO,
				Vl/A6/Y6/\blckA,
				A9/A7/Y7/\blckN,
				Vl/A8/Y8/\blckA,
				B1/A9/Y9/\blckO%
				}{
				\draw[->,>=latex,very thick,color=\clr] (\x|-\z) to (\y|-\z);
			}

			\foreach \x/\y/\clr in {
				B1/Y2/magenta,
				B1/Y7/magenta,
				B2/Y5/yellow!50!red,
				B2/Y9/yellow!50!red%
				}{
				\node at (\x|-\y)[color=\clr, fill, opacity=0.25, draw]{};
			}

		\end{tikzpicture}
		\caption{Red paths block all color classes,
		green none,
		and blue only the fist color class.
		Magenta and orange blocks indicate all possible starting times of the color-paths of color~$1$ and~$2$,
		respectively.
		Here,
		$A_1=L_1=\{2,7\}$, $A_2=\{5,9\}$, and~$L_2=\{2,5,7,9\}$.
		}
		\label{fig:nphard:paths:capone}
	\end{figure}
 We add the vertices~$x_1,\dots,x_\tau$ and~$y_\tau,\dots,y_1$,
 and make~$v^1_1$ adjacent with~$x_\tau$ having deadline~$\tau$,
 and make~$v^2_k$ adjacent with~$y_\tau$ having deadline~$\tau$.
 For each~$t\in\set{\tau-1}$,
 we add the edges~$\{x_t,x_{t+1}\}$ and $\{y_t,y_{t+1}\}$
 each with deadline~$t$.

 We construct the following \emph{blocker} path~$X_t$ for each~$t\in\set{\tau}$
 left of~$v^*$.
 If~$t$ is in all label sets,
 then~$X_t$ is an~$(x_\tau,x_t)$ path.
 If~$t$ is in no label set,
 then~$X_t$ is an~$(v_\ell,x_t)$ path.
 If~$t$ is in time label set~$L_{i+1}$ but not in~$L_i$,
 then~$X_t$ is an~$(v^1_i,x_t)$ path.
 We also call~$X_t$ left-blocker.
 We construct the following blocker path~$Y_t$ for each~$t\in\set{\tau}$
 right of~$v^*$.
 If~$t$ is in all time label sets,
 then~$Y_t$ is an~$(y_\tau,y_t)$ path
 (as a special case here,
 if~$t=\tau$,
 then there is no path~$Y_\tau$).
 If~$t$ is in no time label set,
 then~$Y_t$ is an~$(v_r,y_t)$ path.
 If~$t$ is in time label set~$R_{i-1}$ but not in~$R_i$,
 then~$Y_t$ is an~$(v^2_i,y_t)$ path.
 We also call~$Y_t$ right-blocker.
\end{construction}

\begin{lemma}
 \label{obs:nphard:paths:capone:xytau}
 Let~$I'$ be a \yes-instance.
 Then at~$x_\tau$,
 the left-blocker path~$X_t$ is exactly at time step~$t$,
 and at~$y_\tau$,
 the right-blocker path~$Y_t$ is exactly at time step~$t$.
\end{lemma}

\begin{proof}
 Suppose not,
 that is,
 some $X_t$ ($Y_t$) is at time step~$t'<t$ at~$x_\tau$ ($y_\tau$)
 (recall that~$X_t$ ($Y_t$) must reach~$x_t$ ($y_t$) latest at time step~$t$
 since each~$\{x_t,x_{t+1}\}$ ($\{y_{t+1},y_t\}$) has deadline~$t$).
 Then~$X_{t'}$ ($Y_{t'}$) must be earlier than~$t'$ at~$x_\tau$ ($y_\tau$).
 Recall that all left-blocker (right-blocker) paths
 visit~$x_\tau$ ($y_\tau$).
 Hence,
 iterating the argument leads to a contradiction.
\end{proof}

\begin{lemma}
\label{obs:nphard:paths:capone:val}
 Let~$I'$ be a \yes-instance.
 Then for every color~$i\in\set{k}$,
 every left color-path
 and validation path
 for~$i$
 starts at time step~$a_v$ for some~$v\in V_i$,
 and every right color-path for~$i$
 starts at time step~$b_v$ for some~$v\in V_i$.
 Moreover,
 at the same time it starts,
 every left color-path
 reaches~$v_\ell$,
 validation path reaches~$v^*$,
 and right color-path reaches~$v_r$.
\end{lemma}

\begin{proof}
 We only discuss the left side since the right side goes analogously.
 Suppose not,
 and lets start with the smallest~$i\in\set{k}$
 where this happens.
 If~$i=1$,
 we know that by~\cref{obs:nphard:paths:capone:xytau}
 and since every path starts at~$x_\tau$ or right of it,
 that the only time steps at which no path is at~$v^1_1$
 are those in~$L_1$.
 Suppose not,
 that is,
 there is a~$a_v<t^*<a_w$,
 $v,w\in V_1$,
 where a color-path starts.
 Then the blocker path~$X_{t^*}$ must be leaving~$v^1_1$ at some time step~$t<t^*$.
 Thus,
 $t^*$ paths must leave at~$t^*$ time steps.
 Hence,
 at each time step one is leaving,
 in particular on time step~$a_v$.
 This path then collides at~$x_\tau$ with~$X_{a_v}$,
 which is again because of~\cref{obs:nphard:paths:capone:xytau},
 a contradiction.

 Now suppose that a color-path~$P$ starting left of~$v^*$ leaving at~$a_v$,
 $v\in V_1$,
 arrives at~$v_\ell$ at some time step~$t>a_v$.
 We know that there is blocker path~$X_{a_v+1}$
 that starts at~$v_\ell$ since there is a difference of at least two between any two interval start times.
 Hence,
 $X_{a_v+1}$ and~$P$ must be at the same vertex at the same time,
 yielding a contradiction.
 With the same argument,
 the validation path must proceed from~$v_\ell$ to~$v^*$
 at time step~$a_v$ when it arrives.

 Now we can assume that it holds true for all~$1\leq i'<i$.
 By induction,
 only the time labels~$a_v$ for each~$v\in V_i$ remain
 as possible starting times for the color-paths regarding color~$i$.
 With the same arguments as before,
 each of the color-paths regarding color~$i$ starting left of~$v^*$
 are leaving~$v^1_i$ and arriving at~$v_\ell$
 (and hence at~$v^*$ for the validation path)
 at time step~$a_v$ for each~$v\in V_i$.
\end{proof}

Finally,
we established the same
setup for the color-paths as in Klobas~\etal~\cite[Theorem~6]{KlobasMMNZ23}.
Thus, the correctness proof goes analogously to theirs,
yet we give it for the sake of completeness.

\begin{proof}[Proof of \cref{thm:nphard:paths:capone}]
 Let~$I'=(\DG,\Pset)$ be the instance of \simevacAcr{}
 obtained from an instance~$I=(G,k)$ with interval graph~$G=(V,E)$
 and partition~$V=V_1\uplus\dots\uplus V_k$
 of \misuigAcr{} via \cref{const:nphard:paths:capone}.
 Note that we can construct~$I'$ in time polynomial in~$|I|$,
 since,
 in accordance with Klobas~\etal~\cite[Theorem~6]{KlobasMMNZ23},
 we can assume a interval graph representation of~$G$ where every two endpoints of the intervals have a difference of at least two and all endpoints are in~$O(k+|V|^2)$.

 \RD{}
 Let~$I$ be a \yes-instance and
 let~$X$ be an independent set in~$G$ such that~$|X\cap V_i|=1$ for every~$i\in\set{k}$.
 Let~$u_i$ be the vertex in~$X\cap V_i$ for every~$i\in\set{k}$.
 Route all the blocker paths in accordance with~\cref{obs:nphard:paths:capone:xytau}.
 Now,
 for every~$i\in\set{k}$,
 assign the~$(v_i^1,v_i^2)$-path the time step~$a_{u_i}$ on all edges left of~$v^*$,
 and~$b_{u_i}$ on all edges right of~$v^*$.
 For each~$v\in V_i\setminus \{u_i\}$,
 assign one of the left color-paths the time step~$a_v$ to all its edges,
 and one of the right color-paths the time step~$b_v$ to all its edges.
 By this, all color-paths are assigned.
 Recall that this is possible since all endpoints differ by at least two.
 Finally,
 observe that for every distinct~$i,j\in\set{k}$,
 it holds that $S_{u_i}\cap S_{u_j}=\emptyset$ since~$u_i$ and~$u_j$ are part of the independent set~$X$.
 Hence,
 since the validation color-paths for~$i$ and~$j$ are located at~$v^*$ at exactly all time steps in~$S_{u_i}$ and in~$S_{u_j}$,
 respectively,
 they do not meet in~$v^*$.
 Thus,
 we constructed a valid temporalization.

 \LD{}
 Let~$I'$ be a \yes-instance
 and let~$\TER$ be a valid temporalization such
 that every validation color-path leaves~$v^*$ earliest possible.
 By \cref{obs:nphard:paths:capone:xytau},
 the left-blocker paths~$X_t$ and $Y_t$ are exactly at time step~$t$
 at~$x_\tau$ and at~$y_\tau$,
 respectively.
 By~\cref{obs:nphard:paths:capone:val},
 for every color~$i\in\set{k}$,
 every left color-path for~$i$
 must start and reach~$v_\ell$ at time step~$a_v$ for some~$v\in V_i$
 and every right color-path
 for~$i$ must start and reach~$v_r$ at time step~$b_w$ for some $w\in V_i$.
 Same for the validation color-path for~$i$:
 it must start and reach~$v^*$ at time step~$a_v$ for some~$v\in V_i$,
 and depart from~$v^*$ and terminate at step~$b_w\geq a_v$ with~$w\in V_i$.
 Since~$V_i$ forms an independent set and due to the choice of~$\TER$ such that every validation color-path leaves~$v^*$ earliest possible,
 it holds true that~$v=w$.
 Thus,
 every validation color-path for color~$i\in\set{k}$ corresponds to an interval
 and hence to a vertex~$u_i\in V_i$.
 We claim that~$X=\{u_i\mid i\in\set{k}\}$ is a solution to~$I$.
 Clearly,
 $|X\cap V_i|=1$ by construction for every~$i\in\set{k}$.
 Since~$\TER$ is a valid temporalization,
 we know that the capacity of 1 of~$v^*$ is respected at all time steps.
 Thus,
 for every two distinct colors~$i,j\in\set{k}$,
 the two validation color-paths for colors~$i$ and~$j$ do not meet in~$v^*$
 (i.e., they are located at~$v^*$ for disjoint sets of time steps).
 Hence,
 the two intervals~$S_{u_i}$ and~$S_{u_j}$ are disjoint,
 and hence,
 there is no edge between~$u_i$ and~$u_j$.
 It follows that~$X$ is solution.
\end{proof}

\subsubsection{On paths with constant deadlines}
\label{sec:nphard:paths:constdl}

\newcommand{\occ}{\operatorname{vl}}
\newcommand{\InS}{R}
\newcommand{\OutS}{L}
\newcommand{\InF}{\rho}
\newcommand{\OutF}{\lambda}

Next we show that
presumably,
\cref{thm:nphard:paths:capone} does not hold for constant deadlines,
even if we allow general capacities.
In fact,
we show an even stronger result that
if the maximum \emph{vertex load} $\occ(I)\ceq \max_{v\in V} |\Pset(v)|$,
i.e., the maximum number of paths containing the same vertex,
is constant,
then
\simevacAcr{} is solvable in polynomial-time.

\begin{theorem}
 \label{thm:path:dp}
 \simevacAcr{} on decaying paths is solvable
 in
 $O(n\cdot |\Pset| + \tau^{4\occ(I)}\cdot \occ(I)^2 \cdot n)$
 time,
 where~$I$ denotes any input instance with~$n$ vertices,
 lifetime~$\tau$,
 and set~$\Pset$ of paths.
\end{theorem}

For an instance~$I$ on a decaying path with lifetime~$\tau$,
we can assume that~$\occ(I)\leq 4\tau$,
as otherwise~$I$ is a \no-instance:
at most four paths can arrive and leave at a time step in a vertex.
Thus,
\simevacAcr{} on decaying paths is
polynomial-time
solvable for constant
deadlines.

We prove \cref{thm:path:dp} via dynamic programming
over the vertices of the decaying path from ``left to right''.
On a high level,
the dynamic program at vertex~$v_{i+1}$
considers all possible ways paths that traverse to or from ``the left''
are scheduled to leave or arrive
from or at~$v_i$ and~$v_{i+1}$.
Thereby,
it takes care only about the capacities of~$v_{i}$ and~$v_{i+1}$
as well as about monotonicity and temporally edge-disjointness
regarding the connection(s) between~$v_{i}$ and~$v_{i+1}$.

Let~$\DG=\DGtuple$ be a decaying path
with vertex set~$V=\{v_1,\dots,v_n\}$ and connections exactly between~$v_i$ and~$v_{i+1}$ for all~$i\in\set{n-1}$.
Let~$c$ denote the capacities and~$\Pset$ the set of paths.
For every~$i\in\set{n}$,
let~$\DG_i=(V_i,E_i,A_i,c_i,\trt_i,\dl_i,\tau)$ denote~$\DG$ restricted to~$V_i\ceq \{v_1,\dots,v_i\}$,
where~$E_i$ and~$A_i$ contains every connection from~$E$ and~$A$ with both endpoints in~$V_i$,
$c_i\colon V_i\to \N$ with~$c_i(v)=c(v)$,
$\dl_i\colon E_i\cup A_i\to \set{\tau}$ with~$\dl_i(e)=\dl(e)$,
$\trt_i\colon E_i\cup A_i\to \set[0]{\tau-1}$ with~$\trt_i(e)=\trt(e)$.
Let~$\Pset_i$ denote the set of all paths~$P\in\Pset$ with~$|V(P)\cap V_i|\geq 2$ restricted to~$V_i$;
that is,
a path~$P=(v_j,\dots,v_k)$ with~$j<k$ and~$j<i$ restricted to~$V_i$ is the path~$P_i=(v_j,\dots,v_{\min\{k,i\}})$,
and a path~$P=(v_j,\dots,v_k)$ with~$j>k$ and~$k<i$ restricted to~$V_i$ is the path~$P_i=(v_{\min\{i,j\}},\dots,v_k)$.
Let~$\InS_i\subseteq \Pset_i$ be the set of paths that arrive at~$v_i$
and let $\OutS_i\subseteq \Pset_i$ be the set of paths that depart from~$v_i$.

\begin{construction*}
 \label{constr:path:dp}
Let $I=(\DG, \Pset)$ be an instance of \simevacAcr{}
consisting of a decaying path~$\DG=\DGtuple$ with~$V=\{v_1,\dots,v_n\}$
and
a set~$\Pset=\{P_1,\dots,P_p\}$ of paths in~$G=(V,E,A)$.
For all~$i\in\set{n}$,
for all assignments $\InF_i\colon \InS_i\to \set{\tau}$
of the arrival times of paths in~$\InS_i$ at~$v_i$,
and for all assignments~$\OutF_i\colon \OutS_i\to\set{\tau}$ of the departure times of paths in~$\OutS_i$ from~$v_i$,
define $T[i,\InF_i,\OutF_i]\in\{\true,\false\}$ as follows.

Set~$T[1,\InF_1,\OutF_1] = \true$
(since~$\InS_1=\OutS_1=\emptyset$,
as convention,
there is only one~$\InF_1$ and~$\OutF_1$ mapping $\InS_1$ and $\OutS_1$ to 1 and $\tau$,
respectively).
Iteratively for~$i=1,\dots,n-1$,
we set~$T[i+1,\InF_{i+1},\OutF_{i+1}]=\true$ if and only if
there exist $\InF_{i}$ and~$\OutF_{i}$ such that $T[i,\InF_{i},\OutF_{i}]=\true$ and each of \eqref{eq:path:dp:monoton:arr}--\eqref{eq:path:dp:cap} holds true,
which are defined as follows.
Herein,
with $\ora{e_i}$ we refer to~$(v_i,v_{i+1})$ if $(v_i,v_{i+1})\in A$ and to~$\{v_i,v_{i+1}\}$ otherwise;
similarly, with~$\ola{e_i}$ we refer to~$(v_{i+1},v_i)$ if $(v_{i+1},v_i)\in A$ and to~$\{v_i,v_{i+1}\}$ otherwise.

First,
we address monotonicity and deadlines.
\begin{align}
  \forall\, P\in \InS_i\cap \InS_{i+1}: && \InF_i(P) &\leq \InF_{i+1}(P)-\trt(\ora{e_i})
  \label{eq:path:dp:monoton:arr}
  \\
  \forall\, P\in \OutS_i\cap \OutS_{i+1}: && \OutF_i(P) &\geq \OutF_{i+1}(P)+\trt(\ola{e_i})
  \label{eq:path:dp:monoton:dep}
  \\
  \forall\, P\in \InS_{i+1}: && \dl(\ora{e_i}) &\geq \InF_{i+1}(P)
  \label{eq:path:dp:dl:arr}
  \\
  \forall\, P\in \OutS_{i+1}: && \dl(\ola{e_i}) &\geq \OutF_{i+1}(P)+\trt(\ola{e_i})
  \label{eq:path:dp:dl:dep}
\end{align}
Next,
we address temporal disjointness.
\begin{align}
  \forall\, P,Q\in \InS_{i+1},\, P\neq Q:\qquad\quad  \InF_{i+1}(P) \neq \InF_{i+1}(Q)\qquad\qquad
  \label{eq:path:dp:td:arr}
  \\
  \forall\, P,Q\in \OutS_{i+1},\, P\neq Q:\qquad\quad \OutF_{i+1}(P) \neq \OutF_{i+1}(Q)\qquad\qquad
  \label{eq:path:dp:td:dep}
  \\
\begin{aligned}
	e_i=\{v_i,v_{i+1}\}\in E \implies \big(\forall P&\in \InS_{i+1}, Q\in \OutS_{i+1}:
  \\
	\max\{1,\trt(e_i)\}
  &\leq  |\InF_{i+1}(P) - \trt(e_i) - \OutF_{i+1}(Q)| \big)
  \label{eq:path:dp:td}
  \end{aligned}
\end{align}
Finally, we address the capacities. First for~$v_{i+1}$, second for~$v_i$:
\begin{align}
 &c(v_{i+1})=1 \implies \forall P\in \InS_{i+1}, Q\in \OutS_{i+1}:\: \InF_{i+1}(P) \neq \OutF_{i+1}(Q)
 \label{eq:path:dp:cap:end} \\
 &\forall t\in \set{\tau}:\:  c(v_i)  \geq |\{P\in\Pset(v_i)\mid t\in S_i(P)\}|,
 \label{eq:path:dp:cap}
\end{align}
where for all~$P\in\Pset(v_i)$ we have
\begin{align*}
   S_i(P) = \begin{cases}
							[\InF_i(P), \InF_{i+1}(P)-\trt(\ora{e_i})],& \cif{}P\in \InS_i\cap \InS_{i+1} ,
							\\
							[\OutF_{i+1}(P)+\trt(\ola{e_i}), \OutF_{i}(P)],& \cif{}P\in \OutS_i\cap \OutS_{i+1} ,
							\\
							[\InF_i(P), \InF_i(P)],& \cif{}P\in \InS_i\setminus \InS_{i+1} ,
							\\
							[\InF_{i+1}(P)-\trt(\ora{e_i}), \InF_{i+1}(P)-\trt(\ora{e_i})],& \cif{}P\in \InS_{i+1} \setminus \InS_i,
							\\
							[\OutF_i(P), \OutF_i(P)],& \cif{}P\in \OutS_i\setminus \OutS_{i+1} ,
							\\
							[\OutF_{i+1}(P)+\trt(\ola{e_i}), \OutF_{i+1}(P)+\trt(\ola{e_i})],& \cif{}P\in \OutS_{i+1} \setminus \OutS_i. \quad \diamond
							\\
						\end{cases}
\end{align*}
\end{construction*}

Intuitively,
the dynamic program is correct since at vertex~$v_{i+1}$,
to extend paths coming from ``left'' to~$v_{i+1}$ or leaving to the ``left'' from~$v_{i+1}$,
all that matters is whether this extension respects the deadline of the connection(s) between $v_i$ and~$v_{i+1}$ (see~\eqref{eq:path:dp:dl:arr} and~\eqref{eq:path:dp:dl:dep}),
is temporally edge-disjoint between~$v_{i}$ and~$v_{i+1}$ (see~\eqref{eq:path:dp:td:arr}--\eqref{eq:path:dp:td}),
respects the capacities of~$v_{i}$ and~$v_{i+1}$ (see~\eqref{eq:path:dp:cap:end} and \eqref{eq:path:dp:cap}),
and fits to a valid temporalization of the set~$\Pset_{i}$ of paths in~$\DG_{i}$ (in the sense of monotonicity, see~\eqref{eq:path:dp:monoton:arr} and \eqref{eq:path:dp:monoton:dep}).
Whether such an extension fits depends mainly on when the paths to extend arrive or leave~$v_{i}$,
which is stored in~$T[i,\InF_{i},\OutF_{i}]$ for all possible entries $\InF_{i}$ and $\OutF_{i}$.

For each of the at most~$\tau^{2|\Pset(v_{i+1})|}$ entries~$T[i+1,\cdot,\cdot]$,
we have to check at most~$\tau^{2|\Pset(v_{i})|}$ entries~$T[i,\cdot,\cdot]$ with respect to~\eqref{eq:path:dp:monoton:arr}--\eqref{eq:path:dp:cap}.
Assuming all of the path subsets like~$R_i$ and~$L_i$ are precomputed in~$O(n\cdot|\Pset|)$ time,
each of such a check takes time
$O(|\Pset(v_{i+1})|^2)$ (for \eqref{eq:path:dp:monoton:arr}--\eqref{eq:path:dp:cap:end}) and
$O(|\Pset(v_{i-1})| \cdot \log(|\Pset(v_{i-1})|)$ (for~\eqref{eq:path:dp:cap}).

\begin{lemma}
 \label{obs:path:dp:clique}
 For every~$i\in\set{n-1}$,
 constraint~\eqref{eq:path:dp:cap} can be verified
 in~$O(|\Pset(v_i)| \cdot \log(|\Pset(v_i)|))$ time.
\end{lemma}

\begin{proof}
 Let~$i\in\set{n}$.
 The collection~$\calF\ceq\{S_i(P)\mid P\in\Pset(v_i)\}$ can be understood as the input to the problem of finding the maximum clique size in an interval graph,
 where a clique is a subset of intervals that have pairwise non-empty intersection.
 We have the following immediate equivalence:
 $\max_{t\in \set{\tau}} |\{P\in\Pset\mid t\in S_i(P)\}| = r$
 if and only if the maximum clique size of the interval graph formed by~$\calF$ is~$r$.
 This follows directly from the equivalence that a subset~$\calF'$ of intervals have pairwise non-empty intersection
 if and only if
 there exist a~$t\in\set{\tau}$ with~$t\in S_i(P)$ for all~$S_i(P)\in\calF'$.

 Due to~\citet{GuptaLL82},
 we can compute the size of a maximum clique
 for an interval graph with~$N$ vertices
 in~$O(N \cdot \log(N))$ time.
 We can compute~$\calF$ in~$O(|\Pset(v_i)|)$ time,
 since for each~$P\in\Pset(v_i)$,
 we only need to check the indices of $P$'s endpoints
 to determine which of the six cases apply.
 Since we have~$|\calF|\leq |\Pset(v_i)|$,
 the claim follows.
\end{proof}

\begin{lemma}
 \label{lem:path:dp:runtime}
 We can compute all entries of~$T$ from~\cref{constr:path:dp}
 in~$O(n\cdot |\Pset| + \tau^{4\occ(G)}\cdot \occ(G)^2 \cdot n)$.
\end{lemma}

\begin{proof}
 We precompute for each~$i\in\set[2]{n}$
 the set~$\Pset(v_i)$ and its subsets~$\InS_i$ and~$\OutS_i$
 in~$O(|\Pset|)$ time
 by relating the indices of each path's endpoints with~$i$.
 In total,
 we can compute all such sets in~$O(n\cdot |\Pset|)$ time.

 For each~$i\in\set[2]{n}$,
 we need to compute at most~$\tau^{2\cdot|\Pset(v_i)|}$ entries.
 For each such entry,
 we need to compare with at most~$\tau^{2\cdot|\Pset(v_{i-1})|}$ values of~$T[i-1,\cdot,\cdot]$
 (for~$i=2$ we only need one entry to compare with).
 Thus,
 in total,
 we consider at most $\tau^{2\cdot(|\Pset(v_{i-1})|+|\Pset(v_i)|)}$ entries.
 For two fixed pairs of assignments $\InF_i,\OutF_i$ and $\InF_{i-1},\OutF_{i-1}$
 with entries~$T[i,\InF_i,\OutF_i]$ and $T[i-1,\InF_{i-1},\OutF_{i-1}]$,
 we can verify \eqref{eq:path:dp:monoton:arr}--\eqref{eq:path:dp:cap:end} in~$O(|\Pset(v_i)|^2)$.
 Due to~\cref{obs:path:dp:clique},
 we know that we can verify \eqref{eq:path:dp:cap} in~$O(|\Pset(v_{i-1})| \cdot \log(|\Pset(v_{i-1})|)$ time.
 Altogether,
 we fill all entries~$T[i,\cdot,\cdot]$ in~$O(\tau^{2\cdot(|\Pset(v_{i-1})|+|\Pset(v_i)|)}\cdot (|\Pset(v_i)|^2 + |\Pset(v_{i-1})| \cdot \log(|\Pset(v_{i-1})|))$ time.
 It follows that after precomputing~$\Pset(v_i)$ for all~$i\in\set{n}$,
 all entries from~$T$ can be computed in~$O(\tau^{4\occ(G)}\cdot \occ(G)^2 \cdot n)$ time.
\end{proof}

\begin{lemma}
 \label{lem:path:dp:correct}
 For every~$i\in\set{n}$,
 $T[i,\InF_i,\OutF_i] = \true$ if and only if there is a valid temporalization~$\TER$ for instance~$(\DG_i,\Pset_i)$ such that for all~$P\in \InS_i$ we have~$\TER(P,v_{i-1}) = \InF_i(P) -\trt(\ora{e_{i-1}})$ and for all~$P\in \OutS_i$ we have~$\TER(P,v_i) = \OutF_i(P)$.
\end{lemma}

\begin{proof}
 We prove this by induction on~$i$.
 For~$i=1$, the statement is trivially correct.
 Now suppose it holds for~$i\geq 1$,
 and consider~$T[i+1,\InF_{i+1},\OutF_{i+1}]$ for some~$\InF_{i+1},\OutF_{i+1}$.

 \RD{}
 Let $T[i+1,\InF_{i+1},\OutF_{i+1}]=\true$.
 Then,
 there is~$T[i,\InF_{i},\OutF_{i}]=\true$ that witnesses this.
 By induction,
 there is a valid temporalization~$\TER'$ for~$(\DG_i,\Pset_i)$ such that for all~$P\in \InS_i$ we have~$\TER'(P,v_{i-1}) = \InF_i(P) -\trt(\ora{e_{i-1}})$ and for all~$P\in \OutS_i$ we have~$\TER'(P,v_i) = \OutF_i(P)$.
 Let~$\TER$ be the assignment~$\TER'$ where additionally
 for all~$P\in \InS_{i+1}$ we have~$\TER(P,v_i) = \InF_{i+1}(P) -\trt(\ora{e_i})$ and for all~$P\in \OutS_{i+1}$ we have~$\TER(P,v_{i+1}) = \OutF_{i+1}(P)$.
 We claim that~$\TER$ is a valid temporalization for~$(\DG_{i+1},\Pset_{i+1})$.
 Since~$\TER'$ is a valid temporalization for~$(\DG_i,\Pset_i)$,
 we only need to check paths in $\InS_{i+1}\cup \OutS_{i+1}$ for the deadline and in~$(\InS_i\cap \InS_{i+1})\cup (\OutS_i\cap \OutS_{i+1})$ for monotonicity.
 This hence follows by~\eqref{eq:path:dp:monoton:arr}--\eqref{eq:path:dp:dl:dep}:
 \begin{align*}
  \forall\, P\in \InS_i\cap \InS_{i+1}: && \TER(P,v_{i-1}) + \trt(\ora{e_{i-1}})
  &= \InF_i(P) \\
  &&&\stackrel{\eqref{eq:path:dp:monoton:arr}}\leq \InF_{i+1}(P)-\trt(\ora{e_i})
  \\
  &&&= \TER(P,v_i)
  \\
  \forall\, P\in \OutS_i\cap \OutS_{i+1}: && \TER(P,v_i)
  &= \OutF_i(P)
  \\
  &&&\stackrel{\eqref{eq:path:dp:monoton:dep}}\geq \OutF_{i+1}(P)+\trt(\ola{e_i})
  \\
  &&&= \TER(P,v_{i+1}) +\trt(\ola{e_i})
 \end{align*}

 \begin{align*}
  \forall\, P\in \InS_{i+1}: && \dl(\ora{e_i}) &\stackrel{\eqref{eq:path:dp:dl:arr}}\geq \InF_{i+1}(P) = \TER(P,v_i) + \trt(\ora{e_i})
  \\
  \forall\, P\in \OutS_{i+1}: && \dl(\ola{e_i}) &\stackrel{\eqref{eq:path:dp:dl:dep}}\geq \OutF_{i+1}(P)+\trt(\ola{e_i}) = \TER(P,v_{i+1})+\trt(\ola{e_i})
 \end{align*}

 Thus,
 $\TER$ is a temporalization,
 and we proceed to show that its valid.
 For temporal disjointness,
 we only need to check paths in $\InS_{i+1}\cup \OutS_{i+1}$,
 and thus,
 this follows from \eqref{eq:path:dp:td:arr}--\eqref{eq:path:dp:td}:

 \begin{align*}
  \forall\, P,Q\in \InS_{i+1}: && \TER(P,v_i) + \trt(\ora{e_i}) = \InF_{i+1}(P) &\stackrel{\eqref{eq:path:dp:td:arr}}\neq \InF_{i+1}(Q)
  \\
  &&&= \TER(Q,v_i) + \trt(\ora{e_i})
  \\
  \forall\, P,Q\in \OutS_{i+1}: && \TER(P,v_{i+1}) = \OutF_{i+1}(P) &\stackrel{\eqref{eq:path:dp:td:dep}}\neq \OutF_{i+1}(Q)
  \\
  &&&= \TER(Q,v_{i+1})
\end{align*}

\begin{align*}
	e_i=\{v_i,v_{i+1}\}\in E \implies \big(\forall P&\in \InS_{i+1}, Q\in \OutS_{i+1}:
  \\
	\max\{1,\trt(e_i)\}
  &\stackrel{\eqref{eq:path:dp:td}}\leq  |\InF_{i+1}(P) - \trt(e_i) - \OutF_{i+1}(Q)|
  \\
  &=  |\TER(P,v_i) - \TER(Q,v_{i+1})| \big)
\end{align*}

 We only need to check whether the capacities of~$v_i$ and~$v_{i+1}$ are respected,
 and this follows from \eqref{eq:path:dp:cap:end} and \eqref{eq:path:dp:cap}.

 \begin{align*}
 c(v)=1 \implies \forall P\in \InS_{i+1}, Q\in \OutS_{i+1}:\: \TER(P,v_i) + \trt(\ora{e_i}) &= \InF_{i+1}(P)
 \\
 &\stackrel{\eqref{eq:path:dp:cap:end}}\neq \OutF_{i+1}(Q)
 \\
 &= \TER(Q,v_{i+1})
 \end{align*}

Finally,
for \eqref{eq:path:dp:cap}
we just argue that the time intervals of each of the paths agree.
For all~$P\in\Pset(v_i)$,
we have the following correspondences for $S_i(P)$.
\begin{align*}
	P\in \InS_i\cap \InS_{i+1}&:\: [\InF_i(P), \InF_{i+1}(P)-\trt(\ora{e_i})]
	\\
	&\qquad=[\TER(P,v_{i-1}) + \trt(\ora{e_{i-1}}), \TER(P,v_{i})]
	\\
	P\in \OutS_i\cap \OutS_{i+1}&:\: [\OutF_{i+1}(P)+\trt(\ola{e_i}), \OutF_{i}(P)]
	\\
	&\qquad=[\TER(P,v_{i+1})+\trt(\ola{e_i}), \TER(P,v_{i})]
	\\
	P\in \InS_i\setminus \InS_{i+1} &:\: [\InF_i(P), \InF_i(P)]
	\\
	&\qquad= [\TER(P,v_{i-1}) + \trt(\ora{e_{i-1}}), \TER(P,v_{i-1}) + \trt(\ora{e_{i-1}})]
	\\
	P\in \InS_{i+1} \setminus \InS_i&:\: [\InF_{i+1}(P)-\trt(\ora{e_i}), \InF_{i+1}(P)-\trt(\ora{e_i})] \\
	&\qquad= [\TER(P,v_{i}), \TER(P,v_{i})]
	\\
	P\in \OutS_i\setminus \OutS_{i+1} &:\: [\OutF_i(P), \OutF_i(P)] = [\TER(P,v_{i}), \TER(P,v_{i})]
	\\
	P\in \OutS_{i+1} \setminus \OutS_i&:\: [\OutF_{i+1}(P)+\trt(\ola{e_i}), \OutF_{i+1}(P)+\trt(\ola{e_i})]
	\\
	&\qquad =  [\TER(P,v_{i+1})+\trt(\ola{e_i}), \TER(P,v_{i+1})+\trt(\ola{e_i})]
\end{align*}

 \LD{}
 Let~$\TER$ be a valid temporalizaton for~$(\DG_{i+1},\Pset_{i+1})$.
 Let for all~$P\in \InS_{i+1}$,
 $\InF_{i+1}(P) = \TER(P,v_i) + \trt(e_i)$ and for all~$P\in \OutS_{i+1}$,
 $\OutF_{i+1}(P) = \TER(P,v_{i+1})$.
 Moreover,
 let for all for all~$P\in \InS_i$,
 $\InF_i(P) = \TER(P,v_{i-1}) + \trt(e_{i-1})$ and for all~$P\in \OutS_i$,
 $\OutF_i(P) = \TER(P,v_i)$.
 Note that the restriction~$\TER'$ of~$\TER$ with respect to $(\DG_i,\Pset_i)$
 is a valid temporalization for $(\DG_i,\Pset_i)$.
 By induction,
 we know that~$T[i,\InF_i,\OutF_i] = \true$.
 We claim that $T[i+1,\InF_{i+1},\OutF_{i+1}]=\true$ with $T[i,\InF_i,\OutF_i]$ as its witness.
 This goes analogously to~($\Rightarrow$).
\end{proof}

\begin{proof}[Proof of \cref{thm:path:dp}]
	Let $I=(\DG, \Pset)$ be an instance of \simevacAcr{}
	consisting of a decaying graph~$\DG=\DGtuple$
	and
	a set~$\Pset=\{P_1,\dots,P_p\}$ of paths in~$G=(V,E,A)$.
 We compute all entries of~$T$
 (see \cref{constr:path:dp})
 iteratively in
 $O(n\cdot |\Pset| + \tau^{4\occ(G)}\cdot \occ(G)^2 \cdot n)$ time
 (see~\cref{lem:path:dp:runtime}).
 Then,
 due to \cref{lem:path:dp:correct},
 we return~\yes{}
 if there are~$\InF_n,\OutF_n$ such that~$T[n,\InF_n,\OutF_n] = \true$,
 and otherwise we safely return~\no{}
 since we tested all possible assignments~$\InF_n,\OutF_n$.
\end{proof}

\subsection{On capacitated stars}
\label{sec:nphard:stars}

For every instance~$I$ on a decaying star with lifetime~$\tau$ and set~$\Pset$ of paths,
since every path contains center~$v^*$,
we have~$|\Pset|=|\Pset(v^*)|=\occ(I)$ ($\occ$ is defined in \Cref{sec:nphard:paths:constdl}).
By guessing all paths' arrival and departure times at~$v^*$,
we get the following.

\begin{theorem}%
 \label{obs:nphard:stars:occ}
 \simevacAcr{} on decaying stars is solvable in~$O(n\cdot \occ(I) + \tau^{2\cdot \occ(I)}\cdot\occ(I)^2)$ time,
 where~$I$ denotes any input instance with~$n$ vertices and lifetime~$\tau$.
\end{theorem}

\begin{proof}
 For each path,
 guess the arrival and/or departure times at~$v^*$
 such that adequacy holds true---these are most $\tau^{2\cdot \occ(I)}$ many.
 Since each path consists of at most two connections,
 this yields a temporalization
 (note that each path is fully described since the arrival time at~$v^*$ immediately determines its starting time).
 Similar as in~\cref{sec:nphard:paths:constdl},
 we can check in $\occ(I)^2$ time
 (quadratic in the number of paths)
 whether each temporalization is valid.
 Yet,
 in contrast,
 we have to check all edges incident with~$v^*$ that are contained in a path---these are, however, at most~$2|\occ(I)|$ many edges.
\end{proof}

In contrast to decaying paths,
we cannot upper bound $\occ$ by the lifetime
(see~\cref{thm:nphard:stars:cap}).
Yet,
with maximum capacity~$c^*$ and deadline~$d^*$,
there can be at most~$c^*\cdot d^*$ many paths
(in each of the~$d^*$ time steps, at most~$c^*$ paths can be located at~$v^*$).
Since $|\Pset|=\occ(I)$ on decaying stars,
it follows that
	\simevacAcr{} is linear-time solvable on decaying stars with constant capacities and deadlines.

\subsubsection{General capacities and constant deadlines}
\label{sec:nphard:stars:constdl}

\begin{theorem}
 \label{thm:nphard:stars:cap}
 \simevacAcr{} is \NP-hard even on capacitated decaying stars with constant lifetime.
\end{theorem}

We give a polynomial-time many-one reduction from the \NP-hard~\cite{GAREY1976237} problem \prob{Cubic Independent Set},
where,
 given an undirected graph~$G$,
 where each vertex has exactly three neighbors,
 and an integer~$k$,
 the question is whether there are at least~$k$ vertices
 that are pairwise not adjacent.

\begin{construction}
 \label{const:nphard:stars:cap}
 Let~$I=(G=(V,E),k)$ be an instance of \prob{Cubic Independent Set}
 with $m$~edges and $n$~vertices
 (recall that~$2m=3n$).
 We construct an instance~$I'$ of \simevacAcr{}
 on star~$G'=(V',E')$ with center vertex~$v^*$  as follows
 (see~\Cref{fig:nphard:star:cap} for an illustration).
 \begin{figure}[t]
  \centering
  \begin{tikzpicture}
    \tikzpreamble{}
    \def\xr{0.96}
    \def\yr{1}
    \def\xsh{4.56}%

    \def\n{7} %
		\def\r{1.525} %
		\def\offset{20} %

    \newcommand{\mylbl}[3]{\node at (#1*\xr,#2*\yr)[]{(#3)};}

    \newcommand{\thestar}{
			\node (c) at (0,0)[xcnode]{$C$};

			\foreach \i in {1,...,\n} {
					\pgfmathsetmacro{\angle}{360/\n * (\i - 1) - 141}
					\pgfmathsetmacro{\angleA}{\angle - \offset}
					\pgfmathsetmacro{\angleB}{\angle + \offset}

					\coordinate (A) at (\angleA:\r);
					\coordinate (B) at (\angleB:\r);
					\coordinate (M) at (\angle:\r);
					\coordinate (ML) at (\angle:1.2*\r);

					\draw[-, fill=gray!10!white, draw=none, rounded corners]
							(c) -- (A)
							arc[start angle=\angleA, end angle=\angleB, radius=\r]
							-- (c);

					\node (x\i) at (M)[xcnode]{$2$};
					\ifnum\i=1 %
						\node at (ML)[]{$v$};
						\tikzESd{c/x\i/{$\trtd{4}{19}$}}
					\fi
					\ifnum\i=2
						\node at (ML)[anchor=south east]{$b^i_j$\phantom{b}};%
						\tikzESd{c/x\i/{$\trtd{i}{i+1}$}}
					\fi
					\ifnum\i=3
						\node at (ML)[]{$v_e^\dagger$};
						\tikzESd{c/x\i/{$\trtd{7}{8}$}}
					\fi
					\ifnum\i=4
						\node at (ML)[]{$v_e$};
						\tikzESd{c/x\i/{$\trtd{6}{27}$}}
					\fi
					\ifnum\i=5
						\node at (ML)[anchor=north west]{$v_e'$};
						\tikzESd{c/x\i/{$\trtd{0}{7}$}}
					\fi
					\ifnum\i=6
						\node at (ML)[anchor=north east]{$v_e''$};
						\tikzESd{c/x\i/{$\trtd{0}{20}$}}
					\fi
					\ifnum\i=7
						\node at (ML)[]{$v'$};
						\tikzESd{c/x\i/{$\trtd{0}{1}$}}
					\fi
			}
			\node  at (c)[xcnode, label={[xshift=1pt,label distance=0.0pt]90:$v^*$}]{$C$};
		}

    \begin{scope}[xshift=\xsh*\xr cm]
		 \mylbl{-1.25*\r}{0.95*\r}{a}
     \thestar{}
     \draw[xpathA] ($(x7)+(0*\xr,-0.1*\yr)$) to ($(c)+(-0.3*\xr,-0.05*\yr)$) to ($(x1)+(0*\xr,+0.1*\yr)$);
     \draw[xpathB] ($(x1)+(0.05*\xr,-0.1*\yr)$) to ($(c)+(0.1*\xr,0*\yr)$) to ($(x4)+(0*\xr,-0.1*\yr)$);
     \draw[xpathC] ($(x2)+(0.1*\xr,0*\yr)$) to ($(c)+(0.1*\xr,-0.2*\yr)$);
     \node at ($(x7)+(0.2*\xr,-0.5*\yr)$)[xtypeA]{$P_v$};
     \node at ($(x4)+(-0.3*\xr,-0.6*\yr)$)[xtypeB]{$P_{v,e}$};
     \node at ($(x2)+(0.5*\xr,0.2*\yr)$)[xtypeC]{$P_{b_i^j}$};
    \end{scope}

    \begin{scope}[xshift=2*\xsh*\xr cm]
     \mylbl{-1.25*\r}{0.95*\r}{b}
     \thestar{}
     \draw[xpathA] ($(x3)+(0*\xr,0.1*\yr)$) to ($(c)+(0.25*\xr,-0.1*\yr)$) to ($(x4)+(0*\xr,-0.1*\yr)$);
     \draw[xpathB] ($(x4)+(-0.05*\xr,0.225*\yr)$) to ($(c)+(0.225*\xr,0.25*\yr)$) to ($(x5)+(0.1*\xr,0*\yr)$);
     \draw[xpathC] ($(x4)+(0*\xr,0.1*\yr)$) to ($(c)+(0.1*\xr,0.1*\yr)$) to ($(x6)+(0.1*\xr,0*\yr)$);
     \node at ($(x3)+(0.1*\xr,0.6*\yr)$)[xtypeA]{$P_e^\dagger$};
     \node at ($(x4)+(-0.4*\xr,0.45*\yr)$)[xtypeB]{$P_e'$};
     \node at ($(x6)+(0.65*\xr,-0.2*\yr)$)[xtypeC]{$P_e''$};
    \end{scope}

	\end{tikzpicture}
	\caption{Illustration to~\cref{const:nphard:stars:cap}.
	For each edge~$e$ we indicate the traversal time and deadline as~$\trtd{\trt(e)}{d(e)}$.
	Here,
	we picked $v$ and~$e$ with~$v\in e$ for illustration.}
	\label{fig:nphard:star:cap}
 \end{figure}
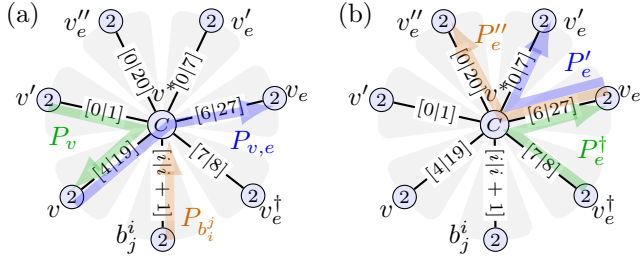

 Let~$\tau=27$.
 For each vertex~$v\in V$,
 add a vertex~$v$ and~$v'$ to~$V'$.
 Make~$v$ adjacent with~$v^*$ with traversal time 4 and deadline 19.
 Make~$v'$ adjacent with~$v^*$ with traversal time 0 and deadline 1.
 For every edge~$e\in E$,
 add the vertex~$v_e$ to~$V'$ and
 make it adjacent with~$v^*$ with travel time 6 and deadline 27,
 add the vertex~$v_e'$ to~$V'$ and
 make it adjacent with~$v^*$ with travel time 0 and deadline 7,
 add the vertex~$v_e''$ to~$V'$ and
 make it adjacent with~$v^*$ with travel time 0 and deadline 20,
 add the vertex~$v_e^\dagger$ to~$V'$ and
 make it adjacent with~$v^*$ with travel time 7 and deadline 8,
 For each~$i\in\set{\tau}$,
 add the set~$B_i\ceq \{b_1^i,\dots,b_{n_i}^i\}$ of \emph{blocker} vertices to~$V'$,
 where~$n_i$ is defined in~\Cref{tab:ni},
 and
 make each vertex from~$B_i$ adjacent with~$v^*$ with travel time~$i$ and deadline~$i+1$
 (thus,
 every path with source from~$B_i$ must start at time step 1 and arrive at time step~$i+1$ at vertex~$v^*$).
 Set~$B\ceq \bigcup_{i=1}^\tau B_i$.
 Set~$c(v^*)=C$ with~$C\ceq 2m+3n+4k$
 and for every node~$x\in V'\setminus \{v^*\}$,
 set~$c(x)=2$ (this enables
 every temporal edge-disjoint temporalization
 to respect each leaf's capacity constraint,
 since at every time step,
 at most one path departs and at most one path arrives in a leaf).
 \begin{table}[t]
  \centering
  \caption{The values for~$n_i$ in \cref{const:nphard:stars:cap}.
  Since each~$n_i$ is of the form~$n_i = C - n_i'$, we give the~$n_i'$ values that correspond to the number of paths that can be located  on~$v^*$ at time step~$i$ next to the paths starting in all vertices from~$B_i$.}
  \label{tab:ni}
		\begin{tabular}{@{}ll || ll || ll@{}}
		\toprule
		$i$ & $C-n_i$ & $i$ & $C-n_i$ & $i$ & $C-n_i$ \\
		\midrule
		1 & $n$ & 7 & $4k+m$ & 18 & $2(n-k)$ \\
		$2,3,4$ & $k$ & 8 & $m$ & 19 & $3(n-k)$ \\
		5 & $2k$ & $9,\dots,16$ & $0$ & 20 & $3(n-k)+m$ \\
		6 & $3k$ & 17 & $n-k$ & $21,\dots,27$ & $3(n-k)$ \\
		\bottomrule
		\end{tabular}
 \end{table}

 For the path set,
 for each~$v\in V$,
 add the \emph{vertex} path~$P_v = (v',v^*,v)$.
 For each~$e\in E$,
 add paths~$P_e'=(v_e,v^*,v_e')$,
 $P_e''=(v_e,v^*,v_e'')$,
 and~$P_e^\dagger=(v_e^\dagger,v^*,v_e)$,
 and for each~$v\in e$,
 add the \emph{verifier} path~$P_{v,e}=(v,v^*,v_e)$.
 For each~$b\in B$,
 add the \emph{blocker} path~$P_b = (b,v^*)$.
\end{construction}

We get the following directly from the construction.

\begin{observation}
 \label{obs:nphard:stars:cap}
 Let~$I'$ be a \yes-instance.
 Then,
 for every solution,
 the following hold:
 \begin{inparaenum}[(i)]
  \item Path~$P_v$ for every~$v\in V$ arrives at time step~1 at~$v^*$
 and at least $n-k$ of the vertex paths leave~$v^*$ also on time step 1.
  \item Path~$P_e^\dagger$ for every~$e\in E$ starts at time step~$1$ and arrives and leaves~$v^*$ at time step~$8$.\label{obs:nphard:stars:cap:dagger}
  \item Path~$P_e'$ for every~$e\in E$ starts at time step~$1$ and arrives and leaves~$v^*$ at time step~$7$.\label{obs:nphard:stars:cap:prime}
  \item Every path that arrives on~$v^*$ latest at time step~$7$ departs from~$v^*$ latest on time step~$7$.\label{obs:nphard:stars:cap:gone}
 \end{inparaenum}
\end{observation}

Assume~$I'$ to be a \yes-instance and consider an arbitrary solution.
For every~$e\in E$,
path~$P_e''$ must start latest at time step~14 by construction.
Indeed,
due to \partref{obs:nphard:stars:cap}{prime} \& \eqref{obs:nphard:stars:cap:dagger}
and the fact that only~$m$ paths can be located at~$v^*$ at time step 8,
path~$P_e''$ must start exactly at time step~$14$.

\begin{lemma}%
 \label{obs:nphard:stars:cap:blocker:edge}
 Let~$I'$ be a \yes-instance.
 Then,
 for every solution and for every~$e\in E$,
 path~$P_e''$ starts at time step~$14$ and arrives and leaves~$v^*$ at time step~$20$.
\end{lemma}

\begin{proof}
 Note that path~$P_e''$ for each~$e\in E$ cannot depart later than~$14$
 since edge~$\{v_e,v^*\}$ has traversal time~6 and edge~$\{v^*,v_e''\}$ has deadline 20.
 Suppose towards a contradiction that there is a path~$P_e''$ that departs from~$v_e$ at~$t<14$.
 Due to~\partref{obs:nphard:stars:cap}{dagger},
 we know that
 $P_e^\dagger$ leaves~$v^*$ to~$v_e$ at time step~$8$.
 This implies that~$t$ must be at most~$2$,
 otherwise they collide
 (recall that edge~$\{v_e,v^*\}$ has traversal time~$6$);
 It is exactly~$2$ since by \partref{obs:nphard:stars:cap}{prime},
 we know that $P_e'$ leaves~$v_e$ at time step~$1$.
 Hence,
 next to all path~$P_e^\dagger$ arriving at~$v^*$ at time step~$8$,
 also~$P_e'$ arrives at time step~$8$ at~$v^*$,
 implying that more than~$m$ paths are located at~$v^*$ at time step~$8$.
 Since the capacity of~$v^*$ allows at most~$m$ paths next to the blocker paths at time step~$8$,
 this yields a contradiction.
\end{proof}

By \cref{obs:nphard:stars:cap}(i),
verifier paths for at most~$k$ different vertices from~$V$ can arrive at~$v^*$ latest at time step 7.
Now,
basically due to \cref{obs:nphard:stars:cap:blocker:edge},
it follows that these are exactly $3k$ verifier paths:
If not,
then more than~$3(n-k)$ verifier path must be located at~$v^*$ at time step~19,
which would violate the capacity constraint of~$v^*$.

\begin{lemma}%
 \label{obs:nphard:stars:cap:kverifier}
 Let~$I'$ be a \yes-instance.
 Then,
 for every solution,
 from exactly~$k$ different vertices all $3k$ verifier paths
 arrive at~$v^*$ latest at time step~$7$.
\end{lemma}

\begin{proof}
 By \cref{obs:nphard:stars:cap}(i),
 at most~$3k$ verifier paths from at most~$k$ different vertices from~$v\in V$
 arrive at~$v^*$ latest at time step~$7$:
 $n-k$ of the vertex paths leave~$v^*$ at time step 1 to its destination,
 and hence,
 the corresponding edge cannot be used by a verifier path before time step 5.
 Suppose towards a contradiction that there are~$k'<3k$ many verifier paths
 arriving at~$v^*$ latest at time step~$7$.
 Then,
 due to the blocker paths and \partref{obs:nphard:stars:cap}{dagger},
 no verifier path can arrive at~$v^*$ at time steps~$8,\dots,16$.
 Since the deadline of edge~$\{v,v^*\}$ for each~$v\in V$ is 19,
 all of the remaining verifier paths must arrive at time steps~$17$, $18$, and~$19$.
 Due to~\cref{obs:nphard:stars:cap:blocker:edge},
 we know that none of these verifier paths can leave~$v^*$ at any of these time steps.
 Hence,
 at time step~$19$,
 there are~$3n-k'>3n-3k$ verifier paths located at~$v^*$.
 Since~$v^*$ has only a capacity of~$3(n-k)$ for paths additional to the blocker paths,
 this yields a contradiction.
\end{proof}

\begin{proof}[Proof of~\cref{thm:nphard:stars:cap}]
 Let~$I'$ be the instance obtained from instance~$I=(G,k)$ with~$G=(V,E)$ of~\prob{Cubic Independent Set}
 using~\cref{const:nphard:stars:cap},
 which can be done in time polynomial in~$|I|$.
 We prove that~$I$ is a \yes-instance if and only if~$I'$ is a \yes-instance.

 \RD{}
 Let~$W\subseteq V$ be an independent set of size~$k$.
 We route all paths addressed in~\cref{obs:nphard:stars:cap}\eqref{obs:nphard:stars:cap:prime}--\eqref{obs:nphard:stars:cap:dagger}
 and \cref{obs:nphard:stars:cap:blocker:edge}
 as described therein.
 Path~$P_v$ for every~$v\in V\setminus W$ arrives and leaves~$v^*$ at time step~$1$,
 and for every~$v\in W$ arrives at $v^*$ at time step~$1$ and leaves~$v^*$ at time step~$7$.
 Also,
 for every~$v\in W$,
 the verifier paths for~$v$ start
 (in arbitrary order)
 at time steps~$1$, $2$, and~$3$,
 and stay at~$v^*$ until time step~$7$.
 Since~$W$ is an independent set,
 for every two paths~$P_{v,e}$ and~$P_{w,e'}$ with~$v,w\in W$,
 $v\in e$, and~$w\in e'$,
 we have that~$e\neq e'$.
 Thus,
 all verifier paths for vertices in~$W$ leave at time step~$7$
 (to a different sink).
 The remaining $3(n-k)$ verifier paths start at time steps~$13$--$15$,
 arrive at~$v^*$ at $17$--$19$,
 and from time step 20 onwards,
 successively and earliest possible leave to their sinks.

 \LD{}
 Let~$\TER$ be a valid temporalization to~$I'$.
 Due to~\cref{obs:nphard:stars:cap:kverifier},
 exactly~$3k$ verifier paths from a set~$W\subseteq V$ of exactly~$k$ vertices
 arrive latest at time step~$7$ at~$v^*$.
 Due
 to \partref{obs:nphard:stars:cap}{gone},
 they all leave~$v^*$ at time step~$7$.
 We claim that~$W$ is an independent set.
 Suppose towards a contradiction that there are~$u,v\in W$ such that~$e=\{u,v\}\in E$.
 Then the verifier paths~$P_{v,e}$ and~$P_{u,e}$
 depart~$v^*$ towards~$v_e$ at the same time step~$7$,
 a contradiction.
\end{proof}

\subsubsection{Constant capacities and general deadlines}
\label{sec:nphard:stars:cap:one}

\begin{theorem}%
 \label{thm:nphard:stars:cap:one}
 \simevacAcr{} is \NP-hard even on decaying stars with maximum capacity one.
\end{theorem}

For \cref{thm:nphard:stars:cap:one} we reduce
from the \NP-hard~\cite{Karp72} problem \prob{Vertex Cover},
where,
given an undirected graph~$G=(V,E)$ and an integer~$k$,
the question is whether there is a set~$W\subseteq V$ with~$|W|\leq k$ such that every edge~$e\in E$ has an endpoint in~$W$,
i.e.,
$e\cap W\neq\emptyset$.
Similar to the proof of \cref{thm:nphard:stars:cap},
from the input instance
of \prob{Vertex Cover},
for each~$v\in V$ and~$e\in E$ there are leaves~$v$ and~$v_e$,
respectively,
in the constructed instance.
Yet,
in contrast,
now there are large traversal times:
for every edge in~$E$ there is a unique traversal time.
Then,
on a high-level,
for every edge~$e\in E$ there are exactly two time windows
for the \emph{verifier} path that starts at~$v_e$ and terminates at~$v$ with~$v\in e$:
one window of size~$|E|$ that they all share,
and another unique window of size one.
By this,
for each edge one of its verifier paths must arrive in the first window
(this corresponds to the covering).
Additionally constructed leaves and paths make only~$k$ vertices available as endpoints of these verifier paths in the first window,
witnessing
a vertex cover of size~$k$.

We give a polynomial-time many-one reduction from the \NP-hard~\cite{Karp72} \prob{Vertex Cover} problem,
where,
given an undirected graph~$G=(V,E)$ and an integer~$k$,
the question is whether there is a set~$W\subseteq V$ with~$|W|\leq k$ such that every edge~$e\in E$ has an endpoint in~$W$,
i.e.,
$e\cap W\neq\emptyset$.

\begin{construction}
 \label{constr:nphard:stars:cap:one}
 Let~$I=(G,k)$ with $G=(V,E)$ be an instance of~\prob{Vertex Cover}
 with~$n=|V|$ and~$m=|E|$.
 We construct an instance~$I'$ of \simevacAcr{} as follows
 on decaying star~$\DG=(V',E',\emptyset,c,\trt,d,\tau)$ with center vertex~$v^*$ and all vertex capacities equal to one
 (see~\Cref{fig:nphard:star:cap:one} for an illustration).
 \begin{figure}[t]
  \centering
  \begin{tikzpicture}
    \tikzpreamble{}
    \def\xr{0.96}
    \def\yr{1}
    \def\xsh{7}

    \def\n{5} %
	\def\r{2.75} %
	\def\offset{20} %

    \newcommand{\mylbl}[3]{\node at (#1*\xr,#2*\yr)[]{(#3)};}

    \newcommand{\thestar}{
			\node (c) at (0,0)[xcnode]{$1$};

			\foreach \i in {1,...,\n} {
					\pgfmathsetmacro{\angle}{360/\n * (\i - 1) - 162}
					\pgfmathsetmacro{\angleA}{\angle - \offset}
					\pgfmathsetmacro{\angleB}{\angle + \offset}

					\coordinate (A) at (\angleA:\r);
					\coordinate (B) at (\angleB:\r);
					\coordinate (M) at (\angle:\r);
					\coordinate (ML) at (\angle:1.14*\r);

					\draw[-, fill=gray!10!white, draw=none, rounded corners]
							(c) -- (A)
							arc[start angle=\angleA, end angle=\angleB, radius=\r]
							-- (c);

					\node (x\i) at (M)[xcnode]{$1$};
					\ifnum\i=1 %
						\node at (ML)[]{$v'$};
						\tikzESd{c/x\i/{$\trtd{0}{4m+k}$}}
					\fi
					\ifnum\i=2
						\node at (ML)[]{\phantom{B}$b_j$};
						\tikzESd{c/x\i/{$[[m+n-k+j]]$}}
					\fi
					\ifnum\i=3
						\node at (ML)[]{$v_{e_i}'$};
						\tikzESd{c/x\i/{$[[5m+k-i]]$}}
					\fi
					\ifnum\i=4
						\node at (ML)[]{$v_{e_i}$};
						\tikzESd{c/x\i/{$\trtd{i}{5m+k+2+i}$}}
					\fi
					\ifnum\i=5
						\node at (ML)[]{$v$};
						\tikzESd{c/x\i/{$\trtd{m+1}{\tau}$}}
					\fi
			}
			\node  at (c)[xcnode, label={[xshift=1pt,label distance=0.0pt]90:$v^*$}]{$1$};
		}

    \begin{scope}[xshift=\xsh*\xr cm]
     \thestar{}
     \draw[xpathA] ($(x4)+(-0.2*\xr,0.0*\yr)$) to ($(c)+(-0.0*\xr,0.2*\yr)$) to ($(x5)+(0.2*\xr,0.0*\yr)$);
     \draw[xpathB] ($(x2)+(0.2*\xr,0*\yr)$) to ($(c)+(0.2*\xr,-0.2*\yr)$);
     \draw[xpathC] ($(x5)+(0*\xr,-0.2*\yr)$) to ($(c)+(-0.2*\xr,0.1*\yr)$) to ($(x1)+(0*\xr,0.2*\yr)$);
     \draw[xpathD] ($(x3)+(0*\xr,0.2*\yr)$) to ($(c)+(0.2*\xr,0.1*\yr)$) to ($(x4)+(0.1*\xr,-0.2*\yr)$);
    \end{scope}

	\end{tikzpicture}
	\caption{Illustration to~\cref{constr:nphard:stars:cap:one}.
	For each edge~$e$ we indicate the traversal time and deadline as~$\trtd{\trt(e)}{d(e)}$
	and short as~$[[\trt(e)]]$ if~$d(e)=\trt(e)+1$
	(i.e., any path on such an edge has to depart on time step 1).
	Here,
	we picked $v$ and~$e_i$ with~$v\in e_i$ for illustration.}
	\label{fig:nphard:star:cap:one}
 \end{figure}
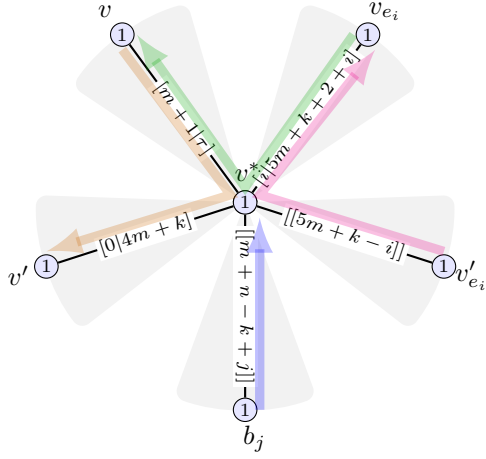
 For each vertex~$v\in V$,
 add two vertices~$v$ and~$v'$ to~$V'$ adjacent to~$v^*$.
 Let the edges~$E=\{e_1,\dots,e_m\}$ be (arbitrarily) enumerated.
 For each edge~$e_i\in E$,
 add two vertices~$v_{e_i}$ and $v_{e_i}'$ to~$V'$ adjacent to~$v^*$.
 Finally,
 add~$b_j$ for every~$j\in\{1,\dots,3m-n+k-1\}$ to~$V'$ adjacent to~$v^*$.
 For each~$v\in V$,
 edge~$\{v,v^*\}$ has traversal time~$m+1$ and deadline~$\tau\ceq 7m+k+3$
 and edge~$\{v',v^*\}$ has traversal time~$0$ and deadline~$4m+k$.
 For each edge~$e_i\in E$,
 edge~$\{v_{e_i},v^*\}$ has traversal time~$i$ and deadline~$5m+k+2+i$
 and edge~$\{v_{e_i}',v^*\}$ has traversal time~$5m+k-i$ and deadline~$5m+k-i+1$.
 For each~$j\in\{1,\dots,3m-n+k-1\}$,
 edge~$\{b_j, v^*\}$ has traversal time~$m+n-k+j$ and deadline~$m+n-k+j+1$.
 For each~$v\in V$,
 we add the path~$P(v,v')$ (which we call \emph{vertex} path).
 For each edge~$e_i$,
 we add the path~$P(v_{e_i}',v_{e_i})$ (which we call the \emph{edge} path of~$e_i$) and for each~$v\in e_i$,
 we add the path~$P(v_{e_i},v)$ (which we call the \emph{verifier} paths).
 For each~$j\in\{1,\dots,3m-n+k-1\}$,
 we add the path~$P(b_j,v^*)$ (which we call \emph{blocker} path).
\end{construction}

\begin{observation}
 Let~$I'$ be a \yes-instance.
 Then,
 at each time step~$m+n-k+2,\dots,4m$,
 vertex~$v^*$ is occupied by a blocker path.
\end{observation}

\begin{observation}
 \label{cor:nphard:star:cap:one:vertexpaths}
 Let~$I'$ be a \yes-instance.
 Then,
 the vertex paths can only be located at~$v^*$ at the time steps~$m+2,\dots,m+n-k+1$
 and~$4m+1,\dots,4m+k$.
 Thus,
 at each of these time steps,
 exactly one vertex path is located at~$v^*$.
 Moreover,
 exactly~$n-k$ paths start from their source to~$v^*$ before time step~$m+1$.
\end{observation}

\begin{lemma}
 \label{obs:occbyedgepaths}
 Let~$I'$ be a \yes-instance.
 Then $v^*$ is occupied by edge paths at all time steps~$4m+k+1,\dots,5m+k$
 and every edge path arrives at its destination earliest at~$5m+k+1$.
\end{lemma}

\begin{proof}
 For each~$i\in\set{m}$,
 the edge path for~$e_i$ arrives at~$v^*$ at time step~$5m+k-i+1$ by construction,
 and if immediately departing,
 arriving at~$v_{e_i}$ at time step~$5m+k+1$
 (recall that the edge~$\{v^*,v_{e_i}\}$ has traversal time~$i$).
\end{proof}

\begin{lemma}
 \label{obs:nphard:stars:cap:one:verif}
 Let~$I'$ be a \yes-instance.
 Then for each~$i\in\set{m}$,
 the verifier paths of edge~$e_i$ can only be located at~$v^*$ at the time steps~$2,\dots,m+1$
 and at time step~$5m+k+2+i$.
 Hence,
 there must be at least one verifier path of edge~$e_i$ be located at~$v^*$ at the time steps~$2,\dots,m+1$.
\end{lemma}

\begin{proof}
 The traversal time of edge~$\{v_{e_i},v^*\}$ equals~$i$ for each~$i\in\set{m}$.
 Thus, the earliest time a verifier path can be located at~$v^*$ is 2.
 Due to \cref{obs:occbyedgepaths},
 $v^*$ is occupied for all time steps~$4m+k+1,\dots,5m+k$.
 Suppose towards a contradiction that~$P(v_{e_i},v)$ for some~$i$ and some~$v\in e_i$ arrives at time step~$t\in\{5m+k+1,\dots,5m+k+1+i\}$ at~$v^*$.
 The path~$P(v_{e_i}, v_{e_i}')$ is located at~$5m+k-i+1$ at~$v^*$
 and arrives earliest at~$v_{e_i}$ at time step~$5m+k+1$.

 We first show that each arrives exactly at time step~$5m+k+1$.
 If~$i>1$,
 then $P(v_{e_i}, v_{e_i}')$ arrives at~$v_{e_i}$ exactly at time step~$5m+k+1$:
 By construction,
 $P(v_{e_{i-1}}, v_{e_{i-1}}')$ arrives at~$v^*$ just one time step after $P(v_{e_i}, v_{e_i}')$ does,
 and hence,
 since~$v^*$ has capacity of only~$1$,
 $P(v_{e_i}, v_{e_i}')$ must leave~$v^*$ at time step~$5m+k-i+1$.
 Next consider~$i=1$.
 the edge path of~$e_1$ arrives at~$v^*$ at time step~$5m+k$.
 Recall that edge~$\{v_{e_1},v^*\}$ has traversal time~$1$ and deadline~$5m+k+2+1$.
 Since the edge path is already located at~$v^*$,
 no verifier path from~$v_{e_1}$ can depart
 (since~$v^*$ has capacity~$1$).
 Thus,
 assume the edge path departs from~$v^*$ at time step~$5m+k+j$, $j\geq 1$,
 then it arrives at $v_{e_1}$ at time step~$5m+k+j+1\geq 5m+k+2$.
 Since~$v_{e_i}$ has capacity~$1$,
 any verifier path can leave~$v_{e_i}$ earliest at time step~$5m+k+3$,
 which is too late due to the traversal time of~$1$;
 Thus,
 either the two paths meet on the edge or in~$v_{e_i}$,
 which has capacity of only one;
 a contradiction to~$I'$ being a \yes-instance.

 It follows that each path~$P(v_{e_i}, v_{e_i}')$ arrives exactly at~$v_{e_i}$ at time step~$5m+k+1$.
 Thus,
 when $P(v_{e_i},v)$
 arrives at time step~$t\in\{5m+k+1,\dots,5m+k+1+i\}$ at~$v^*$,
 it meets $P(v_{e_i}, v_{e_i}')$ either at $v_{e_i}$ or on the edge~$\{v^*,v_{e_i}\}$;
 a contradiction to~$I'$ being a \yes-instance.
\end{proof}

\begin{lemma}
 \label{obs:nphard:stars:cap:one:exactlyone}
 Let~$I'$ be a \yes-instance.
 Then for every~$i\in\set{m}$,
 there is exactly one verifier path of edge~$e_i$ located at~$v^*$ at one of the time steps~$2,\dots,m+1$.
\end{lemma}

\begin{proof}
 Suppose towards a contradiction that the statement is not true,
 that is,
 there are two verifier paths of edge~$e_i$ located at~$v^*$ at the time steps~$2,\dots,m+1$
 (the only remaining case due to \cref{obs:nphard:stars:cap:one:verif}).
 By the pigeon hole principle,
 it follows that there is some~$j\in\set{m}$ such that no verifier path of edge~$e_j$ is located at~$v^*$ at one of the time steps~$2,\dots,m+1$.
 Then,
 both verifier paths for~$e_j$ can only be located at~$v^*$ at time step~$5m+k+2+j$;
 a contradiction to~$I'$ being a \yes-instance.
\end{proof}

\begin{proof}[Proof of~\cref{thm:nphard:stars:cap:one}]
 Let~$I'=(\DG,\Pset)$ be the instance of \simevacAcr{}
 on decaying star~$\DG=(V',E',\emptyset,c,\trt,d,\tau)$ with center vertex~$v^*$
 obtained via \cref{constr:nphard:stars:cap:one}
 from a given instance $I=(G,k)$ of~\prob{Vertex Cover}
 with $G=(V,E)$, $n=|V|$, and~$m=|E|$.
 We claim that~$I$ is a \yes-instance if and only if~$I'$ is a \yes-instance.

 \RD{}
 Let~$W\subseteq V$ be a vertex cover of~$G$ of size~$k$.
 We construct a valid temporalization~$\TER$ as follows.
 Let~$\bar{W}=V\setminus W$ and let~$\sigma\colon W\to\set{|W|}$ and~$\bar\sigma\colon \bar{W}\to\set{|\bar{W}|}$ be two bijections.
 Let $\TER(P(b_j,v^*),b_j)=1$ for each~$j\in\{1,\dots,3m-n+k-1\}$.
 Recall that these paths arrive exactly at all time steps~$m+n-k+2,\dots,4m$.
 Let~$e_i=\{v,w\}$.
 Let~$x_i$ be the smallest element in~$e_i\cap W$ with respect to~$\sigma$
 and $y_i\ceq e_i\setminus \{x_i\}$.
 Set
 \begin{align*}
  \TER(P(v_{e_i},x_i),v_{e_i}) &= 1, \\
 \TER(P(v_{e_i},x_i),v^*) &= 1+i, \\
 \TER(P(v_{e_i},y_i),v_{e_i}) &= 5m+k+2, \\
 \TER(P(v_{e_i},y_i),v^*) &= 5m+k+2+i,\text{ and }\\
 \TER(P(v_{e_i}',v_{e_i}),v_{e_i}') &= 1,\\
 \TER(P(v_{e_i}',v_{e_i}),v^*) &= 5m+k-i+1.
 \end{align*}
 Thus,
 these paths are located at~$v^*$ at each time step~$2,\dots,m+1$ in the first phase,
 and at each time step~$4m+k+1,\dots,5m+k$ and~$5m+k+3,\dots,6m+k+2$
 in the second phase.
 Also note that the paths of the first phase
 arrive at their endpoints at the time steps~$m+2,\dots,2m+1$.
 For~$v\in \bar{W}$,
 let~$\TER(P(v,v'),v)=\bar{\sigma}(v)$
 and~$\TER(P(v,v'),v^*)=\bar{\sigma}(v)+m+1$.
 Recall that~$|\bar{W}|=n-k$,
 and thus all these paths arrive and depart exactly at the time steps~$m+2,\dots,m+n-k+1$.
 For~$v\in W$,
 let~$\TER(P(v,v'),v)=3m-1+\sigma(v)$
 (this is possible,
 since all verifier paths of the first phase arrive latest at~$2m+1$ at their endpoints)
 and~$\TER(P(v,v'),v^*)=4m+\sigma(v)$
 (this is possible,
 since the last blocker path arrives at time step~$4m$).
 Recall that~$|W|=k$,
 and thus all these paths arrive and depart exactly at the time steps~$4m+1,\dots,4m+k$.

 \LD{}
 Let~$\TER$ be a valid assignment.
 We claim that~$W=\{v\in V\mid \TER(P(v,v'),v^*)\geq 4m+1\}$
 is a vertex cover of~$G$.
 First note that~$|W|=k$ due to~\cref{cor:nphard:star:cap:one:vertexpaths}.
 Further,
 note that by~\cref{obs:nphard:stars:cap:one:exactlyone},
 we know that exactly one verifier path for each edge with endpoint in~$W$
 arrive at and depart from~$v^*$ at some time step~$2,\dots,m+1$
 (again, due to~\cref{cor:nphard:star:cap:one:vertexpaths},
 none of their endpoints can be in~$V\setminus W$).
 Thus,
 we have that every edge has one endpoint in~$W$,
 proving that~$W$ is a vertex cover.
\end{proof}

\subsection{On uncapacitated trees}
\label{sec:nphard:trees}

\begin{theorem}
\label{thm:nphard:trees}
\simevacAcr{} is \NP-hard even on uncapacitated,
\exogenous{}
decaying trees with lifetime four.
\end{theorem}

We give a polynomial-time many-one reduction from the following \NP-hard~\cite{DarmannD20} problem \prob{(2,2)-3SAT}:
Given a set~$X=\{x_1,\dots,x_N\}$ of variables and a boolean CNF-formula~$\phi=\bigland_{j=1}^M C_j$ over~$X$ with clauses $C_1,\dots,C_M$,
each consisting of exactly three distinct variables,
such that
every variable appears exactly twice negated and exactly twice unnegated
in~$\phi$,
the question is whether
there is a truth assignment~$\truthass\colon X\to\{\true,\false\}$ that satisfies~$\phi$?

\newcommand{\varapp}[4]{
	\node (#1A) at (#3*\xr+-2*\xr,-1*\yr)[xnode,label={[label distance=-3pt]{225-#2*90}:{$v_{i,#2}$}}]{};
	\node (#1T) at (#3*\xr+-3.5*\xr,-2*\yr)[xnode]{};
	\node (#1F) at (#3*\xr+-0.5*\xr,-2*\yr)[xnode]{};
	\ifnum#4<1
		\node at (#1T)[label=180:{$t_{i,#2}$}]{};
		\node at (#1F)[label=0:{$f_{i,#2}$}]{};
		\node (#1T1) at (#3*\xr+-3.5*\xr,-3.5*\yr)[xnode,label=180:{$t_{i,#2}'$\,\,}]{};
		\node (#1T2) at (#3*\xr+-3.5*\xr,-5*\yr)[xnode,label=0:{$t_{i,#2}''$}]{};
		\node (#1F1) at (#3*\xr+-0.5*\xr,-3.5*\yr)[xnode,label=0:{\phantom{f}$f_{i,#2}'$}]{};
		\node (#1F2) at (#3*\xr+-0.5*\xr,-5*\yr)[xnode,label=180:{$f_{i,#2}''$}]{};
		\tikzESd{#1T/#1T1/$2$,
			#1T1/#1T2/1,
			#1F/#1F1/$2$,
			#1F1/#1F2/1%
		}
	\else
		\node at (#1T)[label={[label distance=-2pt]-90:{$t_{i,#2}$}}]{};
		\node at (#1F)[label={[label distance=-2pt]-90:{$f_{i,#2}$}}]{};
	\fi

	\tikzES{z/#1A,
			#1A/#1T,
			#1A/#1F%
	}

	\ifnum#4=1
		\node (x) at (-3*\xr,0.875*\yr)[xnode,label=90:{$u_i$}]{};
		\node (y) at (3*\xr,0.875*\yr)[xnode,label=90:{$w_i$}]{};
		\tikzES{z/x,z/y}
	\fi

	\ifnum#4=0
		\node (x) at (-1.75*\xr,0.875*\yr)[xnode,label={[label distance=2pt]90:{$u_i$}}]{};
		\node (y) at (1.75*\xr,0.875*\yr)[xnode,label={[label distance=2pt]90:{$w_i$}}]{};
		\tikzES{z/x,z/y}
		\node (x1) at (-3.5*\xr,0.75*\yr)[xnode,label={$u_i'$}]{};
		\node (x2) at (-5.25*\xr,0.625*\yr)[xnode,label={$u_i''$}]{};
		\node (y1) at (3.5*\xr,0.75*\yr)[xnode,label={$w_i'$}]{};
		\node (y2) at (5.25*\xr,0.625*\yr)[xnode,label={$w_i''$}]{};
		\tikzESd{x/x1/$3$,
				y/y1/$3$,
				x1/x2/1,
				y1/y2/1%
		}
	\fi
}

\begin{construction}
\label{constr:nphard:trees}
Let~$I=(X,\phi)$ be an instance of~\prob{(2,2)-3SAT}
with~$X=\{x_1,\dots,x_N\}$ and~$\phi=\bigland_{j=1}^M C_j$.
Construct an instance~$I'=(\DG,\Pset)$ with~$\DG=(V,E,\emptyset,c,\trt,d,\tau)$ of~\simevacAcr{}
with~$\tau=4$ and uncapacitated~$c$
as follows
(see~\Cref{fig:nphard:trees:A,fig:nphard:trees:B} for an illustration).

\begin{figure}[t!]
	\centering
	\begin{tikzpicture}
	\tikzpreamble{}
	\def\xr{0.7375}
	\def\yr{0.5}
	\newcommand{\mylbl}[3]{\node at (#1*\xr,#2*\yr)[]{(#3)};}

	\begin{scope}
		\mylbl{-5.75}{1.75}{a}
		\node (z) at (0,1*\yr)[xnode,label={$z$}]{};
		\varapp{one}{1}{-1}{0}
		\varapp{two}{2}{5}{0}

		\draw[xpathB] ($(oneT1)+(0.25*\xr,0*\yr)$) to ($(oneT)+(0.25*\xr,0*\yr)$) to node[xemark]{$1\times$}($(oneA)+(-0.0*\xr,-0.15*\yr)$) to ($(oneF)+(-0.25*\xr,0*\yr)$) to ($(oneF1)+(-0.25*\xr,0*\yr)$);
		\draw[xpathA] ($(oneT)+(-0.25*\xr,-0.0*\yr)$) to node[xemark]{$1\times$}($(oneT2)+(-0.25*\xr,0*\yr)$);
		\draw[xpathA] ($(oneF)+(0.25*\xr,-0.0*\yr)$) to node[xemark]{$1\times$}($(oneF2)+(0.25*\xr,0*\yr)$);

		\draw[xpathB] ($(twoT1)+(0.25*\xr,0*\yr)$) to ($(twoT)+(0.25*\xr,0*\yr)$) to node[xemark]{$1\times$}($(twoA)+(0.0*\xr,-0.15*\yr)$) to ($(twoF)+(-0.25*\xr,0*\yr)$) to ($(twoF1)+(-0.25*\xr,0*\yr)$);
		\draw[xpathA] ($(twoT)+(-0.25*\xr,-0.0*\yr)$) to node[xemark]{$1\times$}($(twoT2)+(-0.25*\xr,0*\yr)$);
		\draw[xpathA] ($(twoF)+(0.25*\xr,-0.0*\yr)$) to node[xemark]{$1\times$}($(twoF2)+(0.25*\xr,0*\yr)$);

		\draw[xpathA] ($(x)+(0*\xr,-0.25*\yr)$) to node[xemark]{$1\times$}($(x2)+(0*\xr,-0.25*\yr)$);
		\draw[xpathA] ($(y)+(0*\xr,-0.25*\yr)$) to node[xemark]{$1\times$}($(y2)+(0*\xr,-0.25*\yr)$);
		\draw[xpathB] ($(x1)+(0*\xr,0.25*\yr)$) to ($(x)+(0*\xr,0.25*\yr)$) to node[xemark]{$2\times$}($(z)+(0.0*\xr,0.25*\yr)$) to ($(y1)+(0*\xr,0.25*\yr)$);

	\end{scope}

		\end{tikzpicture}
		\caption{Illustration to the construction regarding~$V_i$ and its blocker paths. Edges with deadlines~$<4$ are labeled.}
		\label{fig:nphard:trees:A}
	\end{figure}
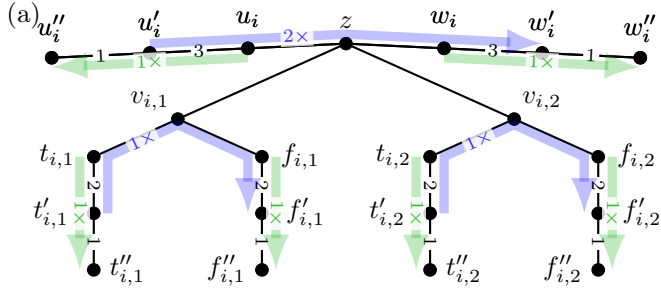

	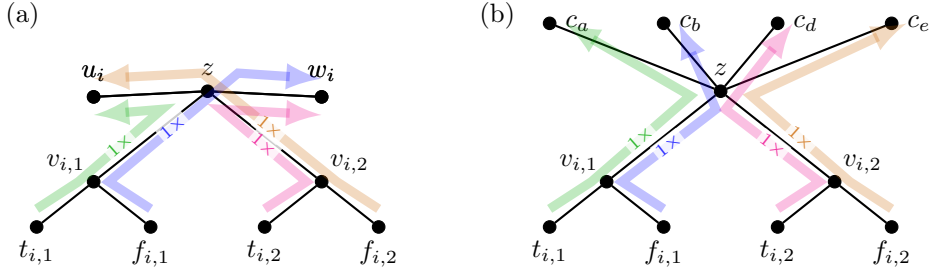
\begin{figure}[t!]
		\centering
		\begin{tikzpicture}
		\tikzpreamble{}
		\def\xr{0.3375}
		\def\yr{0.4}
		\def\xpd{0.4}
		\newcommand{\mylbl}[3]{\node at (#1*\xr,#2*\yr)[]{(#3)};}

		\begin{scope}[xshift=0*\xr cm, yshift=0*\yr cm]
		\def\xr{0.5}
		\def\yr{0.6}
		\def\xpd{0.4}
		\mylbl{-4.875}{2.675}{a}
		\node (z) at (0,1*\yr)[xnode,label={$z$}]{};
		\varapp{one}{1}{-1}{1}
		\varapp{two}{2}{5}{1}

		\draw[xpathA] ($(oneT)+(0.0*\xr,\xpd*\yr)$)
		to ($(oneA)+(-\xpd*\xr,0*\yr)$)
		to node[xemark]{$1\times$}($(z)+(-3*\xpd*\xr,-\xpd*\yr)$) to ($(x)+(0*\xr,-\xpd*\yr)$);
		\draw[xpathB] ($(oneF)+(0*\xr,\xpd*\yr)$) to ($(oneA)+(\xpd*\xr,0*\yr)$) to node[xemark]{$1\times$}($(z)+(2*\xpd*\xr,\xpd*\yr)$) to ($(y)+(-0.0*\xr,\xpd*\yr)$);
		\draw[xpathC] ($(twoF)+(0*\xr,\xpd*\yr)$)
		to ($(twoA)+(\xpd*\xr,0*\yr)$)
		to node[xemark]{$1\times$}($(z)+(-0.35*\xpd*\xr,\xpd*\yr)$) to ($(x)+(0.0*\xr,\xpd*\yr)$);
		\draw[xpathD] ($(twoT)+(0.0*\xr,\xpd*\yr)$) to ($(twoA)+(-\xpd*\xr,0*\yr)$) to node[xemark]{$1\times$}($(z)+(0.35*\xr,-\xpd*\yr)$) to ($(y)+(0*\xr,-\xpd*\yr)$);

	\end{scope}

		\begin{scope}[xshift=20*\xr cm, yshift=0*\yr cm]
		\def\xr{0.5}
		\def\yr{0.6}
		\def\xpd{0.4}
		\mylbl{-5.875}{2.675}{b}
		\node (z) at (0,1*\yr)[xnode,label={$z$}]{};
		\varapp{one}{1}{-1}{2}
		\varapp{two}{2}{5}{2}

		\newcommand{\cnodes}[2]{%
			\node (#1c) at (#2*\xr,2.5*\yr)[xnode,label=0:{$c_{#1}$}]{};
			\tikzES{z/#1c}

		}

		\cnodes{a}{-4.5}
		\cnodes{b}{-1.5}
		\cnodes{d}{1.5}
		\cnodes{e}{4.5}

		\draw[xpathA] ($(oneT)+(0.0*\xr,\xpd*\yr)$) to ($(oneA)+(-\xpd*\xr,0*\yr)$) to node[xemark]{$1\times$}($(z)+(-2*\xpd*\xr,-0.1*\yr)$) to ($(ac)+(\xpd*\xr,0*\yr)$);
		\draw[xpathB] ($(oneF)+(0*\xr,\xpd*\yr)$) to ($(oneA)+(\xpd*\xr,0*\yr)$) to node[xemark]{$1\times$}($(z)+(-0.125*\xr,-\xpd*\yr)$) to ($(bc)+(\xpd*\xr,0*\yr)$);
		\draw[xpathC] ($(twoF)+(0*\xr,\xpd*\yr)$) to ($(twoA)+(\xpd*\xr,0*\yr)$) to node[xemark]{$1\times$}($(z)+(2*\xpd*\xr,-0.1*\yr)$) to ($(ec)+(\xpd*\xr,0*\yr)$);
		\draw[xpathD] ($(twoT)+(0.0*\xr,\xpd*\yr)$) to ($(twoA)+(-\xpd*\xr,0*\yr)$) to node[xemark]{$1\times$}($(z)+(0.125*\xr,-\xpd*\yr)$) to ($(dc)+(\xpd*\xr,0*\yr)$);

		\end{scope}

		\end{tikzpicture}
		\caption{Illustration to the construction for (a) validation paths
		and (b) verifier paths.}
		\label{fig:nphard:trees:B}

	\end{figure}

Let~$V=\{z\}\cup \bigcup_{i=1}^N V_i\cup \bigcup_{j=1}^M \{c_j\}$,
where $V_i = \{u_i,u_i',u_i''\}\cup \{w_i,w_i',w_i''\}\cup \bigcup_{j\in\{1,2\}}\{\{v_{i,j}\}\cup\{t_{i,j},t_{i,j}',t_{i,j}''\}\cup\{f_{i,j},f_{i,j}',f_{i,j}''\}\}$
for each~$i\in\set{N}$.

	\smallskip\noindent\emph{Edge set.}
Construct edge set~$E$ with deadlines~$d$ as follows,
where every edge~$e\in E$ has traversal time~$\trt(e)=0$:

For each $i\in\set{N}$ (see~\Cref{fig:nphard:trees:A}),
add edges~$\{z,u_i\}$, $\{z,w_i\}$ with deadline~$4$,
$\{u_i,u_i'\}$, $\{w_i,w_i'\}$ with deadline~$3$,
and $\{u_i',u_i''\}$, $\{w_i',w_i''\}$ with deadline~$1$;
For each~$j\in\{1,2\}$,
add
$\{z,v_{i,j}\}$ with deadline~$4$,
$\{t_{i,j},t_{i,j}'\}$, $\{f_{i,j},f_{i,j}'\}$ with deadline~$2$,
and $\{t_{i,j}',t_{i,j}''\}$, $\{f_{i,j}',f_{i,j}''\}$ with deadline~$1$.

For each~$j\in\set{M}$ (see~\Cref{fig:nphard:trees:B}(b)),
add edge~$\{z,c_j\}$ with deadline~$4$.

	\smallskip\noindent\emph{Path set.}
Since~$G$ is a decaying tree,
we represent each path as~$P(s,t)$ with source~$s$ and sink~$t$.

For each $i\in\set{N}$,
add the following \emph{blocker} paths (see~\Cref{fig:nphard:trees:A}).
Add~$P(u_i,u_i'')$,
$P(w_i,w_i'')$,
$P_1(u_i',w_i')$ and $P_2(u_i',w_i')$.
For all~$j\in\{1,2\}$,
add
$P(t_{i,j},t_{i,j}'')$, $P(f_{i,j},f_{i,j}'')$,
and
$P(t_{i,j}',f_{i,j}')$.

Next (see~\Cref{fig:nphard:trees:B}(a)), add the \emph{validity} paths~$P(t_{i,1},u_i)$, $P(t_{i,2},w_i)$, $P(f_{i,1},w_i)$, and $P(f_{i,2},u_i)$.

Finally (see~\Cref{fig:nphard:trees:B}(b)),
when variable~$x_i$ first (resp.\ second) unnegated appearance is in clause~$C_a$ (resp.\ $C_d$),
add the \emph{verifier} path~$P(t_{i,1},c_a)$ (resp.\ $P(t_{i,2},c_d)$).
When variable~$x_i$ first (resp.\ second) negated appearance is in clause~$C_b$ (resp.\ $C_e$),
add the verifier path~$P(f_{i,1},c_b)$ (resp.\ $P(f_{i,2},c_e)$).
Finally add the following \emph{verifier} paths (see~\Cref{fig:nphard:trees:B}(b)).
When variable~$x_i$ appears unnegated in clauses~$C_a$ and~$C_d$ with~$a<d$,
add the \emph{verifier} paths~$P(t_{i,1},c_a)$ and $P(t_{i,2},c_d)$.
When variable~$x_i$ appears negated in clauses~$C_b$ and~$C_e$ with~$b<e$,
add the verifier paths~$P(f_{i,1},c_b)$ and $P(f_{i,2},c_e)$.

\end{construction}

By construction,
we immediately get the following.

\begin{observation}
	\label{obs:nphard:trees:blockingpath:A}
	Let~$I'$ be a \yes-instance.
	For every~$i\in\set{N}$,
	on their edges,
	$P(u_i,u_i'')$,
	$P(w_i,w_i'')$,
	$P(t_{i,j},t_{i,j}'')$ and $P(f_{i,j},f_{i,j}'')$ for $j\in\{1,2\}$,
	take exactly the time step~$1$,
	and $P_1(v_{i,j},v_{i,j}'')$ and $P_2(v_{i,j},v_{i,j}'')$
	take exactly the time steps~1 and~2.
\end{observation}

By construction and~\cref{obs:nphard:trees:blockingpath:A},
we now get the following.

\begin{observation}
	\label{obs:nphard:trees:blockingpath}
	Let~$I'$ be a \yes-instance.
	For every~$i\in\set{N}$,
	on their edges,
	$P_k(u_i',w_i')$ take exactly the time steps~$\set[2]{\tau-1}$,
	$P_k(t_{i,j}',f_{i,j}')$, $j\in\{1,2\}$,
	take exactly the time steps~$\set[2]{\tau-2}$,
	and $P_k(v_{i,1}',v_{i,2}')$
	take exactly the time steps~$\set[3]{\tau-2}$.
\end{observation}

By construction,
we immediately get the following.

\begin{observation}
\label{obs:nphard:trees:blockingpath}
Let~$I'$ be a \yes-instance.
For every solution
and
every~$i\in\set{N}$
we have:
\begin{inparaenum}[(i)]
\item on their edges,
$P(u_i,u_i'')$,
$P(w_i,w_i'')$,
take exactly the time step~$1$,
and $P_1(u_i',w_i')$ and~$P_2(u_i',w_i')$ take exactly the time steps~$2$ and~$3$.\label{obs:nphard:trees:blockingpath:A}
\item for every~$j\in\{1,2\}$,
on their edges,
$P(t_{i,j},t_{i,j}'')$,
and $P(f_{i,j},f_{i,j}'')$,
take exactly the time step~$1$,
and $P(t_{i,j}',f_{i,j}')$
take exactly the time step~$2$.\label{obs:nphard:trees:blockingpath:B}
\end{inparaenum}
\end{observation}

Assume~$I'$ to be a \yes-instance and consider an arbitrary solution.
For each~$i\in\set{N}$,
note that by construction,
only one of the two validity paths starting at~$t_{i,j}$ and~$f_{i,j}$,
$j\in\{1,2\}$,
can be located at~$z$ at time step 1.
There must be exactly two validity paths with sinks~$u_i$ and~$w_i$ located at~$z$ at time step 1
since,
by~\partref{obs:nphard:trees:blockingpath}{A},
each of the edges~$\{z,u_i\}$ and~$\{z,w_i\}$
is only available at the two time steps 1 and~$4$
for exactly two out of four validity paths for~$i$ containing it.
We thus have the following.

\begin{lemma}%
\label{obs:nphard:trees:validity}
Let~$I'$ be a \yes-instance.
For every solution and every~$i\in\set{N}$,
on time step~$1$,
edges~$\{z,u_i\}$,
$\{z,w_i\}$,~$\{z,v_{i,1}\}$,
and $\{z,v_{i,2}\}$,
are taken either by
the two validity paths~$P(t_{i,1},u_i)$ and $P(t_{i,2},w_i)$
or by
the two validity paths~$P(f_{i,1},w_i)$ and $P(f_{i,2},u_i)$.
\end{lemma}

\begin{proof}
By~\partref{obs:nphard:trees:blockingpath}{A},
we know that the edges~$\{z,u_i\}$ and $\{z,w_i\}$
are only available at time steps 1 and~$4$.
Thus,
if at time step~1 no path is assigned on one of these edges,
then there is only one time step left for two paths that need to traverse this edge,
yielding a contradiction.
Hence,
at time step 1,
the edges~$\{z,u_i\}$,
$\{z,w_i\}$,
and (thus)
$\{z,v_{i,j}\}$, $j\in\{1,2\}$,
are taken by validity paths.

We claim that for each~$i\in\set{N}$,
this can be established only by the pairs $P(t_{i,1},u_i)$ and $P(t_{i,2},w_i)$ on the one hand,
and $P(f_{i,1},w_i)$ and $P(f_{i,2},u_i)$ on the other hand.
Note that the pairs with same destination $P(t_{i,1},u_i)$  and $P(f_{i,2},u_i)$
or $P(t_{i,2},w_i)$ and $P(f_{i,1},w_i)$ are not possible.
Now suppose that~$P(t_{i,1},u_i)$ and $P(f_{i,1},w_i)$
are on the edges $\{z,u_i\}$ and~$\{z,w_i\}$ at time step~1.
Then,
both must be on edge~$\{z,v_{i,1}\}$
at time step~1,
a contradiction to~$I'$ being a \yes-instance.
The argument works analogously for $P(t_{i,2},w_i)$ and $P(f_{i,2},u_i)$.
Note that
if one pair takes time step~1,
the other pair takes time step~$4$.
\end{proof}

Combining \cref{obs:nphard:trees:validity} with \partref{obs:nphard:trees:blockingpath}{B},
we now get that if there are verifier paths at~$z$ at time step 2,
then they departed at time step~1 either from~$t_{i,1}$ and~$t_{i,2}$
or from~$f_{i,1}$ and~$f_{i,2}$.

\begin{lemma}%
\label{obs:nphard:trees:verif}
Let~$I'$ be a \yes-instance.
For every solution and every~$i\in\set{N}$,
there are at most two verifier paths for~$i$ located at~$z$ at time step 2,
and
if there are exactly two,
then they
start either from~$t_{i,1}$ and~$t_{i,2}$
or from~$f_{i,1}$ and~$f_{i,2}$.
\end{lemma}

\begin{proof}
By~\partref{obs:nphard:trees:blockingpath}{B},
we know that the paths $P(t_{i,1}',f_{i,1}')$ and $P(t_{i,2}',f_{i,2}')$
take exactly the time step~$2$
on their respective edges,
and hence,
every verifier path located at~$z$ at time step at most 2
must have been started at time step~1 from their source~$t_{i,j}$ or~$f_{i,j}$
to~$v_{i,j}$, $j\in\{1,2\}$.
Due to~\cref{obs:nphard:trees:verif},
we know that at least two out of the four edges connecting~$v_{i,j}$ with~$t_{i,j}$ and~$f_{i,j}$,
$j\in\{1,2\}$,
is taken by at least two validity paths.
Thus,
there are at most two verifier paths starting at time step~1
to~$v_{i,j}$,
$j\in\{1,2\}$.
Again by~\cref{obs:nphard:trees:validity},
on time step 1 the edges connecting~$z$ with~$v_{i,j}$,
$j\in\{1,2\}$,
are taken by two validity paths,
which are either
$P(t_{i,1},u_i)$ and $P(t_{i,2},w_i)$,
or $P(f_{i,1},w_i)$ and $P(f_{i,2},u_i)$.
Thus, there is no verifier path located at~$z$ at time step~1.

Now assume that there are two verifier paths
started at time step~1 from their source~$t_{i,j}$ or~$f_{i,j}$
to~$v_{i,j}$, $j\in\{1,2\}$,
and located at~$z$ at time step~2.
If the edges connecting~$z$ with~$v_{i,j}$,
$j\in\{1,2\}$,
are taken at time step~1 by the two validity paths~$P(t_{i,1},u_i)$ and $P(t_{i,2},w_i)$,
then the two verifier paths must have started from~$f_{i,1}$ and~$f_{i,2}$,
since otherwise two paths are assigned to the edge~$\{t_{i,j},v_{i,j}\}$ at the same time for some~$j\in\{1,2\}$,
yielding a contradiction to~$I'$ being a \yes-instance.
Similarly,
if the edges connecting~$z$ with~$v_{i,j}$,
$j\in\{1,2\}$,
are taken at time step~1 by the two validity paths~$P(f_{i,1},w_i)$ and $P(f_{i,2},u_i)$,
then the two verifier paths must have started from~$t_{i,1}$ and~$t_{i,2}$,
since otherwise two paths are assigned to the edge~$\{f_{i,j},v_{i,j}\}$ at the same time for some~$j\in\{1,2\}$,
yielding a contradiction to~$I'$ being a \yes-instance.
\end{proof}

Intuitively,
all verifier paths arriving at~$z$ at time step~$2$ give a satisfying truth assignment because of the following.
By \cref{obs:nphard:trees:verif},
every two validity paths at~$z$ at time step~$2$
cannot correspond to different truth assignments of the same variable.
By \cref{obs:nphard:trees:validity,obs:nphard:trees:blockingpath},
every verifier path can arrive at~$z$ only at time steps 2, 3, and~4.
Each clause vertex is the sink of three verifier paths,
and thus,
one of them must be at~$z$ at time step~2.

\begin{proof}[Proof of~\cref{thm:nphard:trees}]
	Let instance~$I'=(\DG,\Pset)$ of \simevacAcr{}
	with $\DG=(V,E,\emptyset,c,\trt,d,\tau)$ and~$\tau=4$
	be
	obtained from instance~$I=(X,\phi)$ of \prob{(2,2)-3SAT}
	with~$X=\{x_1,\dots,x_N\}$ and~$\phi=\bigland_{j=1}^M C_j$
	using~\cref{constr:nphard:trees} in time polynomial in~$|I|$.
	We claim that
	$I$ is a \yes-instance
	if and only if~$I'$ is a \yes-instance.

\RD{}
Let~$\truthass$ be a truth assignment satisfying~$\phi$.
Construct valid temporalization~$\TER$ as follows.
Route all blocker paths in accordance with~\cref{obs:nphard:trees:blockingpath}.

For each~$i\in\set{N}$,
if~$\truthass(x_i)=\true$,
then let~$S_i^{\rm val}\ceq \{f_{i,1},f_{i,2}\}$ and $S_i^{\rm ver}\ceq \{t_{i,1},t_{i,2}\}$,
and if $\truthass(x_i)=\false$,
then let $S_i^{\rm val}\ceq \{t_{i,1},t_{i,2}\}$ and $S_i^{\rm ver}\ceq \{f_{i,1},f_{i,2}\}$.
Move the validation paths starting in~$S_i^{\rm val}$
over~$z$ to~$u_i$ and~$w_i$,
all on time step 1;
Move the verifier paths starting on~$S_i^{\rm ver}$
on time step 1 to~$v_{i,1}$ and~$v_{i,2}$,
then on time step 2 to~$z$,
and finally
each to its final destinations at time step~$2+q_i$,
where~$q_i$ is the number of verifier paths for all~$i'<i$ located at~$z$ at time step~2 with the same destination.
By this,
for each~$j\in\set{M}$,
at least one verifier path with sink~$c_j$ finished at time step~2,
since each clause is satisfied by at least one variable's assignment.
At time step~$3$,
move the remaining verifier paths to~$z$ and then latest on time step~$4$ to their destination
(any path needs to wait for at most one time step).
At time step~$4$,
move the remaining validation paths to their destinations~$u_i$ and~$w_i$.

\LD{}
Let~$\TER$ be a valid temporalization.
Construct~$\truthass\colon X\to \{\true,\false\}$ as follows.
For each~$i\in\set{N}$,
set
$\truthass(x_i) = \true$
if there is no verifier path for~$i$ at~$z$ at time step~2
or there is a verifier path at~$z$ at time step~2 starting from~$t_{i,j}$,
$j\in\{1,2\}$.
Set
$\truthass(x_i) = \false$
if there is a verifier path at~$z$ at time step~2 starting from~$f_{i,j}$,
$j\in\{1,2\}$.
Note that~$\truthass$ is well-defined due to~\cref{obs:nphard:trees:verif}.
We claim that~$\truthass$ is a satisfying truth assignment.
Suppose towards a contradiction
that there is clause~$C_j$ not satisfied by~$\truthass$.
Then none of the three verifier paths ending in~$c_j$ are located at~$z$ at time step~2.
Thus,
each of the three verifier paths arrives
at~$z$ earliest at time step~$3$
(recall that no verifier path is at~$z$ at time step 1).
Hence,
only two time steps remain
for the three paths to proceed to~$c_j$,
a contradiction.
\end{proof}

We leave open whether \simevacAcr{} is still \NP-hard even on uncapacitated,
\exogenous{}
decaying trees with lifetime 2 and 3.

\section{Integer Linear Programming}
\label{sec:ilp}

\newcommand{\gap}{\gamma}
\newcommand{\ilpcore}{\eqref{ilp:vars}, \eqref{ilp:monoton}}
\newcommand{\capcoA}{\eqref{ilp:capcoA:vars}--\eqref{ilp:capcoA:ysum}}%
\newcommand{\capcoB}{\eqref{ilp:capcoB:vars}, \eqref{ilp:capcoB}, \eqref{ilp:capcoB:sum}}
\newcommand{\RelB}{\overleftarrow{\Rel}}
\newcommand{\Psetin}{\Pset_{\rm in}}

We give an integer linear program (ILP) for~\simevacAcr{}.
Let
\begin{align}
	\forall P\in\Pset, (v,w)\in\Rel(P): && x^P_{v,w} &\in\set{\tau} \label{ilp:vars}
\end{align}
be the time when~$P$ leaves~$v$ towards~$w$.
We first ensure adequacy.
We write $\theta(v,w)$ short for~$\theta((v,w))$ and $\theta(\{v,w\})$
(same for~$d$).
\begin{align}
 \forall P\in\Pset, (u,v),(v,w)\in\Rel(P): &&  x^P_{u,v} + \trt(u,v) &\leq x^P_{v,w} \label{ilp:monoton}
 \\
 \forall P\in\Pset, (v,w)\in\Rel(P): &&  x^P_{v,w} + \trt(v,w) &\leq d(v,w) \label{ilp:deadline}
\end{align}
Next we ensure that any two paths are temporally edge-disjoint.
Let~$\RelB(P)\ceq \{(w,v)\mid (v,w)\in\Rel(P)\}$.
\begin{align}
 \begin{aligned}
 \forall P,P'\in\Pset, (v,w)\in \Rel(P)\cap\Rel(P')&: |x^P_{v,w} - x^{P'}_{v,w}| \geq 1 %
 \\
 \forall P,P'\in\Pset, \{v,w\}\in E \text{~s.t.~}\qquad\qquad\,\,\,\, & \\
 (v,w)\in \Rel(P) \cap \RelB(P'): |x^P_{v,w} - x^{P'}_{w,v}&| \geq \max\{1,\trt(\{v,w\})\} .
 \label{ilp:ted}
 \end{aligned}
\end{align}
For the vertex-capacity constraints,
we follow the idea of a ``sweeping line'' to find the largest clique in an interval graph.
Observe that the time steps at which a path is located at some vertex~$v$
forms an interval.
We introduce the following variables.
\begin{align}
 && \forall v\in V, P,Q\in\Pset(v): &&  \alpha^{P,Q}_{v}, \beta^{P,Q}_{v}, \gamma^{P,Q}_{v} &\in\{0,1\} \label{ilp:capcoB:vars}
\end{align}
We define
the arrival and departure times
for a path~$P$ at its vertex~$v$
with predecessor~$u$ and successor~$w$
(if existent):
$t_{\rm arr}^P(v) = x^P_{u,v} + \trt(u,v)$ if $v\neq \src(P)$,
and $t_{\rm arr}^P(v) = x^P_{v,w}$ otherwise;
$t_{\rm dep}^P(v) = x^P_{v,u}$ if $v\neq \snk(P)$,
and~$t_{\rm dep}^P(v) = x^P_{u,v} + \trt(u,v)$ otherwise.
We ensure next that~$\gamma_v^{P,Q}=1$ if
$Q$ is located at~$v$ when~$P$ arrives at~$v$,
i.e.,
when $t_{\rm arr}^P(v)\in[t_{\rm arr}^Q(v),t_{\rm dep}^Q(v)]$.
Let
$M$ be a sufficiently large number.
\begin{align}
 \begin{aligned}
 \forall v\in V, P,Q\in\Pset(v):\
 t_{\rm arr}^P(v) &\leq t_{\rm arr}^Q(v) - 1 + M\cdot (1 - \alpha^{P,Q}_{v})
 \\
 t_{\rm arr}^P(v) &\geq t_{\rm dep}^Q(v) + 1 - M\cdot (1 - \beta^{P,Q}_{v})
 \\
 1 &= \gamma^{P,Q}_{v} + \alpha^{P,Q}_{v} + \beta^{P,Q}_{v}
 \label{ilp:capcoB}
 \end{aligned}
\end{align}
Note that if~$\gamma^{P,Q}_{v}=\gamma^{P,Q'}_{v}=1$,
then also~$Q$ and~$Q'$ must be intersecting (at least in~$t_{\rm arr}^P(v)$).
Observe that the arrival times are sufficient:
For every vertex~$v$ consider any two paths arriving at~$v$ at time steps~$t_1$ and~$t_2$ such that no other path arrives between~$t_1$ and~$t_2$.
Then the number of paths located at vertex~$v$ is not increasing in $\{t_1,\dots,t_2-1\}$,
since other paths can only depart at these time steps.
Hence,
to ensure that the capacity constraints are respected,
we only need to count all the overlaps for each path via~$\gamma$.
\begin{align}
 && \forall v\in V, P\in\Pset(v):  &&  \sum\nolimits_{Q\in \Pset(v)}\gamma^{P,Q}_{v} &\leq c(v) \label{ilp:capcoB:sum}
\end{align}
Our final ILP for~\simevacAcr{} is composed of
\eqref{ilp:vars}--\eqref{ilp:capcoB:sum}
and provides a valid temporalization for every feasible instance.

\paragraph*{Optimization goal.}

When an extreme event like flooding is predicted to happen,
practically we wish to know
the minimum in-advance starting time~$d^*$
to obtain a feasible evacuation schedule.
To this end,
we only need to adjust~\eqref{ilp:deadline} (and possibly~$M$) in our ILP.
\begin{align}
 \begin{aligned}
 \min \, d^* \text{~s.t.}  && \text{\ilpcore, \eqref{ilp:ted}}&\text{--\eqref{ilp:capcoB:sum}} \text{, and}
 \\
 \forall P\in\Pset, (v,w)\in\Rel(P): &&  x^P_{v,w} + \trt(v,w) &\leq d^* + d(v,w) \label{ilp:dstar}
 \end{aligned}
\end{align}

\begin{remark}
	The optimal~$d^*$ is upper bounded by the number of paths
	times the maximum independent travel time of any path,
	which in turn is upper bounded my the number of vertices and the maximum travel time of any connection.
	Thus,
	we can compute the minimum solution to \eqref{ilp:dstar}
	by polynomially many calls of~\simevacAcr{}.
\end{remark}

\section{Experiments}
\label{sec:experiments}

We implemented the ILP~\eqref{ilp:dstar} using Gurobi(py) 12.0. in two ways,
where we implemented~\eqref{ilp:capcoB}
as outlined,
referred to as \bigM{},
and using Gurobi indicator constraints instead,
referred to as \indc{}.
We set~$M$ for~\bigM{} as follows.
Greedily color paths such that
every two paths of the same color have disjoint vertex sets.
Let~$\Pset_1,\dots,\Pset_{\chi}$ be the resulting partition of the paths with~$\chi$ colors.
Set~$M=\sum_{x\in \set{\chi}} \max_{P\in \Pset_x} (1 + \sum_{(u,v)\in \Rel(P)}\trt(u,v))$,
which corresponds to iteratively,
one color after the other,
initiating a traversal without waiting for all paths from the same color starting at the same time
when the latest path from the preceding color finished.
For both implementation we give a warm start regarding~$d^*$:
for \bigM,
we use the ILP relaxation,
and for \indc{},
we take the largest violated deadline
when every path starts at time step 1 without waiting to its sink.
For each instance,
we let both implementations compete in parallel
and only recorded the runtime of the winner.

Our experiments ran on
Intel(R) Xeon(R) Silver 4310 CPU at 2.10GHz (12 cores),
125 GB RAM,
Ubuntu 22.04.3 LTS (x86\_64).

\subsection{Data}

We conducted experiments on artificial decaying paths and stars
and semi-artificial decaying graphs
where the networks are obtained from ten German cities
each located at at least one river.

\paragraph*{Artificial Data.}
For each number~$n\in\{8,12,16\}$ of vertices,
we create decaying stars and decaying paths.
Here,
for every two vertices that are supposed to be adjacent
(e.g., any leaf and the center for stars,
and~$v_i$ and $v_j$ with~$|i-j|=1$ for paths),
we equally likely connect them either with a single arc,
two bidirectional arcs,
or an undirected edge
(a single arc's direction is also equally likely drawn).
The traversal times are picked uniformly at random from the set~$\{5,\dots,20\}$
(corresponding to street segments of roughly \SI{70}{\meter} to \SI{280}{\meter} at \vkmh{50}).
We add~$p=\frc{p}\cdot n$ many paths to
each decaying path and decaying star,
where~$\frc{p}\in\{0.5,0.67,0.83,1\}$ for decaying paths and~$\frc{p}\in\{0.5,1,1.5,2\}$ for decaying stars.\footnote{Throughout, we round to the next integer if required and not specified otherwise.}
Herein, a source is picked uniformly at random,
and then extended.
For decaying paths,
we additionally restrict the paths length~$\ell=\frc{\ell}\cdot n$ with~$\frc{\ell}\in\{0.33,0.44,0.55,0.66\}$.
From the source,
two auxiliary paths start to the left and right,
each of which stops latest when having length~$\ell$.
Then, we only add the path of larger length,
and a fair coin decides when both have the same length.
For decaying stars,
when the source is a leaf,
a fair coin decides whether the end vertex is the center or another leaf picked uniformly at random from all reachable leaves.
For each vertex~$v\in V$,
we set~$c(v)=\max\{1,\lceil\frc{c}\cdot |\Pset(v)|\rceil\}$ with~$\frc{c}\in \{0.1,0.4,0.7,1\}$
(i.e., $\frc{c}=1$ yields an uncapacitated instance).
For each connection~$e$ we set the deadline~$d(e) = \frc{d}\cdot \ddlb(e)$ with~$\frc{d}\in\{0.2,0.47,0.73,1\}$,
where~$\ddlb(e)$ is $\trt(e)$ plus the maximum of the number of paths using connection~$e$
and the latest arrival time of any path using connection~$e$ when each paths is routed independently and earliest without waiting.
We constructed 1920 instances of the form~$I\_n\_\frc{p}\_\frc{c}\_\frc{d}\text{-}x$ for decaying stars
(referred to by \stars{})
and 7680 instances of the form~$I\_n\_\frc{p}\_\frc{c}\_\frc{d}\_\frc{\ell}\text{-}x$ for decaying paths (referred to by \paths{}),
with~$x\in\{0,1,\dots,9\}$
(i.e., 10 instances per constellation).
See \Cref{fig:artifical:violin} for more details on our constructed instances.

\begin{figure*}[t]
 \centering
 \includegraphics[width=0.475\textwidth]{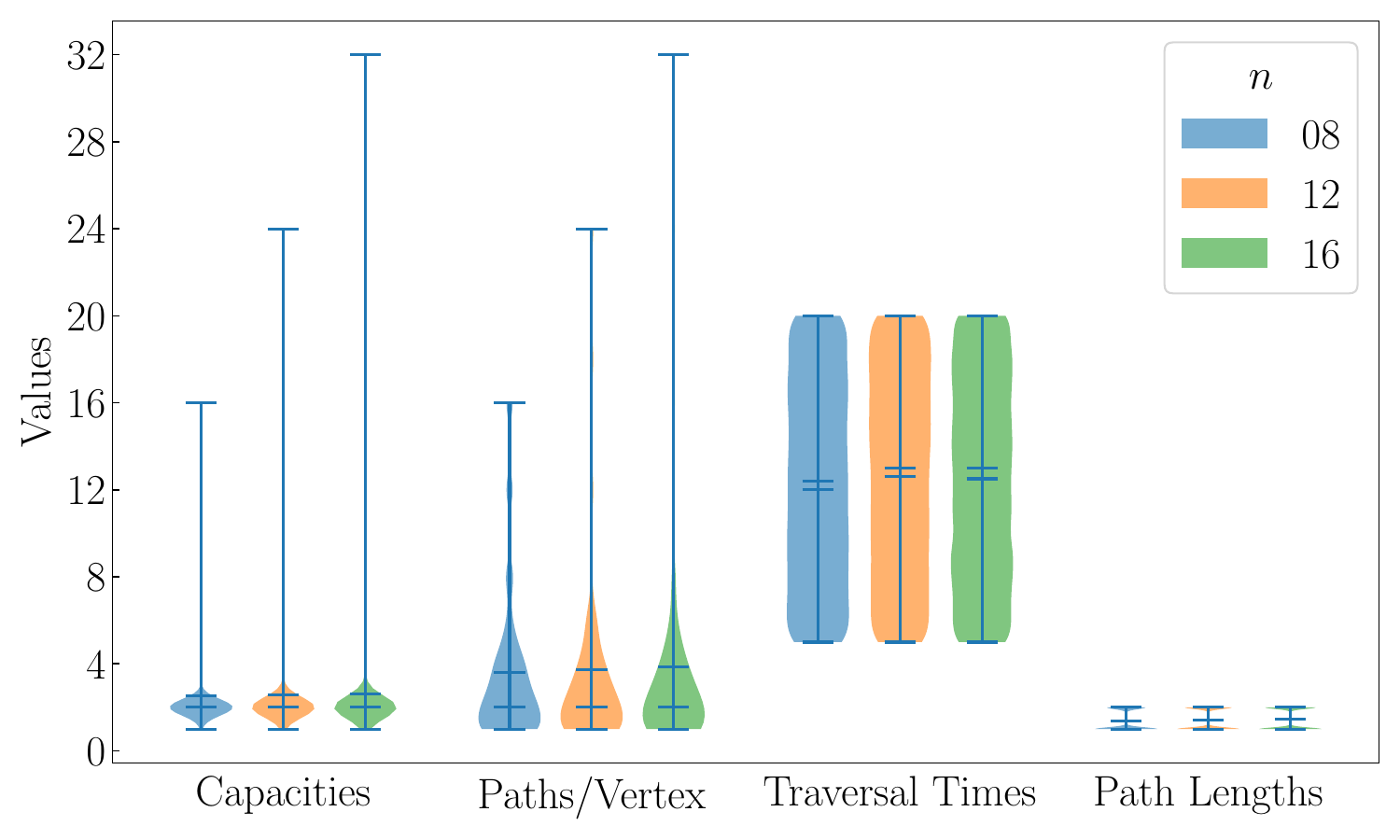}
 \hfill
 \includegraphics[width=0.475\textwidth]{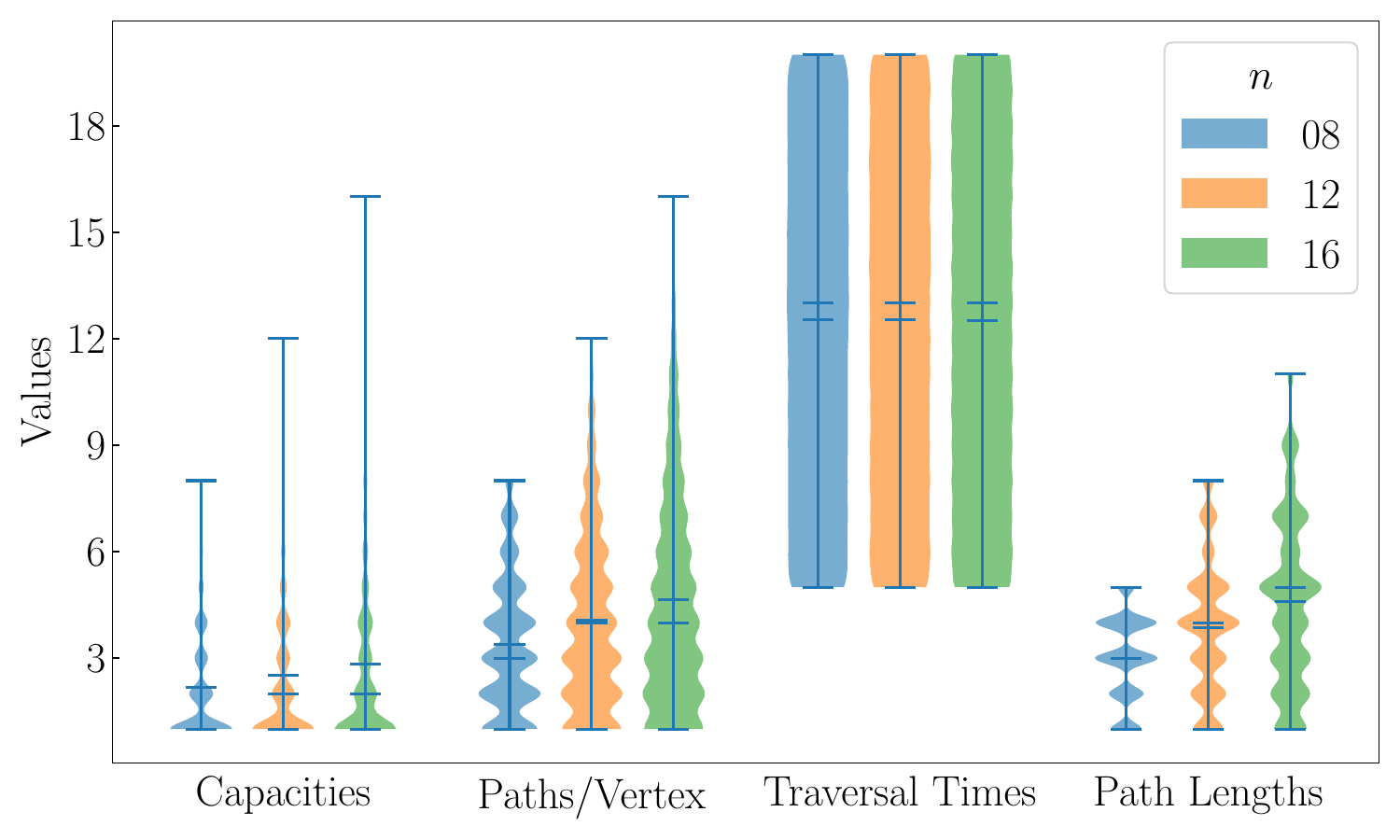}
 \caption{Overview (violin plots) for our artificial instances \stars{} (left) and \paths{} (right). We outlined the distributions of (from left to right) the capacities,
 number of paths per vertex,
 traversal times,
 and lengths of the paths.}
 \label{fig:artifical:violin}
\end{figure*}

\paragraph*{Semi-artificial Data.}
We chose the ten German cities that have,
according to~\cite{GDV2024_HochwasserAdressen},
the highest ratio of addresses under severe threat of flooding
(see~\Cref{fig:osm_german_cities},
\Cref{fig:osm_german_cities_snd} and \Cref{tab:osm:instances}).
\begin{figure}[t]
 \centering
 \includegraphics[width=0.45\textwidth]{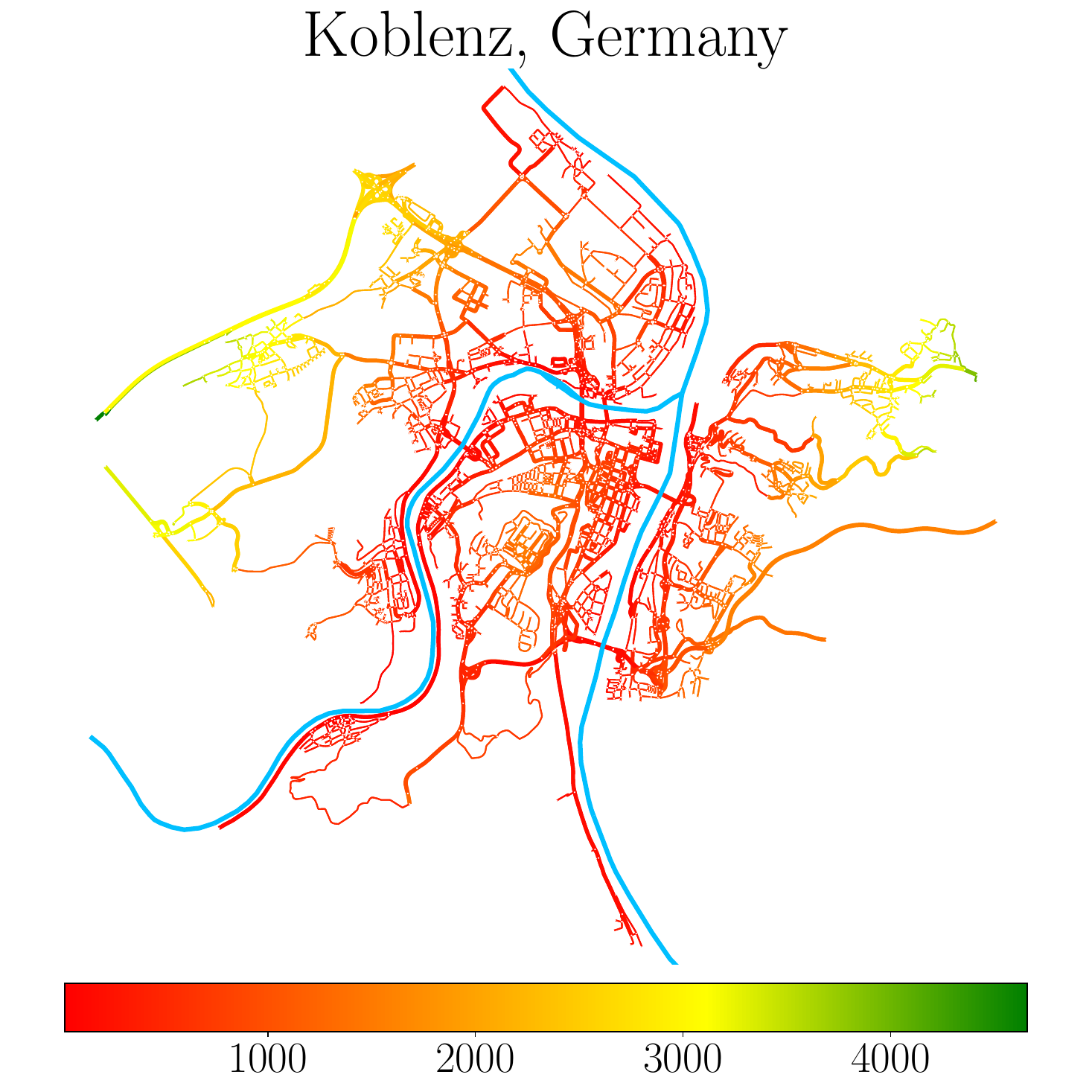}
 \hfill
 \includegraphics[width=0.45\textwidth]{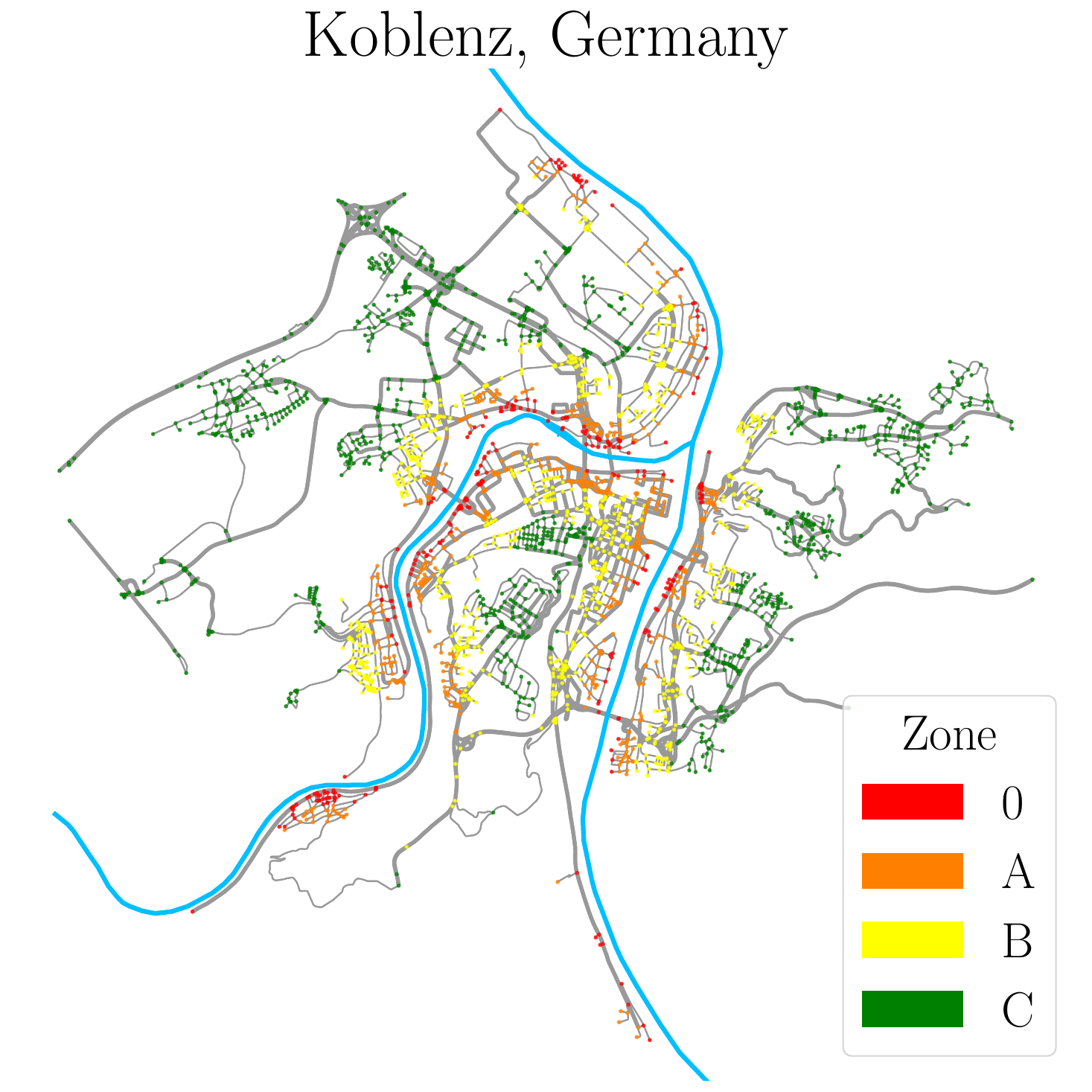}
 \caption{Illustration for one \osm{} with waterways depicted blue and
 color-codings for (left) deadlines and (right) zones.}
 \label{fig:osm_german_cities}
\end{figure}%
\begin{figure*}[t]
 \centering
 \includegraphics[width=0.32\textwidth]{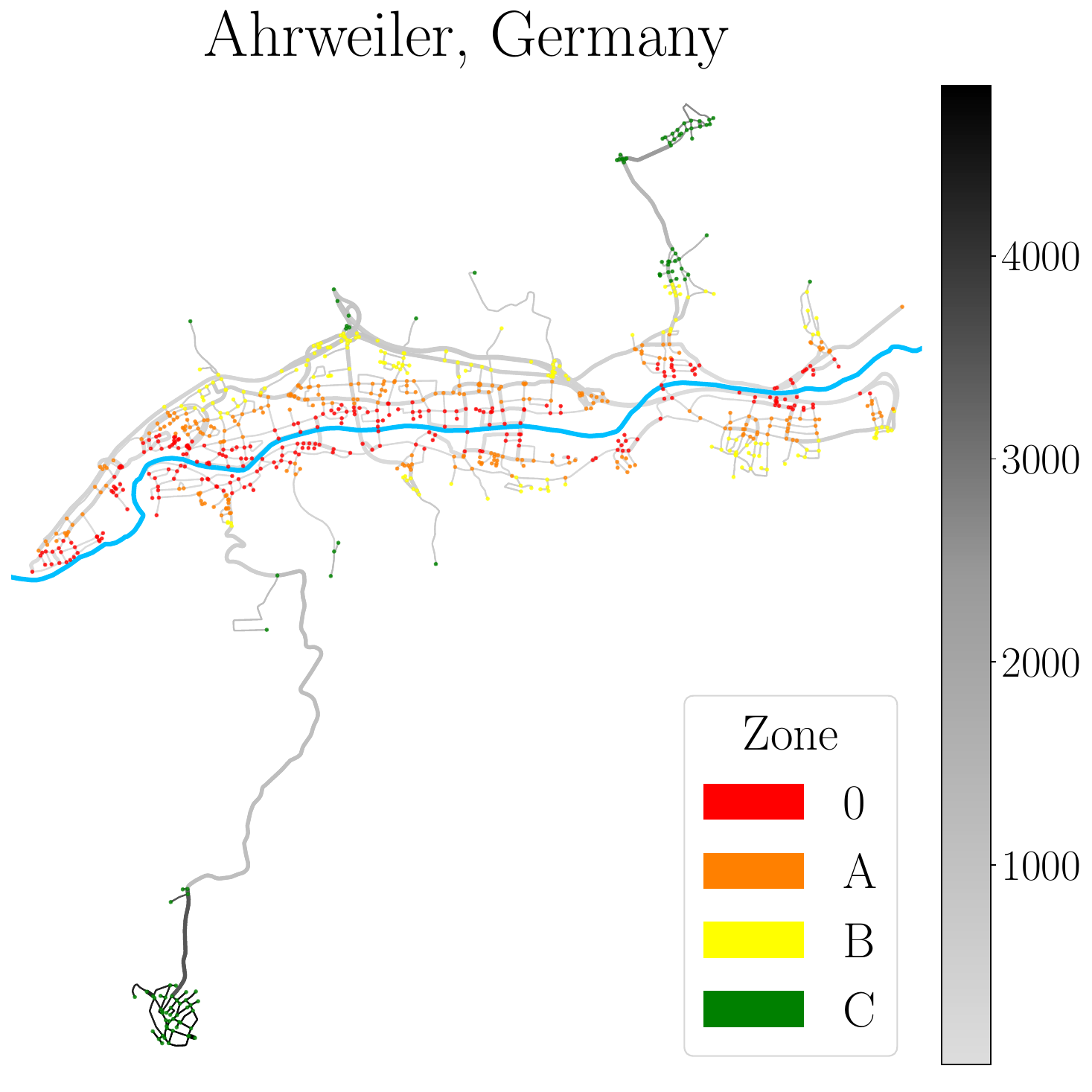}
 \hfill
 \includegraphics[width=0.32\textwidth]{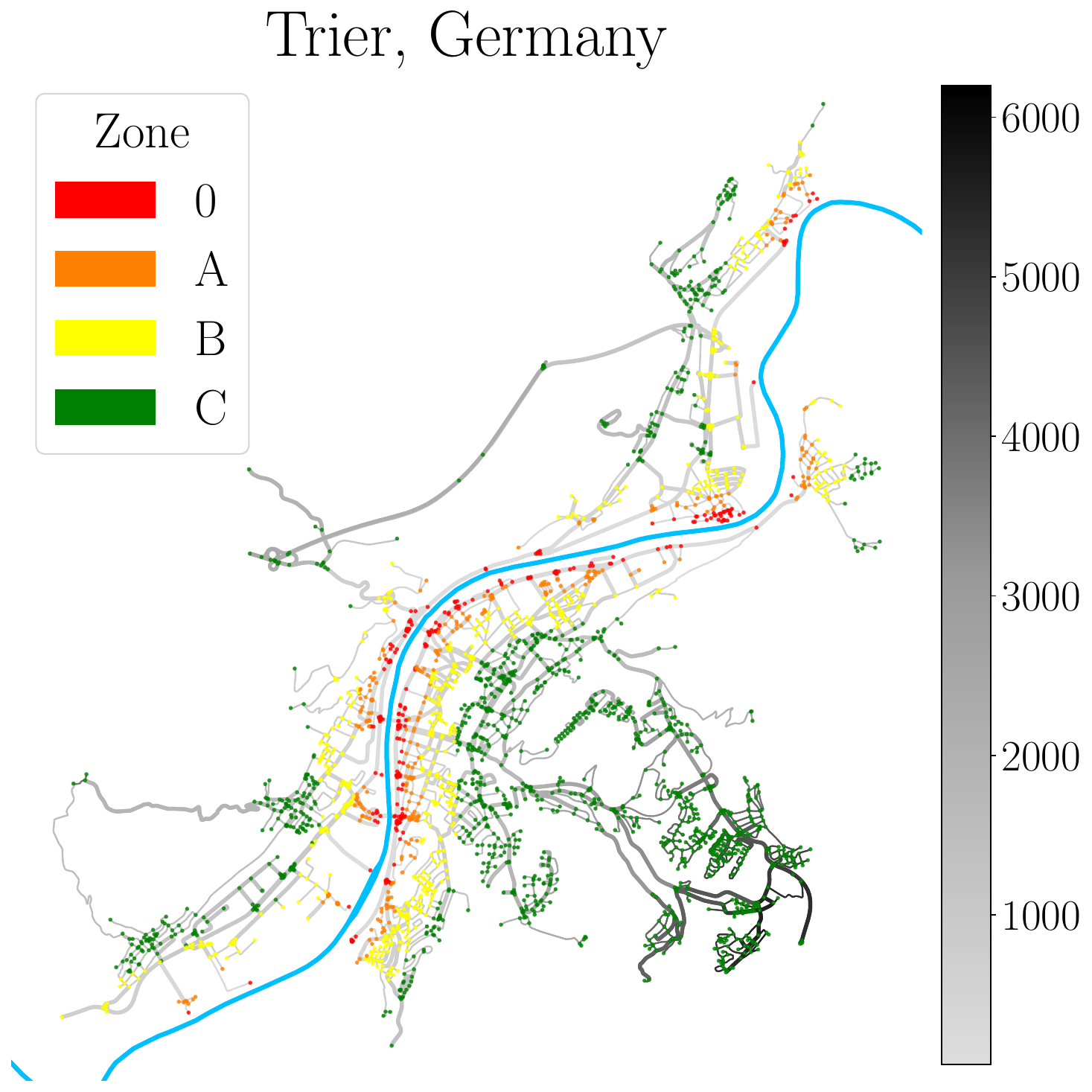}
 \hfill
 \includegraphics[width=0.32\textwidth]{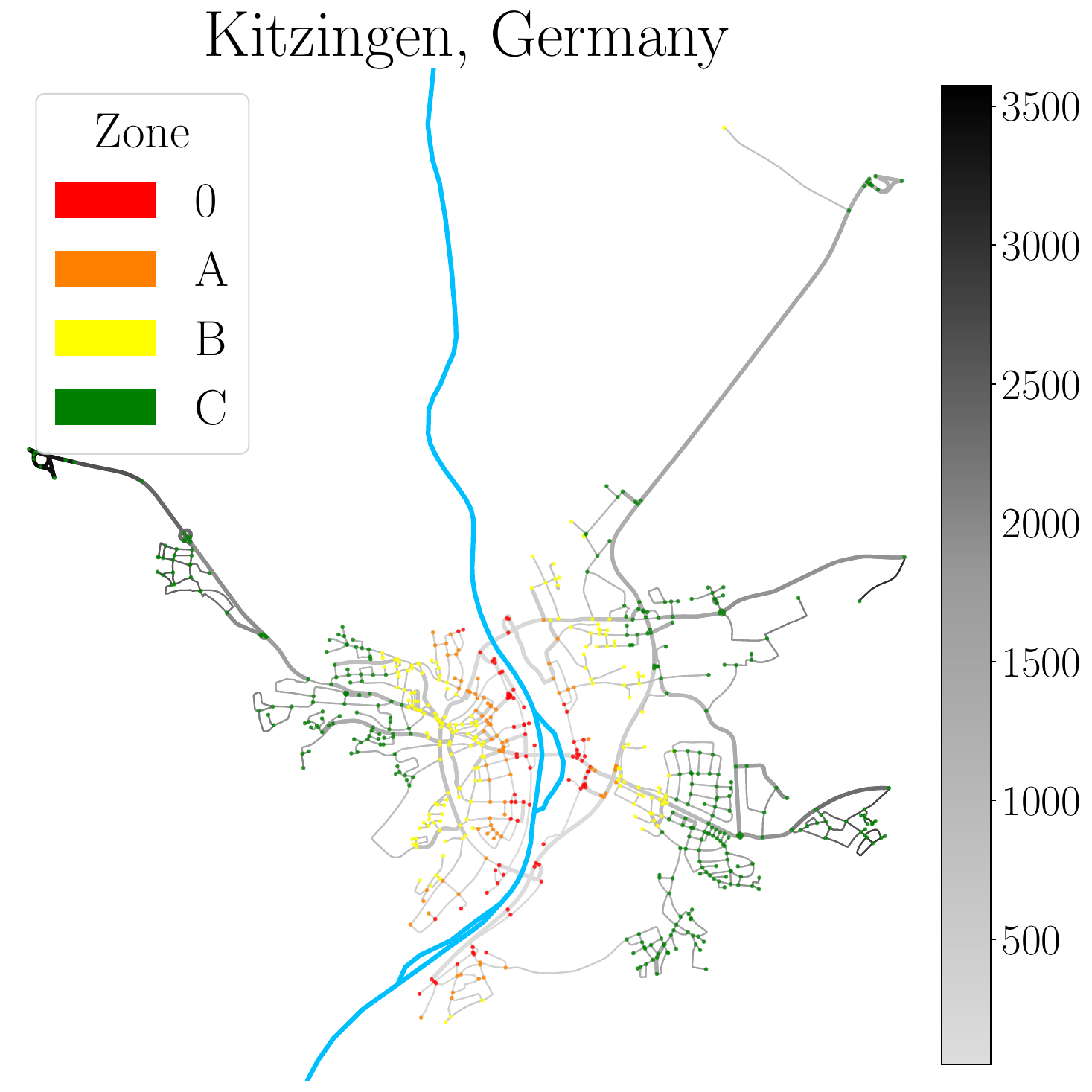}
 \hfill
 \includegraphics[width=0.32\textwidth]{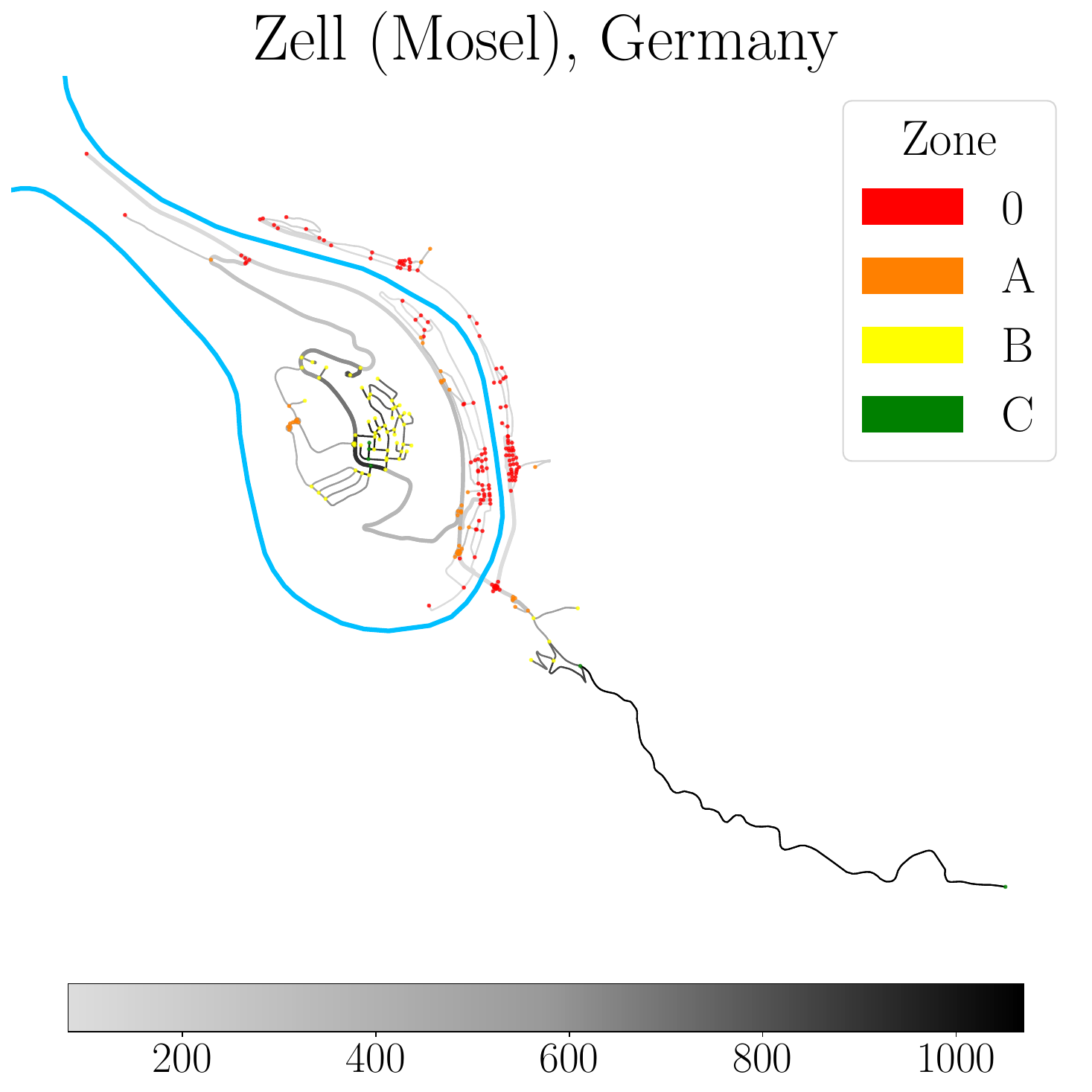}
 \hfill
 \includegraphics[width=0.32\textwidth]{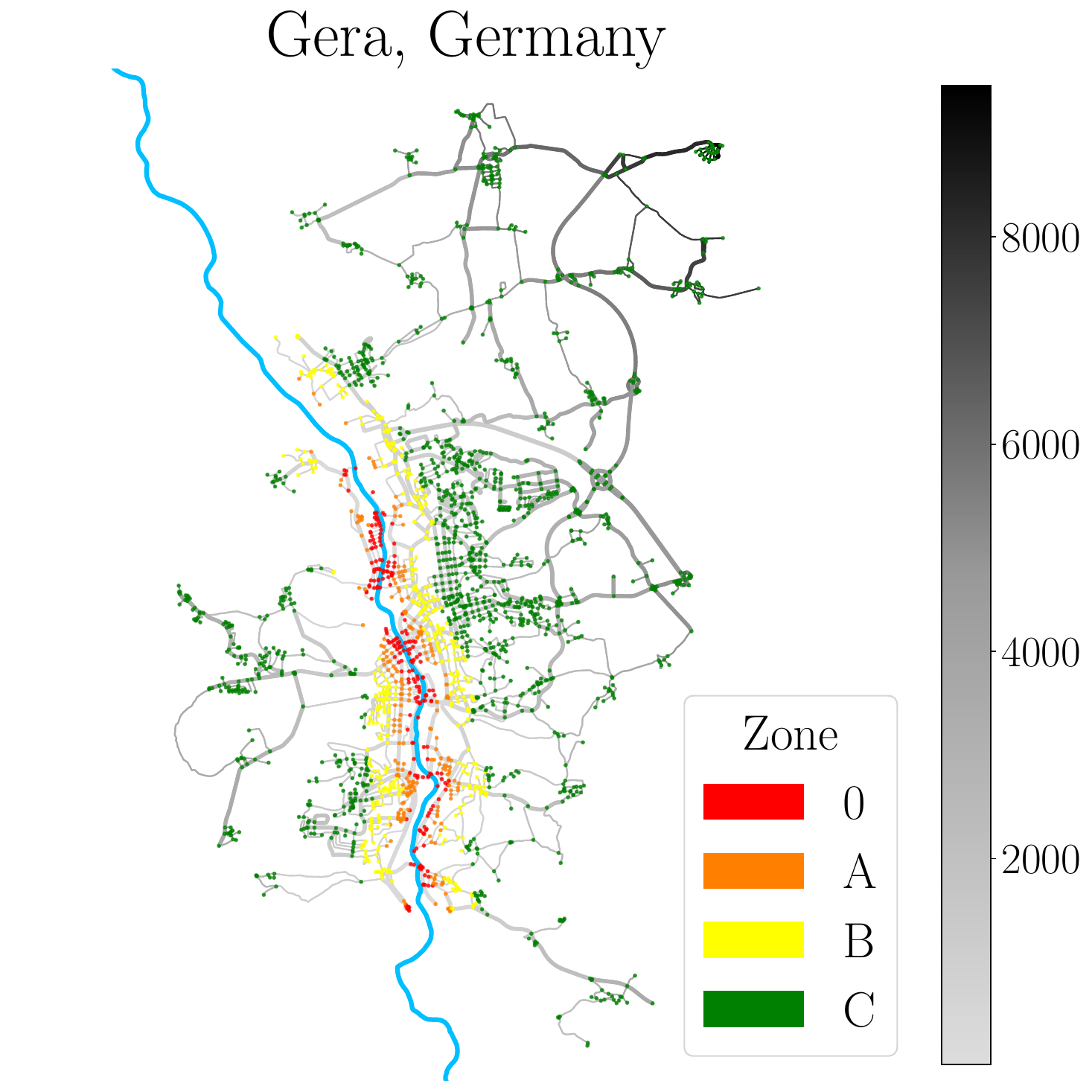}
 \hfill
 \includegraphics[width=0.32\textwidth]{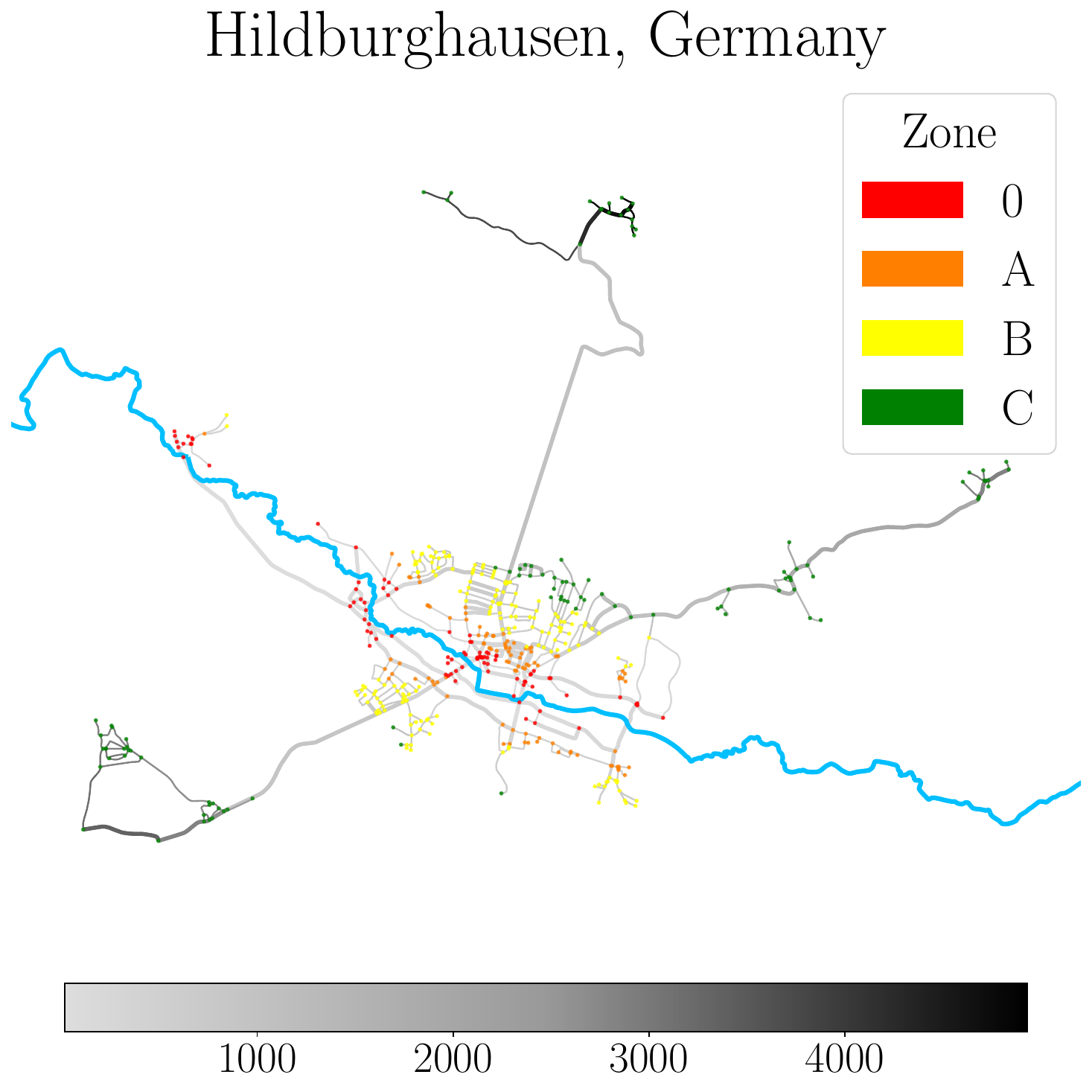}
 \hfill
 \includegraphics[width=0.32\textwidth]{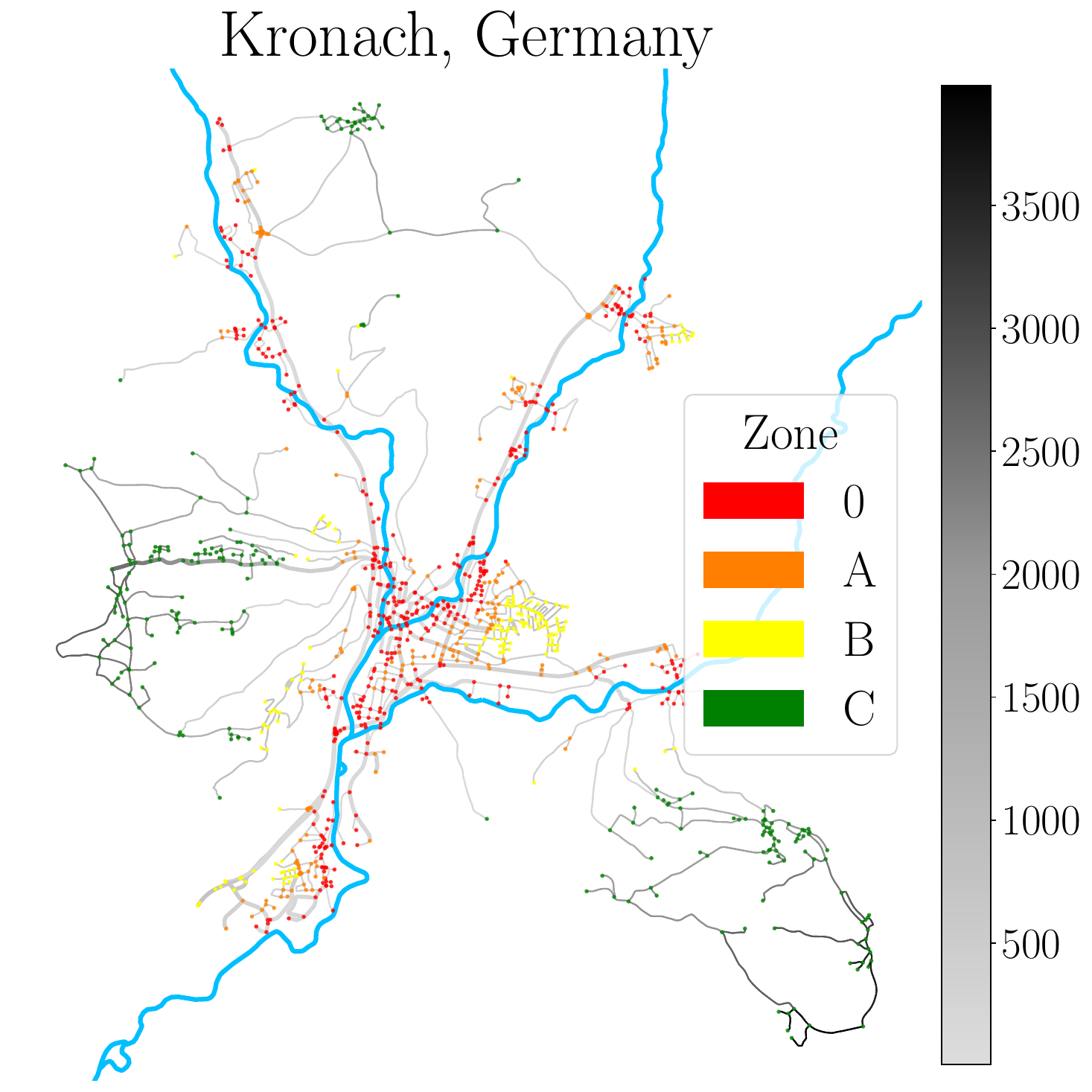}
 \hfill
 \includegraphics[width=0.32\textwidth]{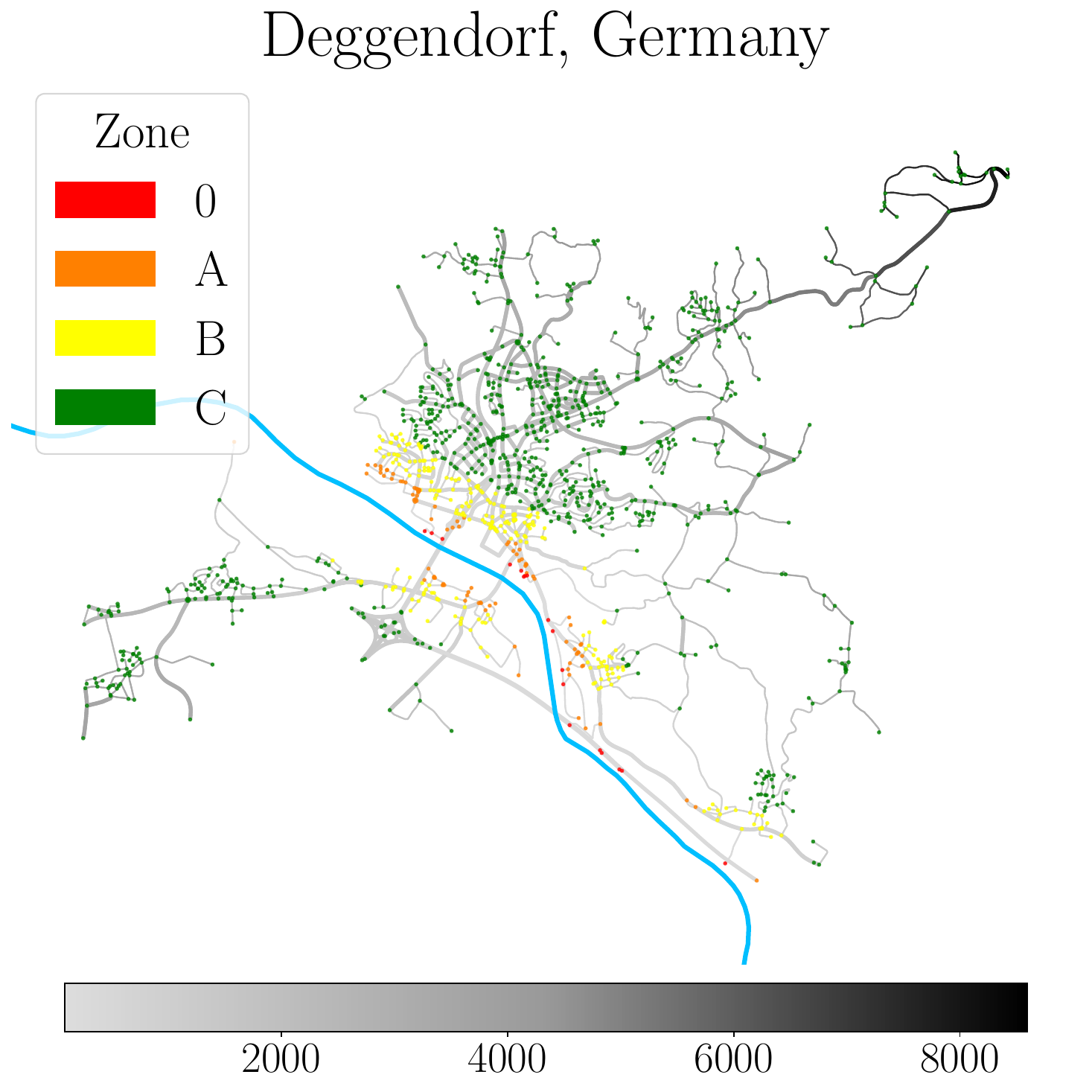}
 \hfill
 \includegraphics[width=0.32\textwidth]{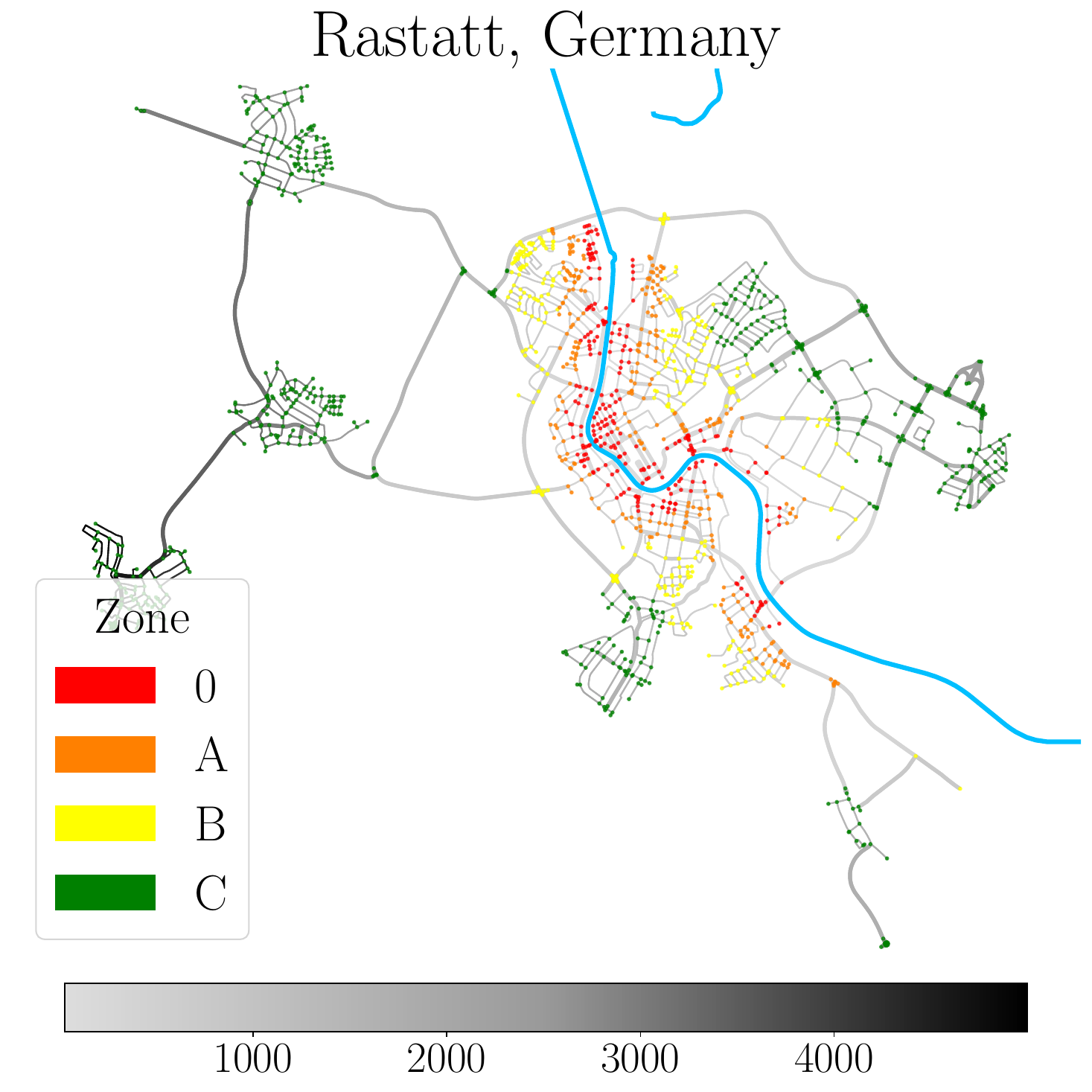}
 \caption{Overview on \osm{} (except for Koblenz, Germany; see \Cref{fig:osm_german_cities})
 where waterways are blue,
 and zones and deadlines are color-coded from red to green and in grayscale,
 respectively.}
 \label{fig:osm_german_cities_snd}
\end{figure*}

From OSM data,
we extracted the street network
and the main river(s) of each city
(assuming these to cause flooding).
For the network,
if two nodes are connected by a single lane traversable both-ways,
we made this connection an undirected edge.
For any connection,
we set its travel time (in seconds) to its length divided by the maximum allowed speed
of \SI{50}{\kilo\meter\per\hour};
Its deadline we set
to the smallest distance of its two endpoints
to the river(s)
divided by a flood speed of \SI{1}{\meter\per\second}
(i.e.,
the farther, the later it ceases).
For each vertex,
we set the capacity to the number of incident connections
(each serves as a waiting spot).
Moreover,
we assigned each vertex~$v$ to a zone,
determined by~$v$'s closest distance~$\dist(v)$ to the river(s):
vertex~$v$ is in zone 0 if~$\dist(v)\leq{}$\SI{250}{\meter},
in zone A if \SI{250}{\meter}${}<\dist(v)\leq{}$\SI{500}{\meter},
in zone B if \SI{500}{\meter}${}<\dist(v)\leq{}$\SI{1000}{\meter},
and in zone C if $\dist(v)>{}$\SI{1000}{\meter}.
Zone 0 will always be evacuated.
Let~$\zeta = (0,A,B,C)$.
The path set is constructed in the following way,
where $\frc{p}\in[0,1]$ determines the number of paths:
For a zone $\zeta_i\in \{A,B,C\}$ with~$i\geq 2$,
for~$\frc{p}\cdot |\bigcup_{j=1}^{i-1} \zeta_j|$
times
we pick a random source~$v$ from zones $\bigcup_{j=1}^{i-1} \zeta_j$,
then pick a random sink~$w$ from the 20\% of the vertices closest to~$v$ in zones $\bigcup_{j=i}^{4} \zeta_j$,
and add a shortest $(v,w)$-path.
We constructed 600 instances of the form~$I$\_$i$\_$\frc{p}$\_$z$-$x$
(referred to by~\osm),
where $x\in\{0,1,\dots,9\}$ and $i\in\{0,1,\dots,9\}$ corresponds to the index of the German city at hand,
for each tuple~$(\frc{p},z)$ in~$\{(0.1,X)\mid X\in\{A,B,C\}\}\cup \{(0.2,Y)\mid Y\in\{A,B\}\} \cup\{(0.3,A)\}$.

\subsection{Runtime Comparisons}

We analysed the mean over all constellations,
where for each constellation we either took the median
or the mean over all~$x\in\{0,1,\dots,9\}$.
For our results for \paths{} and for \stars{},
see \Cref{fig:exp:paths:medianandmean} and \Cref{fig:exp:starsandpaths:mean},
respectively.
\begin{figure}[t]
 \includegraphics[width=0.48\textwidth]{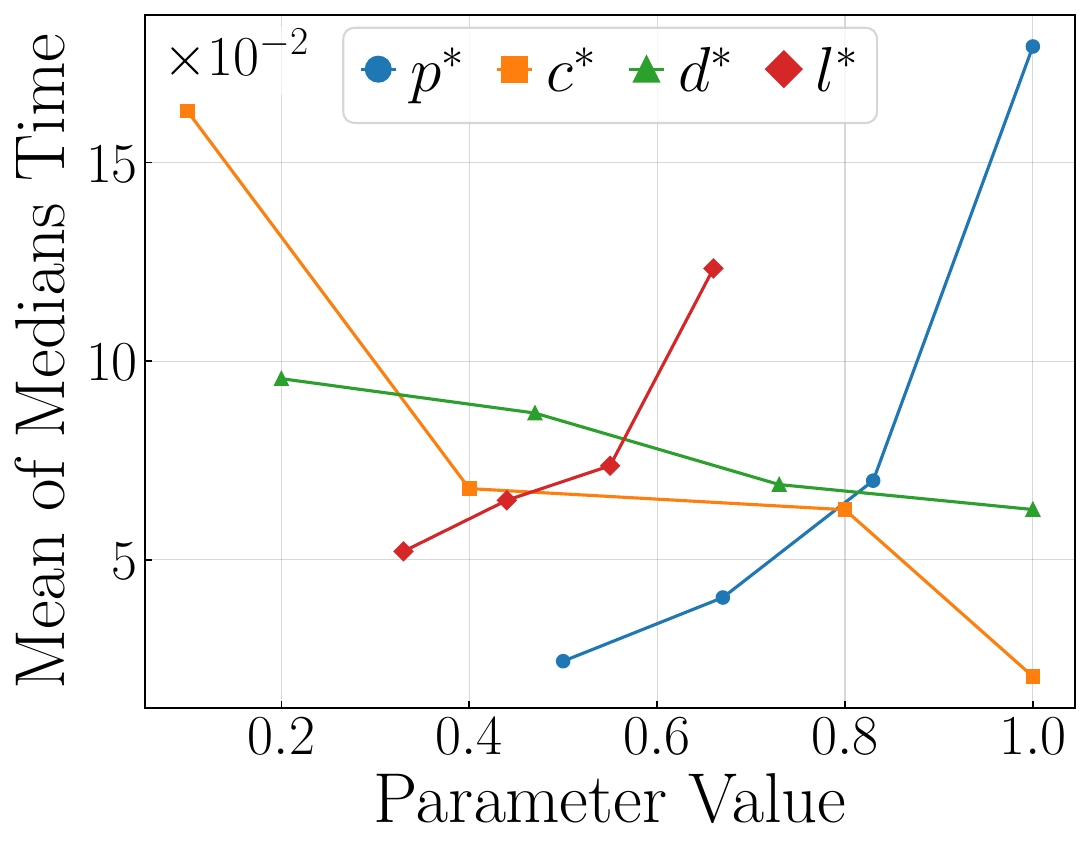}
 \hfill
 \includegraphics[width=0.48\textwidth]{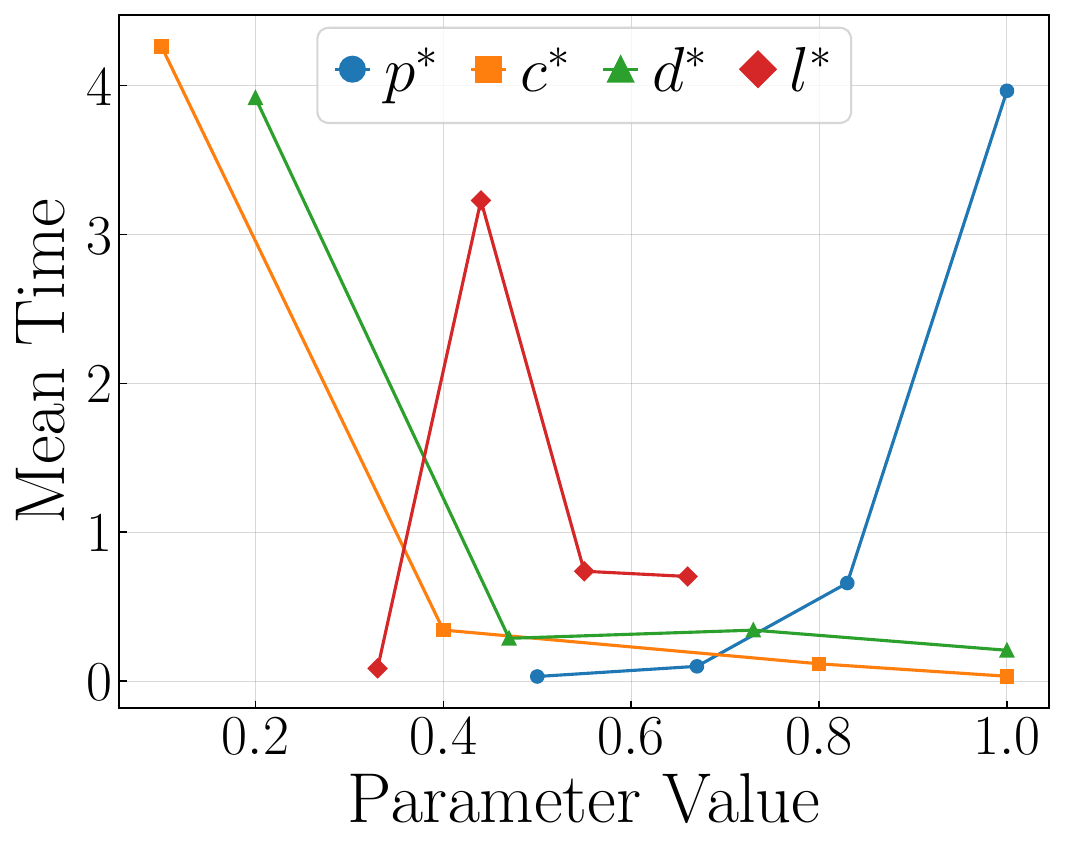}
 \caption{Mean of medians (left) and mean (right) over all runtimes [sec] for solving \paths{}
 when a parameter is fixed.}
 \label{fig:exp:paths:medianandmean}
\end{figure}
\begin{figure*}
 \includegraphics[width=0.495\textwidth]{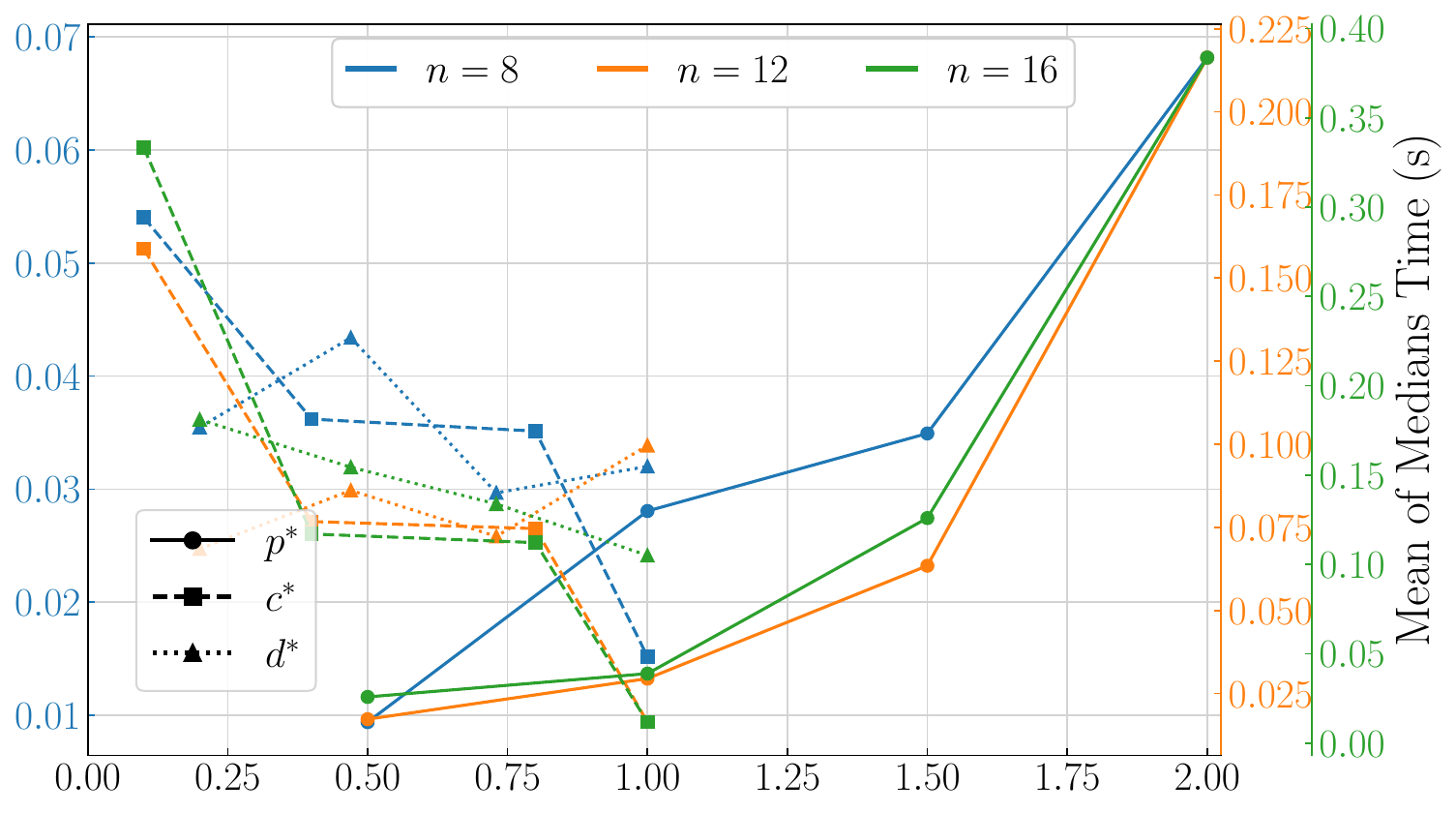}
 \hfill
 \includegraphics[width=0.495\textwidth]{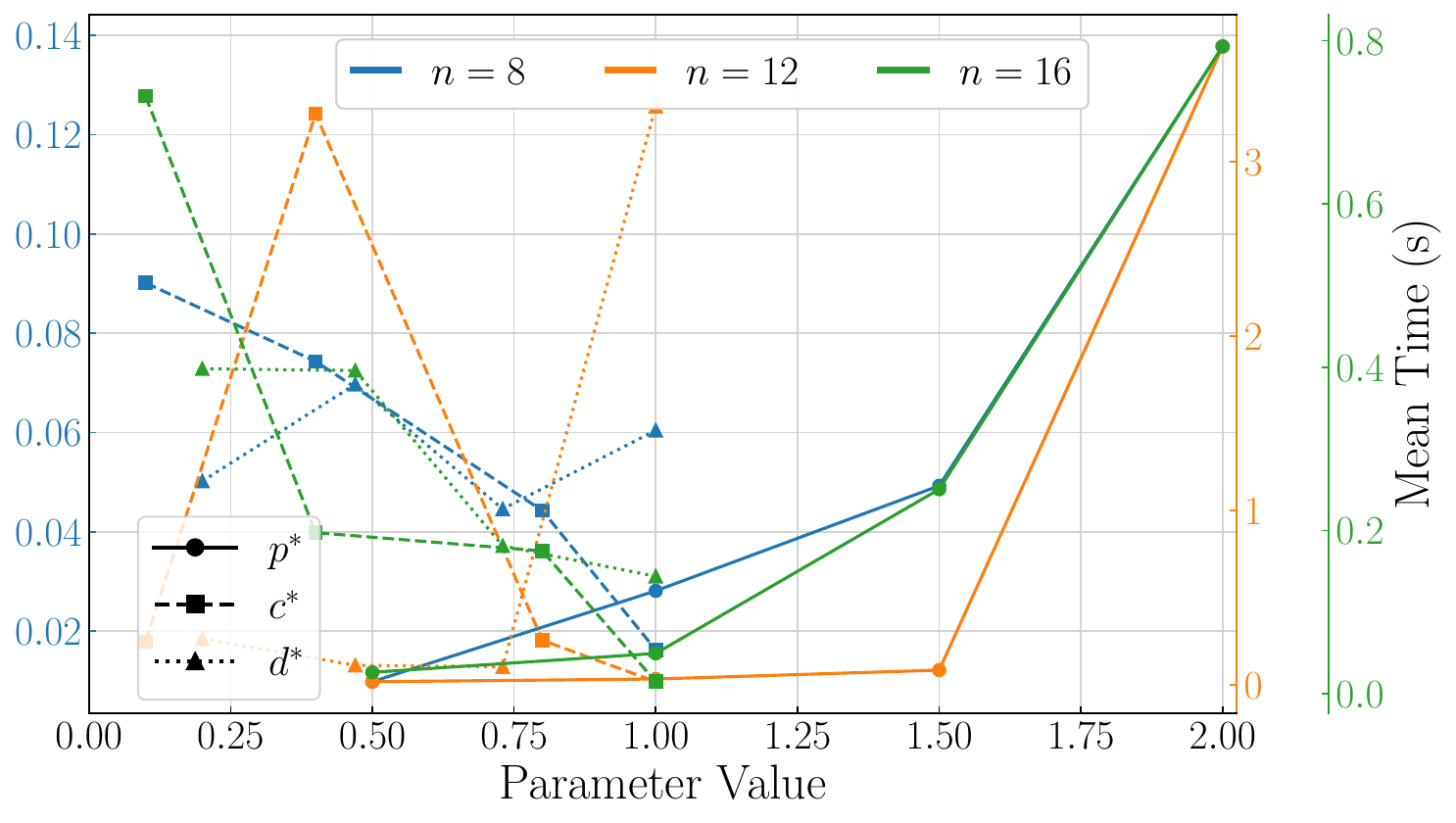}
 \hfill
 \includegraphics[width=0.495\textwidth]{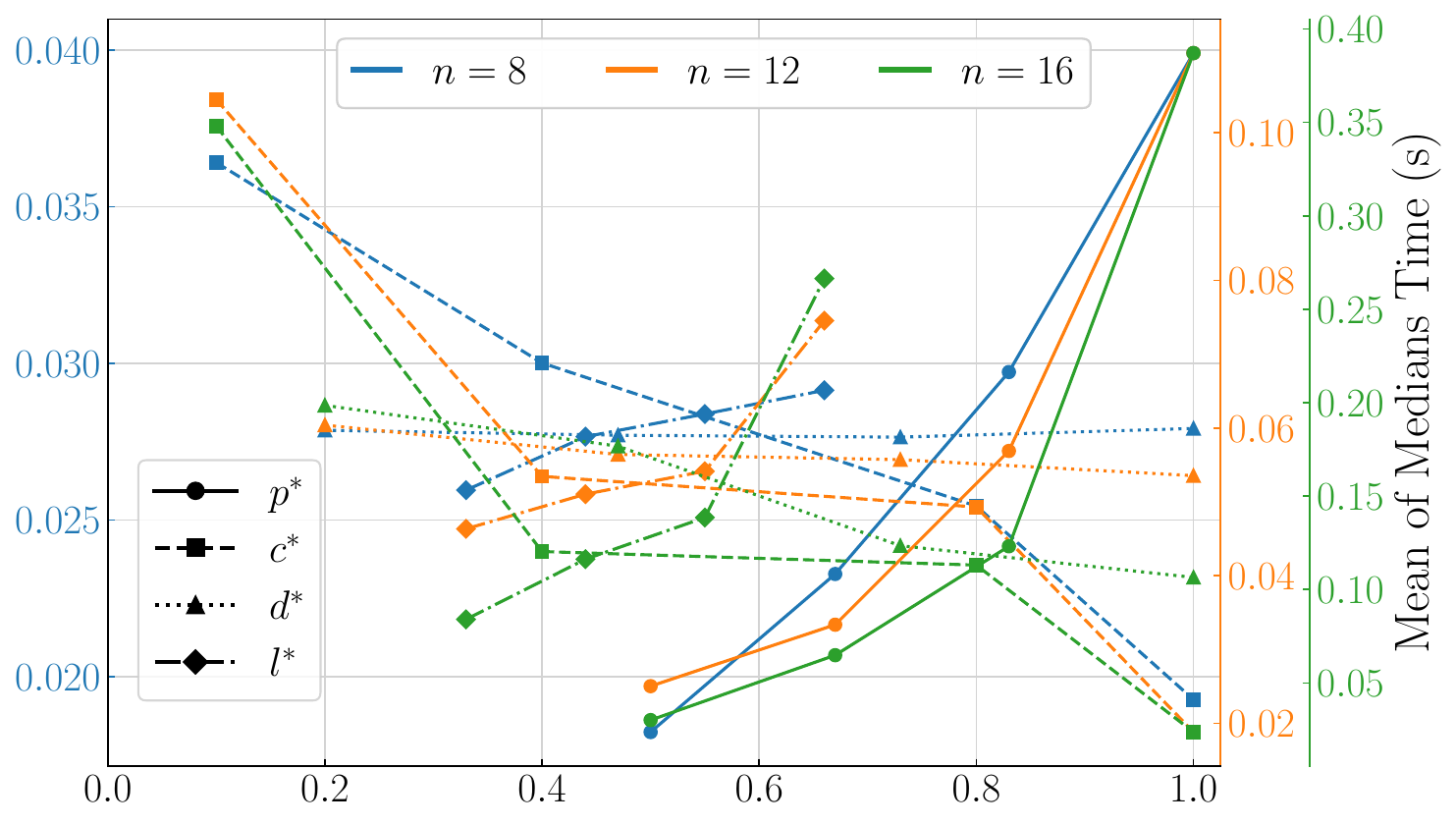}
 \includegraphics[width=0.495\textwidth]{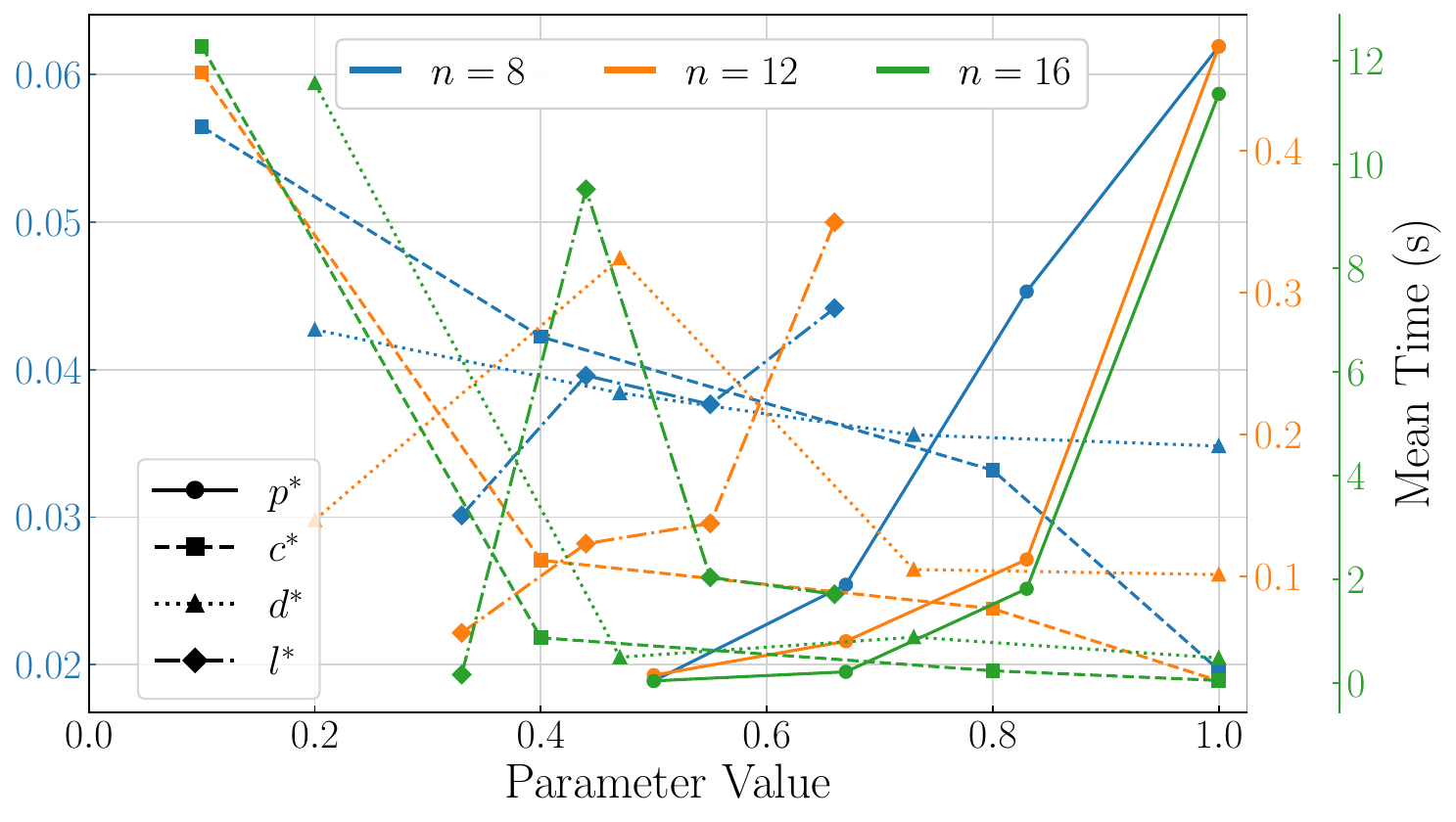}
 \caption{Mean of medians (left) and mean (right) over all runtimes
 when a parameter is fixed,
 separated for each~$n$,
 for (top) \stars{} and (bottom) \paths.}
 \label{fig:exp:starsandpaths:mean}
\end{figure*}
Considering the median,
our results indicate that runtime increases strongly with higher~$p$,
strongly with smaller~$c$,
and moderately with higher $\ell$ (for \paths).
Our results detect no correlation with~$d$.
The correlation with~$p$ and~$\ell$ can be explained since the number of variables
depends quadratically and linearly on~$p$ and~$\ell$,
respectively.
The correlation with~$c$ follows our intuition that~$d^*$ is monotonically increasing with decreasing~$c$.
Considering the mean
reveals severe outliers,
showing that the random assignments of the traversal times and paths
can have a severe impact on the ILP's performance.
In fact,
for some few instances the computation time is very large
(up to two hours).
Implementation \indc{} is faster than \bigM{} on~54.8\% of \stars{}
and on 62.7\% of \paths{}.

For our results on \osm{},
see~\Cref{tab:osm:instances:results}.
Our results show that our ILP can handle, e.g.,
zone-C evacuation with 10\% of the nodes serving as evacuation paths
in less than 40 seconds on average
(note that for,
e.g., Koblenz, these 10\% correspond to 176 evacuation paths).
Our results for zone-A and zone-B evacuation indicate
that increasing the number of paths severely impacts the running time
(for zone-C evacuation,
even for 15\% our machines did not terminate in reasonable time).
Implementation \indc{} performed faster than \bigM{} on 73.3\% of \osm{}.

\subsection{Performance of the ILP Relaxation}

The ILP relaxation performs surprisingly good for \osm{},
see \Cref{fig:scatter} and \Cref{tab:exp:relaxed:stats}:
\begin{figure}[t]
 \includegraphics[width=0.48\textwidth]{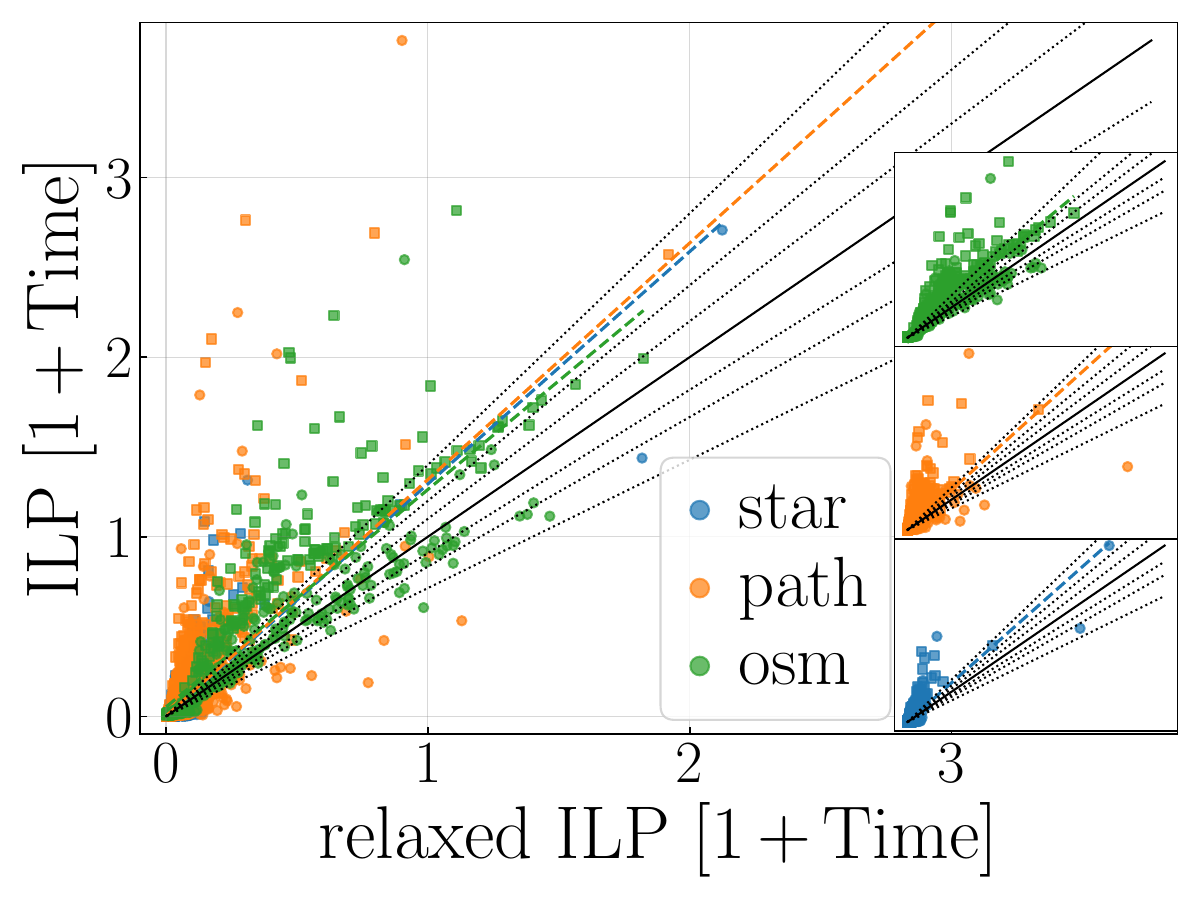}
 \hfill
 \includegraphics[width=0.48\textwidth]{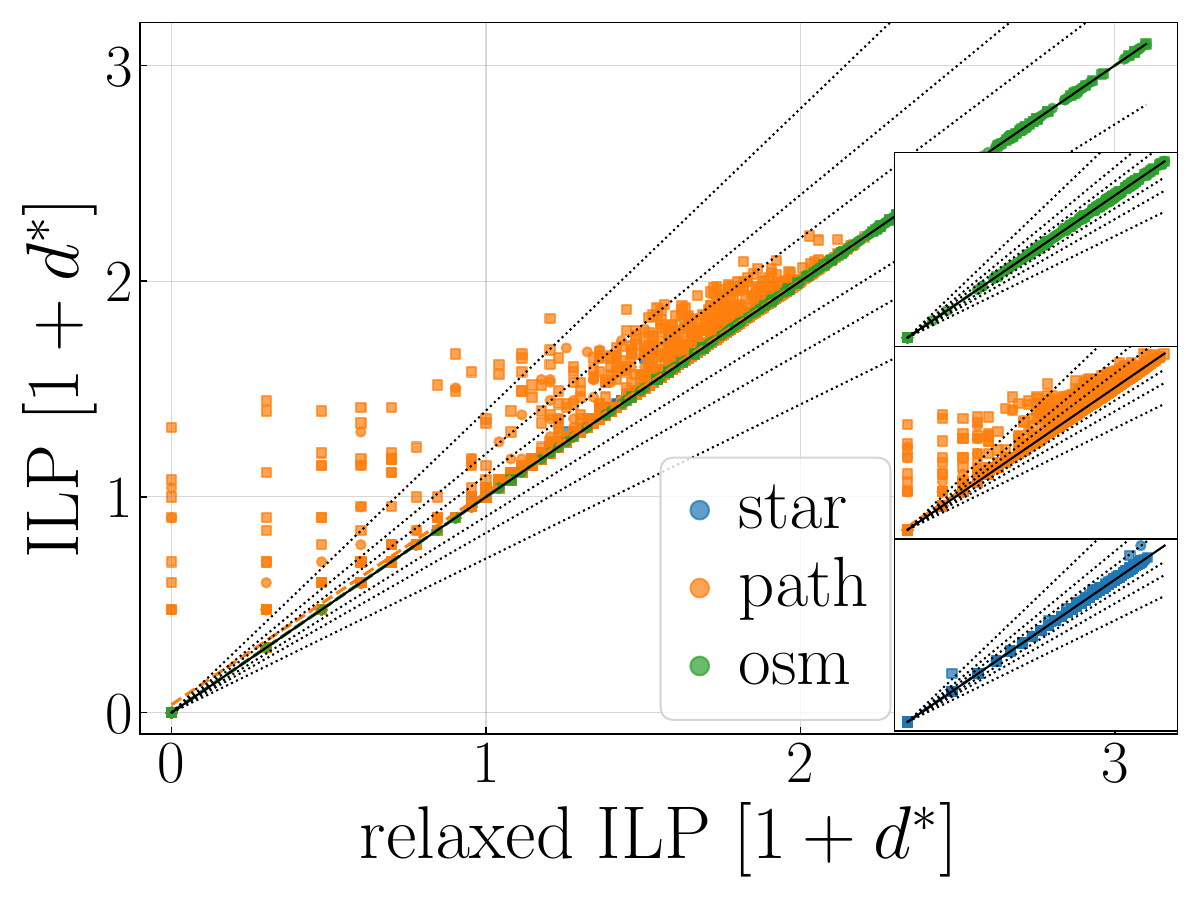}
 \caption{Scatter log-log plot for relaxed ILP versus ILP regarding runtime (left) and $d^*$ (right) on \stars, \paths, and \osm.}
 \label{fig:scatter}
\end{figure}
\begin{table*}[t]
 \centering
 \caption{Statistics of the ratios relaxed ILP over ILP regarding runtime and $d^*$.}
 \label{tab:exp:relaxed:stats}
 \begin{tabular}{rrrrrrr}\toprule
    & \multicolumn{3}{c}{\cellcolor{gray!10}time} & \multicolumn{3}{c}{\cellcolor{gray!20}$d^*$} \\
    & mean & std & median & mean & std & median \\
    \hline
    star & {\footnotesize [0.036, 25.390]} 2.297 & 3.840 & 1.018 & {\footnotesize [0.500, 1.000]} 0.998 & 0.017 & 1.000 \\
path & {\footnotesize [0.001, 77.933]} 1.560 & 1.749 & 1.060 & {\footnotesize [0.000, 1.000]} 0.977 & 0.101 & 1.000 \\
osm & {\footnotesize [0.018, 4.461]} 0.982 & 0.754 & 0.918 & {\footnotesize [1.000, 1.000]} 1.000 & 0.000 & 1.000 \\
\bottomrule\end{tabular}
\end{table*}
On average,
the ILP relaxation needs $98.2\%$ ($91.8\%$ in the median) of the running time of the best exact ILP
and yet computes on almost all instances the optimal $d^*$ from the exact ILP.
Also for \paths{} and \stars{} the solution quality is quite high with~$97.7\%$ and~$99.8\%$,
respectively.
Yet,
the standard deviation for \paths{} is significantly larger than for~\stars{}.
On both \paths{} and \stars{},
the average runtime of the ILP relaxation is even larger than the one of the best exact ILP,
which means that in many cases,
\indc{} performs fast.

\section{Epilogue}

We proved that smooth evacuation is a computationally hard task
even in well-structured decaying trees
under quite restrictive settings.
Given the complexity,
we formulated an ILP.
Our experiments demonstrate the practicality of our ILP.
Facing larger cities, larger path sets, or longer paths,
our recommendations are as follows:
(1)~Start both of our implementations and the relaxed ILP in parallel on separate machines.
(2)~Use powerful clusters (our machine already failed to handle instances slightly larger than our largest).

Our work opened several future pathways.
On the theoretical side,
we wonder about the computational complexity on
uncapacitated decaying paths and stars.
Further restrictions on the path set such as
only few vertices being eligible as sinks is also well-motivated.
On the experimental side,
we wish to handle larger instances.
For this,
reduction rules may be beneficial.
For both sides,
decaying caterpillars---generalizing both stars and paths---may be of interest,
in particular in terms of small deadlines.

{\begingroup
	\let\clearpage\relax
	\renewcommand{\url}[1]{\href{#1}{$\ExternalLink$}}
	\bibliography{simevac-bib}

\newcommand{\noopsort}[1]{}
\begin{thebibliography}{31}
\providecommand{\natexlab}[1]{#1}
\providecommand{\url}[1]{\texttt{#1}}
\expandafter\ifx\csname urlstyle\endcsname\relax
  \providecommand{\doi}[1]{doi: #1}\else
  \providecommand{\doi}{doi: \begingroup \urlstyle{rm}\Url}\fi

\bibitem[Akrida et~al.(2019)Akrida, Czyzowicz, Gasieniec, Kuszner, and
  Spirakis]{AkridaCGKS19}
Eleni~C. Akrida, Jurek Czyzowicz, Leszek Gasieniec, Lukasz Kuszner, and Paul~G.
  Spirakis.
\newblock Temporal flows in temporal networks.
\newblock \emph{J. Comput. Syst. Sci.}, 103:\penalty0 46--60, 2019.

\bibitem[Almagor et~al.(2024)Almagor, Kottinger, and Lahijanian]{AlmagorKL24}
Shaull Almagor, Justin Kottinger, and Morteza Lahijanian.
\newblock Temporal segmentation in multi agent path finding with applications
  to explainability.
\newblock \emph{Artif. Intell.}, 330:\penalty0 104087, 2024.

\bibitem[Atzmon et~al.(2020)Atzmon, Stern, Felner, Wagner, Bart{\'{a}}k, and
  Zhou]{AtzmonSFWBZ20}
Dor Atzmon, Roni Stern, Ariel Felner, Glenn Wagner, Roman Bart{\'{a}}k, and
  Neng{-}Fa Zhou.
\newblock Robust multi-agent path finding and executing.
\newblock \emph{J. Artif. Intell. Res.}, 67:\penalty0 549--579, 2020.

\bibitem[Boeckmann et~al.(2023)Boeckmann, Thielen, and Wittmann]{BoeckmannTW23}
Jan Boeckmann, Clemens Thielen, and Alina Wittmann.
\newblock Complexity of the temporal shortest path interdiction problem.
\newblock In \emph{Proceedings of the 2nd Symposium on Algorithmic Foundations
  of Dynamic Networks ({SAND} 2023)}, volume 257 of \emph{LIPIcs}, pages
  9:1--9:20. Schloss Dagstuhl - Leibniz-Zentrum f{\"{u}}r Informatik, 2023.

\bibitem[Casteigts et~al.(2021)Casteigts, Himmel, Molter, and
  Zschoche]{CasteigtsHMZ21}
Arnaud Casteigts, Anne{-}Sophie Himmel, Hendrik Molter, and Philipp Zschoche.
\newblock Finding temporal paths under waiting time constraints.
\newblock \emph{Algorithmica}, 83\penalty0 (9):\penalty0 2754--2802, 2021.

\bibitem[Chimani and Troost(2023)]{ChimaniT23}
Markus Chimani and Niklas Troost.
\newblock Multistage shortest path: Instances and practical evaluation.
\newblock In David Doty and Paul~G. Spirakis, editors, \emph{Proceedings of the
  2nd Symposium on Algorithmic Foundations of Dynamic Networks {SAND}}, volume
  257 of \emph{LIPIcs}, pages 14:1--14:19. Schloss Dagstuhl - Leibniz-Zentrum
  f{\"{u}}r Informatik, 2023.

\bibitem[Darmann and Döcker(2021)]{DarmannD20}
Andreas Darmann and Janosch Döcker.
\newblock On simplified {NP}-complete variants of {Monotone 3-Sat}.
\newblock \emph{Discrete Applied Mathematics}, 292:\penalty0 45--58, 2021.

\bibitem[Dhamala(2015)]{dhamala2015survey}
Tanka~Nath Dhamala.
\newblock A survey on models and algorithms for discrete evacuation planning
  network problems.
\newblock \emph{Journal of Industrial \& Management Optimization}, 11\penalty0
  (1):\penalty0 265, 2015.

\bibitem[Fluschnik et~al.(2019)Fluschnik, Morik, and Sorge]{FLUSCHNIK201969}
Till Fluschnik, Marco Morik, and Manuel Sorge.
\newblock The complexity of routing with collision avoidance.
\newblock \emph{Journal of Computer and System Sciences}, 102:\penalty0 69--86,
  2019.
\newblock ISSN 0022-0000.

\bibitem[Fluschnik et~al.(2023)Fluschnik, Niedermeier, Schubert, and
  Zschoche]{FluschnikNSZ23}
Till Fluschnik, Rolf Niedermeier, Carsten Schubert, and Philipp Zschoche.
\newblock Multistage s-t path: Confronting similarity with dissimilarity.
\newblock \emph{Algorithmica}, 85\penalty0 (7):\penalty0 2028--2064, 2023.

\bibitem[Gao et~al.(2024)Gao, Li, Li, Yan, Lin, and Wu]{GaoLLYLW24}
Jianqi Gao, Yanjie Li, Xinyi Li, Kejian Yan, Ke~Lin, and Xinyu Wu.
\newblock A review of graph-based multi-agent pathfinding solvers: From
  classical to beyond classical.
\newblock \emph{Knowl. Based Syst.}, 283:\penalty0 111121, 2024.

\bibitem[Garey et~al.(1976)Garey, Johnson, and Stockmeyer]{GAREY1976237}
M.R. Garey, D.S. Johnson, and L.~Stockmeyer.
\newblock Some simplified np-complete graph problems.
\newblock \emph{Theoretical Computer Science}, 1\penalty0 (3):\penalty0
  237--267, 1976.
\newblock ISSN 0304-3975.

\bibitem[{Gesamtverband der Deutschen Versicherungswirtschaft
  (GDV)}(2024)]{GDV2024_HochwasserAdressen}
{Gesamtverband der Deutschen Versicherungswirtschaft (GDV)}.
\newblock Amtliche zahlen zeigen: Mehr als 300 000 adressen in deutschland sind
  von hochwasser bedroht.
\newblock GDV Pressemitteilung, June 2024.
\newblock URL
  \url{https://gdv.de/gdv/medien/medieninformationen/amtliche-zahlen-zeigen-mehr-als-300000-adressen-in-deutschland-sind-von-hochwasser-bedroht-168828}.
\newblock Last accessed on August 29, 2025.

\bibitem[Gupta et~al.(1982)Gupta, Lee, and Leung]{GuptaLL82}
U.~I. Gupta, D.~T. Lee, and J.~Y.-T. Leung.
\newblock Efficient algorithms for interval graphs and circular-arc graphs.
\newblock \emph{Networks}, 12\penalty0 (4):\penalty0 459--467, 1982.

\bibitem[Higashikawa et~al.(2024)Higashikawa, Katoh, Teruyama, and
  Tokuni]{HigashikawaKTT24}
Yuya Higashikawa, Naoki Katoh, Junichi Teruyama, and Yuki Tokuni.
\newblock Faster algorithms for evacuation problems in networks with a single
  sink of small degree and bounded capacitated edges.
\newblock \emph{J. Comb. Optim.}, 48\penalty0 (3):\penalty0 18, 2024.

\bibitem[Karp(1972)]{Karp72}
Richard~M. Karp.
\newblock Reducibility among combinatorial problems.
\newblock In \emph{Proceedings of a symposium on the Complexity of Computer
  Computations, held March 20-22, 1972, at the {IBM} Thomas J. Watson Research
  Center, Yorktown Heights, New York, {USA}}, The {IBM} Research Symposia
  Series, pages 85--103. Plenum Press, New York, 1972.
\newblock URL \url{https://doi.org/10.1007/978-1-4684-2001-2\_9}.

\bibitem[Klobas et~al.(2023)Klobas, Mertzios, Molter, Niedermeier, and
  Zschoche]{KlobasMMNZ23}
Nina Klobas, George~B. Mertzios, Hendrik Molter, Rolf Niedermeier, and Philipp
  Zschoche.
\newblock Interference-free walks in time: temporally disjoint paths.
\newblock \emph{Auton. Agents Multi Agent Syst.}, 37\penalty0 (1):\penalty0 1,
  2023.

\bibitem[Kunz et~al.(2023)Kunz, Molter, and Zehavi]{KunzMZ23}
Pascal Kunz, Hendrik Molter, and Meirav Zehavi.
\newblock In which graph structures can we efficiently find temporally disjoint
  paths and walks?
\newblock In \emph{Proceedings of the Thirty-Second International Joint
  Conference on Artificial Intelligence, {IJCAI} 2023, 19th-25th August 2023,
  Macao, SAR, China}, pages 180--188. ijcai.org, 2023.

\bibitem[Litman(2006)]{Litman06}
Todd Litman.
\newblock Lessons from katrina and rita: What major disasters can teach
  transportation planners.
\newblock \emph{Journal of Transportation Engineering}, 132\penalty0
  (1):\penalty0 11--18, 2006.

\bibitem[Mishra et~al.(2018)Mishra, Mazumdar, and Pal]{MishraMP18}
Gopinath Mishra, Subhra Mazumdar, and Arindam Pal.
\newblock Improved algorithms for the evacuation route planning problem.
\newblock \emph{Journal of Combinatorial Optimization}, 36\penalty0
  (1):\penalty0 280--306, 2018.

\bibitem[Murano et~al.(2015)Murano, Perelli, and Rubin]{MuranoPR15}
Aniello Murano, Giuseppe Perelli, and Sasha Rubin.
\newblock Multi-agent path planning in known dynamic environments.
\newblock In \emph{Proceedings of the 18th International Conference on
  Principles and Practice of Multi-Agent Systems ({PRIMA} 2015)}, volume 9387
  of \emph{Lecture Notes in Computer Science}, pages 218--231. Springer, 2015.

\bibitem[Pertzovskiy et~al.(2025)Pertzovskiy, Stern, Zivan, and
  Felner]{PertzovskiySZF25}
Arseni Pertzovskiy, Roni Stern, Roie Zivan, and Ariel Felner.
\newblock Multi-agent corridor generating algorithm.
\newblock In \emph{Proceedings of the Thirty-Fourth International Joint
  Conference on Artificial Intelligence ({IJCAI} 2025)}, pages 240--247.
  ijcai.org, 2025.

\bibitem[Pyakurel et~al.(2019)Pyakurel, Nath, Dempe, and
  Dhamala]{PyakurelNDD19}
Urmila Pyakurel, Hari~Nandan Nath, Stephan Dempe, and Tanka~Nath Dhamala.
\newblock Efficient dynamic flow algorithms for evacuation planning problems
  with partial lane reversal.
\newblock \emph{Mathematics}, 7\penalty0 (10), 2019.

\bibitem[Schichl and Sellmann(2015)]{SchichlS15}
Hermann Schichl and Meinolf Sellmann.
\newblock Predisaster preparation of transportation networks.
\newblock In \emph{Proceedings of the 29th {AAAI} Conference on Artificial
  Intelligence ({AAAI}'15)}, pages 709--715. {AAAI} Press, 2015.

\bibitem[Shahabi and Wilson(2014)]{ShahabiW14}
Kaveh Shahabi and John~P. Wilson.
\newblock {CASPER:} intelligent capacity-aware evacuation routing.
\newblock \emph{Comput. Environ. Urban Syst.}, 46:\penalty0 12--24, 2014.

\bibitem[Stern et~al.(2019)Stern, Sturtevant, Felner, Koenig, Ma, Walker, Li,
  Atzmon, Cohen, Kumar, Bart{\'{a}}k, and Boyarski]{SternSFK0WLA0KB19}
Roni Stern, Nathan~R. Sturtevant, Ariel Felner, Sven Koenig, Hang Ma, Thayne~T.
  Walker, Jiaoyang Li, Dor Atzmon, Liron Cohen, T.~K.~Satish Kumar, Roman
  Bart{\'{a}}k, and Eli Boyarski.
\newblock Multi-agent pathfinding: Definitions, variants, and benchmarks.
\newblock In Pavel Surynek and William Yeoh, editors, \emph{Proceedings of the
  12th International Symposium on Combinatorial Search ({SOCS} 2019)}, pages
  151--158. {AAAI} Press, 2019.

\bibitem[van Bevern et~al.(2015)van Bevern, Mnich, Niedermeier, and
  Weller]{BevernMNW15}
Ren{\'{e}} van Bevern, Matthias Mnich, Rolf Niedermeier, and Mathias Weller.
\newblock Interval scheduling and colorful independent sets.
\newblock \emph{J. Sched.}, 18\penalty0 (5):\penalty0 449--469, 2015.

\bibitem[Wu et~al.(2016)Wu, Sheldon, and Zilberstein]{WuSZ16}
XiaoJian Wu, Daniel Sheldon, and Shlomo Zilberstein.
\newblock Optimizing resilience in large scale networks.
\newblock In \emph{Proceedings of the 30th {AAAI} Conference on Artificial
  Intelligence ({AAAI}'16)}, pages 3922--3928. {AAAI} Press, 2016.

\bibitem[Yu and LaValle(2012)]{YuL12}
Jingjin Yu and Steven~M. LaValle.
\newblock Multi-agent path planning and network flow.
\newblock In \emph{Proceedings of the Tenth Workshop on the Algorithmic
  Foundations of Robotics ({WAFR} 2012)}, volume~86 of \emph{Springer Tracts in
  Advanced Robotics}, pages 157--173. Springer, 2012.

\bibitem[Yu and LaValle(2013)]{YuL13}
Jingjin Yu and Steven~M. LaValle.
\newblock Structure and intractability of optimal multi-robot path planning on
  graphs.
\newblock In \emph{Proceedings of the Twenty-Seventh {AAAI} Conference on
  Artificial Intelligence ({AAAI} 2013)}, pages 1443--1449. {AAAI} Press, 2013.

\bibitem[Yu and LaValle(2016)]{YuL16}
Jingjin Yu and Steven~M. LaValle.
\newblock Optimal multirobot path planning on graphs: Complete algorithms and
  effective heuristics.
\newblock \emph{{IEEE} Trans. Robotics}, 32\penalty0 (5):\penalty0 1163--1177,
  2016.

\end{thebibliography}
\endgroup}

\begin{landscape}
	\begin{table*}[t]
	\centering
	\caption{Overview on \osm{}.
	For the connections we give the percentages of one-way arcs (O), antiparallel arcs (T), and undirected edges (U). diam denotes the diameter of~$G$. For~$\trt$, $d$, and~$c$, we give the mean along with $[\text{min},\,\text{max}]$.}
	\label{tab:osm:instances}
	\begin{tabular}{l r l r l r r r r r}
\toprule
City & \multicolumn{2}{l}{\# Vertices (in zones 0, A, B)} & \multicolumn{2}{l}{\# Connections (O, T, U)} & diam & $\trt$ & $d$ & $c$ \\
\midrule
Ahrweiler & 832 & (267, 275, 179) & 1066 & (17.8\%, 15.4\%, 66.8\%) & 68 & {\footnotesize [1.0, 323.0]} 13.1 & {\footnotesize [17.0, 4841.0]} 675.8 & {\footnotesize [1.0, 8.0]} 4.7 \\
Deggendorf & 1196 & (18, 93, 235) & 1483 & (12.5\%, 29.8\%, 57.7\%) & 61 & {\footnotesize [1.0, 400.0]} 15.1 & {\footnotesize [84.0, 8596.0]} 1893.3 & {\footnotesize [1.0, 10.0]} 4.7 \\
Gera & 2487 & (223, 243, 485) & 3197 & (17.8\%, 21.0\%, 61.2\%) & 88 & {\footnotesize [1.0, 527.0]} 14.2 & {\footnotesize [13.0, 9465.0]} 1958.2 & {\footnotesize [2.0, 10.0]} 4.7 \\
Hildburghausen & 454 & (88, 94, 168) & 578 & (10.7\%, 20.4\%, 68.9\%) & 46 & {\footnotesize [1.0, 345.0]} 14.4 & {\footnotesize [17.0, 4931.0]} 875.0 & {\footnotesize [2.0, 8.0]} 4.9 \\
Kitzingen & 707 & (83, 85, 188) & 904 & (14.3\%, 16.8\%, 68.9\%) & 56 & {\footnotesize [1.0, 117.0]} 12.4 & {\footnotesize [50.0, 3576.0]} 1051.7 & {\footnotesize [2.0, 8.0]} 4.8 \\
Koblenz & 2961 & (278, 621, 864) & 3886 & (35.5\%, 14.4\%, 50.2\%) & 100 & {\footnotesize [1.0, 463.0]} 11.5 & {\footnotesize [22.0, 4658.0]} 1035.1 & {\footnotesize [1.0, 8.0]} 4.4 \\
Kronach & 1024 & (408, 215, 166) & 1207 & (7.0\%, 13.7\%, 79.3\%) & 66 & {\footnotesize [1.0, 212.0]} 13.5 & {\footnotesize [6.0, 3989.0]} 658.4 & {\footnotesize [1.0, 10.0]} 4.6 \\
Rastatt & 1447 & (247, 267, 298) & 1999 & (18.9\%, 11.7\%, 69.4\%) & 75 & {\footnotesize [1.0, 115.0]} 10.8 & {\footnotesize [28.0, 4997.0]} 1214.2 & {\footnotesize [2.0, 10.0]} 5.1 \\
Trier & 2777 & (252, 331, 719) & 3680 & (31.0\%, 19.7\%, 49.3\%) & 124 & {\footnotesize [1.0, 254.0]} 11.6 & {\footnotesize [63.0, 6201.0]} 1590.9 & {\footnotesize [1.0, 8.0]} 4.5 \\
Zell (Mosel) & 246 & (135, 47, 58) & 334 & (20.4\%, 10.2\%, 69.5\%) & 48 & {\footnotesize [1.0, 349.0]} 13.7 & {\footnotesize [82.0, 1070.0]} 322.0 & {\footnotesize [2.0, 8.0]} 4.9 \\
\bottomrule
\end{tabular}

	\end{table*}
	\begin{table*}[t]
	\centering
	\caption{Average runtime [sec] for \osm{} along with $[\text{min},\,\text{max}]$ runtimes.}
	\label{tab:osm:instances:results}
	\begin{tabular}{lrrrrrr}\toprule
City & \cellcolor{gray!10}(0.1, A) & \cellcolor{gray!10}(0.2, A) & \cellcolor{gray!10}(0.3, A) & \cellcolor{gray!20}(0.1, B) & \cellcolor{gray!20}(0.2, B) & \cellcolor{gray!10}(0.1, C) \\ \midrule
Ahrweiler & {\footnotesize [0.03, 0.20]} 0.13 & {\footnotesize [0.23, 2.42]} 0.62 & {\footnotesize [0.77, 7.83]} 2.67 & {\footnotesize [0.43, 3.48]} 1.11 & {\footnotesize [1.82, 24.65]} 6.09 & {\footnotesize [4.30, 34.94]} 12.06 \\ 
Deggendorf & {\footnotesize [0.00, 0.04]} 0.02 & {\footnotesize [0.01, 0.04]} 0.02 & {\footnotesize [0.02, 0.04]} 0.02 & {\footnotesize [0.03, 0.09]} 0.05 & {\footnotesize [0.04, 0.18]} 0.12 & {\footnotesize [0.31, 1.57]} 0.59 \\ 
Gera & {\footnotesize [0.06, 0.18]} 0.13 & {\footnotesize [0.31, 1.32]} 0.59 & {\footnotesize [0.76, 6.82]} 2.34 & {\footnotesize [0.33, 0.88]} 0.49 & {\footnotesize [1.80, 10.06]} 3.97 & {\footnotesize [1.39, 7.84]} 2.97 \\ 
Hildburghausen & {\footnotesize [0.01, 0.04]} 0.02 & {\footnotesize [0.02, 0.10]} 0.07 & {\footnotesize [0.08, 0.32]} 0.17 & {\footnotesize [0.02, 0.10]} 0.05 & {\footnotesize [0.14, 0.43]} 0.23 & {\footnotesize [0.29, 5.64]} 1.10 \\ 
Kitzingen & {\footnotesize [0.02, 0.05]} 0.03 & {\footnotesize [0.02, 0.09]} 0.06 & {\footnotesize [0.03, 0.19]} 0.11 & {\footnotesize [0.02, 0.06]} 0.04 & {\footnotesize [0.03, 0.27]} 0.16 & {\footnotesize [0.22, 1.61]} 0.42 \\ 
Koblenz & {\footnotesize [0.07, 0.49]} 0.26 & {\footnotesize [0.79, 7.33]} 2.29 & {\footnotesize [1.59, 7.56]} 3.82 & {\footnotesize [1.15, 7.08]} 2.32 & {\footnotesize [7.66, 40.80]} 18.02 & {\footnotesize [12.32, 97.08]} 35.87 \\ 
Kronach & {\footnotesize [0.05, 0.36]} 0.22 & {\footnotesize [0.83, 3.31]} 1.57 & {\footnotesize [2.01, 19.27]} 4.08 & {\footnotesize [0.71, 7.29]} 2.80 & {\footnotesize [4.16, 23.18]} 8.33 & {\footnotesize [6.12, 652.81]} 46.06 \\ 
Rastatt & {\footnotesize [0.04, 0.13]} 0.08 & {\footnotesize [0.06, 0.43]} 0.27 & {\footnotesize [0.53, 1.25]} 0.72 & {\footnotesize [0.09, 1.22]} 0.45 & {\footnotesize [1.38, 12.36]} 4.50 & {\footnotesize [2.61, 8.94]} 4.43 \\ 
Trier & {\footnotesize [0.17, 0.88]} 0.31 & {\footnotesize [0.54, 3.21]} 1.01 & {\footnotesize [1.09, 10.71]} 3.63 & {\footnotesize [0.68, 2.31]} 1.10 & {\footnotesize [3.49, 18.73]} 7.15 & {\footnotesize [5.40, 29.56]} 10.92 \\ 
Zell (Mosel) & {\footnotesize [0.06, 1.11]} 0.33 & {\footnotesize [0.41, 11.07]} 2.32 & {\footnotesize [2.04, 170.05]} 18.49 & {\footnotesize [0.33, 98.06]} 7.62 & {\footnotesize [3.33, 347.86]} 25.05 & {\footnotesize [0.30, 105.16]} 9.65 \\ 
\bottomrule\end{tabular}
	\end{table*}
\end{landscape}

\end{document}